\documentclass[11pt]{article}
\usepackage{jheppub}

\usepackage[scr]{rsfso}
\usepackage{amsmath, amsthm, amscd,amssymb, amsfonts, verbatim,subfigure, enumerate}
\usepackage[mathcal]{eucal}
\usepackage[super]{nth}
\usepackage{xcolor}
\usepackage{tikz}
\usetikzlibrary{arrows,decorations.markings,shapes.arrows,patterns, decorations.pathmorphing}
\tikzset{
  ->-/.style={decoration={markings, mark=at position 0.5 with {\arrow{to}}},
              postaction={decorate}},
}
\tikzset{-<-/.style={decoration={markings,mark=at position 0.5 with %
    {\arrow[scale=1.5,>=stealth]{<}}},postaction={decorate}}}

\newcommand{\half}{\tfrac{1}{2}}

\renewcommand{\i}{\mathrm{i}}
\newcommand{\CP}{\mathbb{CP}}

\newtheorem{proposition}{Proposition}
\newtheorem{theorem}{Theorem}
\newtheorem{conjecture}{Conjecture}

\newcommand{\tr}{\triangle}

\newcommand{\iso}{\cong}

\newcommand{\dbar}{\br{\partial}}

\renewcommand{\c}{\mathsf{c}}

\newcommand{\zbar}{\br{z}}

\newcommand{\eps}{\epsilon}
\newcommand{\g}{\mathfrak{g}}

\newcommand{\til}{\widetilde}
\newcommand{\mscr}{\mathscr}
\renewcommand{\det}{\operatorname{det}}

\newcommand{\br}{\overline}

\newcommand{\C}{\mathbb C}

\newcommand{\norm}[1]{\left\| #1 \right\|}
\newcommand{\Oo}{\mscr{O}}
\newcommand{\Z}{\mathbb{Z}}

\newcommand{\into}{\hookrightarrow}

\newcommand{\op}{\operatorname}
\newcommand{\mbf}{\mathbf}
\newcommand{\mbb}{\mathbb}
\newcommand{\mf}{\mathfrak}
\newcommand{\mc}{\mathcal}

\newcommand{\ip}[1]{\left\langle #1 \right\rangle}
\newcommand{\abs}[1]{\left| #1 \right|}

\newcommand{\R}{\mbb R}
\renewcommand{\d}{\mathrm{d}}

\newcommand{\sh}{{\sf{h}}}    
\renewcommand{\t}{{\sf{t}}}

\DeclareMathOperator{\Sym}{Sym} \DeclareMathOperator{\Hom}{Hom}

\newcommand{\ChS}{\mathrm{CS}}

\title{Q-operators are 't Hooft lines}
\author[a]{Kevin Costello,}
\author[a]{Davide Gaiotto}
\author[b]{and Junya Yagi}

\affiliation[a]{Perimeter Institute for Theoretical Physics, Waterloo, ON
  N2L 2Y5 Canada}

\affiliation[b]{Yau Mathematical Sciences Center, Tsinghua University,
  Beijing 100084 China}

\abstract{We study 't Hooft lines in four-dimensional
  holomorphic-topological Chern-Simons theory. We relate them to
  Q-operators in the theory of integrable systems. We give a physical
  interpretation of the fundamental TQ and QQ relations satisfied by
  Q-operators and conventional transfer matrices.}

\begin{document}
\maketitle

\section{Introduction}

Four-dimensional Chern-Simons theory \cite{Costello:2013zra,
  Costello:2017dso, Costello:2018gyb} is a unified approach to
studying integrability which has been able to explain many examples of
integrable systems: spin chains \cite{Costello:2013zra,
  Costello:2017dso}, integrable field theories \cite{Costello:2019tri,
  Vicedo:2019dej, Delduc:2019whp, Fukushima:2020dcp, Bykov:2020nal},
spin chains with boundary \cite{Bittleston:2019gkq,
  Bittleston:2019mol}, and integrable string
theories~\cite{Costello:2020lpi}.

In this paper we generalize the analysis of \cite{Costello:2017dso,
  Costello:2018gyb} relating integrable spin chains and
four-dimensional Chern-Simons theory to include an ingredient
previously missing: the Q-operator.  Before we state the problem we
solve, let us review how four-dimensional Chern-Simons theory is
related to integrable spin chains.

Consider the four-manifold $\R^2 \times \C$ with coordinates $x,y$ on
$\R^2$ and holomorphic coordinate $z$ on $\C$.  The fundamental field
of four-dimensional Chern-Simons theory is a gauge-field $A$ with
three components: $A_{\zbar}$, $A_x$, $A_y$. The Lagrangian is
\begin{equation} 
  \int \ChS(A) \d z ,
\end{equation}
where $\ChS(A)$ is the Chern-Simons three-form.  The gauge field $A$ is
required to tend to $0$ as $z \to \infty$.

The XXX spin chain systems are generated in this setup as follows.
Take the gauge group to be $SL_2(\C)$, and make one of the two
topological directions, say $y$, periodic.  Insert $N$ fundamental
Wilson lines wrapping the $x$-direction.  We can take all the
fundamental Wilson lines to live at $z = 0$.  The Hilbert space of the
spin system is $(\C^2)^{\otimes N}$, which is the tensor product of
the spaces of states living at the end of each Wilson
line.\footnote{For this analysis, it is important to use the boundary
  condition that $A \to 0$ as $z \to \infty$. This breaks the gauge
  symmetry and makes the effective two-dimensional theory on the
  $x$-$y$ plane obtained by integrating out the gauge field massive.}

The transfer matrix $T(z)$ of the spin chain is realized by also
placing a fundamental Wilson line wrapping the (periodic)
$y$-direction, at some point $z$.

An important ingredient in the theory of integrable models was missing
from the analysis of \cite{Costello:2017dso, Costello:2018gyb}:
Baxter's Q-operator \cite{Baxter:1972hz}, first introduced in his
seminal analysis of the $8$-vertex model. The Q-operators
$\mbf{Q}_{\pm}(z)$ are the two linearly independent solutions to
Baxter's TQ relation
\begin{equation} 
  \label{eqn_tq_intro}
  \mbf{T}(z,\phi) \mbf{Q}_{\pm}(z,\phi) = (z - \tfrac{1}{2}\hbar)^N \mbf{Q}_{\pm}(z +\hbar,\phi) +  (z + \tfrac{1}{2}\hbar)^N \mbf{Q}_{\pm}(z -\hbar,\phi).
\end{equation}
The Q-operators also commute with the transfer matrix:
\begin{equation} 
	[\mbf{Q}_{\pm}(z), \mbf{T}(z')] = 0. 
\end{equation}
The Q-operator, and the relations between the T- and Q-operators, are
essential for analyzing the spectrum of an integrable model, and
indeed one can derive the Bethe ansatz from the relations between the
T- and Q-operators (see \cite{Bazhanov:2010ts} for a discussion).

In this paper we will analyze another class of topological line
defects in four-dimensional Chern-Simons theory: 't Hooft lines. We
will find that 't Hooft lines produce naturally Q-operators.

Our main results are the following.
\begin{enumerate} 
\item We show that the Q-operator for $SL_2$, up to an important
  normalizing factor, is the fractional 't Hooft line of charge $1/2$
  for $SL_2$ (or of charge $1$ for $PSL_2$).  More generally, we show
  that Q-operators for $SL_n$ \cite{Bazhanov:2010jq} and $SO(n)$
  \cite{Frassek:2020nki} arise from 't Hooft lines whose charge is a
  minuscule coweight of the adjoint form of the group.  We do this by
  reproducing, from an analysis of 't Hooft lines, the ``oscillator''
  construction of Q-operators presented in these works. We also
  generalize these constructions to include the minuscule coweights of
  $E_6$, $E_7$ (which also give rise to oscillator representations).
  
\item We derive the TQ and QQ relations from a first-principles
  analysis of 't Hooft lines. The TQ relation follows from the
  Witten effect \cite{Witten:1979ey} in four-dimensional gauge theory.

\item We analyze 't Hooft-Wilson lines whose charge is an arbitrary
  coweight of an arbitrary simple group. We argue that these 't
  Hooft-Wilson lines are classified by representations of a certain
  \emph{shifted Yangian} \cite{Brundan:2004ca, MR3248988,
    Frassek:2020lky}, just as in \cite{Costello:2017dso,
    Costello:2018gyb} we found that Wilson lines are classified by
  representations of the Yangian.  We show how to construct these line
  defects by using the quantum Coulomb branches of three-dimensional
  $\mathcal{N}=4$ quiver gauge theories \cite{Braverman:2016pwk},
  which are linked to four-dimensional Chern-Simons theory by string
  dualities \cite{Costello:2018txb}. In a sense, this gives a string
  theory motivation for the known relation between Coulomb branches
  and shifted Yangians \cite{Braverman:2016pwk}.
		
\item We derive new QQ relations valid for the classical groups and
  the exceptional groups $E_6$, $E_7$. In each case, the Q-operator
  is the 't Hooft line labelled by a minuscule coweight, and the QQ
  relation expresses the product of two Q-operators of opposite
  weight as the Wilson lines associated to a parabolic Verma module.

\item We build Q-operators for some of the class of integrable field
  theories constructed in~\cite{Costello:2019tri}. We prove that, for
  $\mf{g}= \mf{sl}_2$, these Q-operators satisfy Baxter's TQ relation.
\end{enumerate}

\section{'t Hooft lines in four-dimensional Chern-Simons theory}
Before we turn to linking the Q-operator and 't Hooft lines, we will
introduce 't Hooft lines in four-dimensional Chern-Simons theory.  The
analysis in this section is mostly at the classical level. An analysis
of 't Hooft lines at the quantum level is quite subtle, and will be
discussed from two different perspectives in sections
\ref{sec:coulomb} and \ref{sec:surface_defects}.

Let us start by recalling the definition of 't Hooft line in ordinary
four-dimensional $U(1)$ Yang-Mills theory.  In this case, a line
operator is said to have magnetic charge $k$ if, in the presence of
the line operator, the gauge field on the complement of the line
defines a topologically non-trivial $U(1)$ bundle on the two-sphere
surrounding the line.  This bundle should have first Chern class $k$.

An 't Hooft line of charge $k$ is by definition a line operator of magnetic charge $k$ and electric charge $0$.   

In order to satisfy the Yang-Mills equations, the field strength in the presence of the 't Hooft line should look like a Dirac monopole:
\begin{equation} 
F \sim \norm{x}^{-3} \eps_{ijk} x_i \d x_j \d x_k + \text{ regular terms},
 \end{equation}
where $x_i$ are the coordinates normal to the line.  

For a non-Abelian Yang-Mills theory with gauge group $G$, we can define a Dirac monopole by choosing a cocharacter 
$\rho\colon U(1) \to G$ to embed the $U(1)$ Dirac monopole into $G$. A Dirac monopole of this coweight has the feature that on the $S^2$ surrounding the line, the gauge field defines a $G$-bundle which is the bundle associated to the $U(1)$ bundle on $S^2$ of first Chern class $k$, by the map $\rho$.  This means that the field strength takes the form
\begin{equation} 
 F \sim \norm{x}^{-3} \eps_{ijk} x_i \d x_j \d x_k \rho + \text{ regular terms} ,
\end{equation}
where we view $\rho$ as an element of $\g$.

In order to define an 't Hooft line one should require the field
strength to approach the Dirac monopole, up to some appropriate class
of gauge transformations. This is a subtle problem; one reason is that
for a non-Abelian gauge group the field strength is not gauge
invariant.  Further, ``bubbling'' gauge field configurations can
potentially collapse on the 't Hooft line and change its charge. Even
semiclassically, a proper UV definition of the 't Hooft line requires
one to understand and quantize the phase space of such field
configurations. This is a rather theory-dependent procedure.

Let us now try to understand what gauge fields in four-dimensional
Chern-Simons theory look like when we ask that they have monopole
charge on the $2$-sphere surrounding a line in the topological
direction.  Let us first recall that four-dimensional Chern-Simons
theory is an analytically-continued theory as in \cite{Witten:2010zr},
where the space of fields is a complex manifold and the action
functional is a holomorphic function.  The functional integral should
be performed over a contour, and in perturbation theory the contour is
irrelevant. Instead of writing the gauge group as a compact group, we
will write it as a complex group, as is more natural in our
setting. The choice of contour might involve a choice of real form of
the complex group.

Because of the mixed holomorphic/topological nature of the theory,
studying the behaviour of the gauge field on the unit two-sphere will
not be convenient. Instead we will consider a topologically equivalent
region, which is the boundary of the solid cylinder where
$\abs{z} \le \eps$ and $\abs{y} \le \eps$.

This boundary is divided into the three regions: 
\begin{align} 
  y = - \eps, & \quad  \abs{z} \le \eps, \\
  -\eps \le y \le \eps, & \quad \abs{z} = \eps, \\
  y = \eps, & \quad \abs{z} \le \eps.
 \end{align}
 A solution to the equations of motion (modulo gauge transformation)
 for four-dimensional Chern-Simons theory on this region is described
 by:
\begin{enumerate} 
\item A holomorphic bundle on the disc $\abs{z} \le \eps$ and
  $y = -\eps$ and at $y = \eps$.
\item An isomorphism between these two bundles when restricted to
  $\abs{z} = \eps$. This isomorphism is provided by parallel transport
  in the $y$-direction on the cylinder $\abs{z} = \eps$,
  $-\eps \le y \le \eps$.
\end{enumerate}
Trivializing the holomorphic bundles at $y = \pm \eps$, the parallel
transport from $y = -\eps$ to $y = \eps$ is an element of the loop
group $LG$ of the complex gauge group $G$. The change of
trivialization at $y = -\eps$ acts by left multiplication with an
element of $L_+ G \subset L G$, where $L_+ G$ is the group of loops
which extend to a holomorphic map from the disc to $G$.  Change of
trivialization at $y = \eps$ acts by right multiplication by $L_+ G$.
Therefore the moduli space of solutions to the equations of motion on
the boundary of the solid cylinder is
\begin{equation} 
  L_+ G \backslash L G / L_+ G . 
\end{equation}
We will define an 't Hooft line (at least on-shell) by declaring that
our solutions to the equations of motion is locally gauge equivalent
to the Dirac singularity given by a coweight $\rho$ of $G$, which is
the point
\begin{equation}
  z^\rho \in LG.
\end{equation}

Let us try to connect this prescription with the more familiar 't
Hooft lines of Yang-Mills theory. In four-dimensional $\mathcal{N}=2$
gauge theories, BPS 't Hooft operators impose an additional
singularity in an extra scalar field $\Phi$, so that Bogomolny
equations $F = * \d\Phi$ are satisfied
\cite{Kapustin:2005py,Kapustin:2006hi}. In particular,
\begin{equation}
  [D_{x_3} + \Phi, D_{x_1} + i D_{x_2}]=0.
\end{equation}

In four-dimensional Chern-Simons theory, one has only three of the
four components of the gauge field: $A_x$, $A_y$, $A_{\zbar}$.  Here
we work on $\R^2 \times \C$, with coordinates $x$, $y$, $z$ and align
the 't Hooft line along the $x$-direction.

If we identify $A_y = A_{x_3} + i \Phi$ and
$A_{\bar z} = A_{x_1} + i A_{x_2}$ from a BPS 't Hooft operator we
obtain an complexified, $x$-independent solution of the equations of
motion which is a good candidate Dirac singularity for
four-dimensional Chern-Simons theory.  One can then check that the BPS
't Hooft operator does indeed give us a solution of the
four-dimensional Chern-Simons equations of motion with the behaviour
described above. See e.g.~\cite{Kapustin:2006pk}.

\subsection{'t Hooft lines and the the affine Grassmannian}

In the mathematics literature, the space $LG / L_+ G$ is known as the
\emph{affine Grassmannian}. Strictly speaking, the affine Grassmannian
is a slight variant of this, where we replace the group $LG$ of loops
into $G$ by the group $G((z))$, whose elements should be viewed as
Laurent series valued in $G$.\footnote{For example $GL_n((z))$ is the
  group of invertible $n \times n$ matrices whose entries are Laurent
  series. For a general group $G$, we can define $G((z))$ by viewing
  $G$ as the subgroup of $GL_n$ cut out by some polynomial equations,
  like $A A^T = 1$ in the case of $O(n)$; and then defining $G((z))$
  to be the subgroup of $GL_n((z))$ cut out by the same polynomial
  equations, e.g.\ $A(z) A^T(z) = 1$ for $O(n)((z))$.}

If we want to consider a very small cylinder surrounding the 't Hooft line, the affine Grassmannian version of this space is in fact the correct thing to use. We should use Laurent series so that we have only finite-order poles at $z = 0$ and we can expand in series in $z$ because we are working in a region where $z$ is very small.

What this analysis tells us is that to define an 't Hooft line, we
should ask that our fields extend across some submanifold of the
affine Grassmannian $\op{Gr}_G = G((z))/ G[[z]]$ which is stable under
the action of left multiplication by $G[[z]]$ and contains $z^\rho$.
Satisfyingly, it is well-known \cite{MR3752460} that the set of
$G[[z]]$ orbits in the affine Grassmannian is precisely parametrized
in this manner in terms of the set of dominant coweights $G$. This is
our semiclassical definition of an 't Hooft line.

Most $G[[z]]$-orbits in the affine Grassmannian are not closed, and
their closures are singular algebraic varieties. It is technically
difficult to work with these varieties, but the difficulties can be
overcome with some input from the theory of Coulomb branches of
three-dimensional ${\cal N}=4$ gauge theories, as we will see in
section \ref{sec:coulomb}. Furthermore, the orbits corresponding to
minuscule coweights are always smooth and closed. For the groups
$PSL_n$, the singularities of more general 't Hooft operators can be
resolved or deformed by splitting them into a collection of minuscule
ones \cite{Gomis:2011pf}. (For other groups there are not enough
minuscule coweights to do this.)

\subsection{'t Hooft lines as generated by singular gauge transformations}
\label{sec:singular_gauge}

A convenient ansatz for us is that an 't Hooft line of charge $\rho$
at $y=0$, $z=0$ in four-dimensional Chern-Simons theory is generated
by a gauge transformation with a singularity at $z = 0$, $y = 0$.  A
gauge field on a topologically non-trivial bundle, like that sourced
by an 't Hooft line, is always obtained from gluing trivial bundles on
patches by using gauge transformations on the overlap.  The unusual
feature about four-dimensional Chern-Simons theory, as opposed to
ordinary Yang-Mills theory, is that locally all solutions to the
equations of motion are trivial.

Therefore, we can engineer the field sourced by an 't Hooft line by
taking a trivial gauge field on the region $y \le 0$, $y \ge 0$, and
gluing these two trivial gauge fields by the gauge transformation
$z^{\mu}$.

If we do this, then the parallel transport of the gauge field from the
region $y < 0$ to the region $y > 0$ is of the form
\begin{equation} 
	g_1(z) z^{\mu} g_2(z) ,
\end{equation}
where $g_1(z)$, $g_2(z)$ are arbitrary $G$-valued holomorphic
functions of $z$, regular at $0$.  Note that the possible values for
the parallel transport only depend on the Weyl orbit of $\mu$, and
that we can derive the same expression by considering the $G[[z]]$
orbit in the affine Grassmannian containing $z^{\mu}$.

If we allow monopole bubbling (which does not happen for minuscule coweights), then, if $\mu$ is a dominant coweight, we also allow field configurations such that the parallel transport takes the form
\begin{equation} 
	g_1(z) z^{\mu'} g_2(z) ,
\end{equation}
where $\mu' \le \mu$.

\subsection{'t Hooft lines at infinity}

Let us define an 't Hooft line at $0$ as above, by a singular gauge
transformation or equivalently by specifying the singularity in the
parallel transport of the gauge field past the 't Hooft line. We can
define an 't Hooft line at $\infty$ in a similar way.

In \cite{Costello:2013zra, Costello:2017dso}, the boundary conditions
at $z = \infty$ is that the gauge fields and gauge transformations go
as $1/z$ near infinity.  We extend the plane $\C$ into $\CP^1$, with
the point at $\infty$ being special because the one-form $\d z$ has a
second-order pole there.

We define the 't Hooft line at $\infty$ to be a modification of the
boundary conditions, as follows.  We ask that for $y \le 0$ and for
$y \ge 0$, the gauge field and all gauge transformations go as
$1/z$. The gauge fields on the regions $y \le 0$, $y \ge 0$, are
related by the gauge transformation $z^{\mu}$.  Because $z^{\mu}$ has
a singularity at $z = \infty$, this gauge transformation is not one of
the permitted ones at $\infty$, so that we have modified the boundary
conditions.

With an 't Hooft line at $\infty$, the parallel transport of the gauge
field on a path from $y < 0$ to $y > 0$ takes the form
$g_1(z) z^{\mu}g_2(z)$, where $g_1(z)$, $g_2(z)$ are gauge
transformations of the kind allowed at $\infty$.  This means that
$g_1(z)$ is regular in a neighbourhood of $z = \infty$ and
$g_1(\infty) = 1$.

The Dirichlet boundary conditions we use for four-dimensional
Chern-Simons theory completely break the gauge symmetry at
$z = \infty$. In the bulk, two 't Hooft lines associated to coweights
in the same Weyl orbit are equivalent, as the corresponding
singularities are related by a gauge transformation.  At the boundary,
this is no longer true, so that 't Hooft lines at $z = \infty$ are
labelled by coweights, and not by Weyl orbits of coweights.

Instead, at $z = \infty$ we have a $G$-global symmetry.  In particular
an element $w$ in the Weyl group gives a global symmetry of the theory
which sends the 't Hooft line at $\infty$ of charge $\mu$ to that of
charge $w(\mu)$.

This will turn out to be \emph{essential} because we will identify a
Q-operator with an 't Hooft line of charge $\mu$ in the bulk together
with the corresponding 't Hooft line at $\infty$ of charge
$-\mu$. Such configurations are labelled by coweights and not by Weyl
orbits of coweights.  Q-operators are also labelled by (minuscule)
coweights, and coweights related by a Weyl transformation give
different Q-operators.  For example, for $G = PSL_2$, there are two
Q-operators $\mbf{Q}_+$, $\mbf{Q}_-$ which correspond to the two
minuscule coweights of $PSL_2$ (or better, to the two fractional
minuscule coweights of $SL_2$).

\section{Background on Q-operators}

We will follow \cite{Bazhanov:2010ts} for definitions and conventions
on Q-operators.  In this section we will mostly discuss gauge groups
$SL_2$ and $PSL_2$.
 
Consider an integrable spin chain with $N$ sites and with periodic
boundary conditions. The boundary condition can be modified by a
twist, where we apply an element of the Cartan $H$ of the group $G$
when we identify the two sides of the open spin chain to make it
periodic. The twist parameter will be modelled by $e^{\i \phi}$, for
$\phi \in \mf{h}$, the Lie algebra of $H$.  The Q-operators are only
well-defined with a non-zero twist parameter $\phi$.

From the point of view of four-dimensional Chern-Simons theory, the
twist parameter is given asking that the gauge field $A$ does not tend
to zero at $z = \infty$, but is the element $\i \phi$ in the
Cartan. For $SL_2(\C)$, the Cartan is of rank $1$, so $\phi$ is a
scalar, and
\begin{equation} 
	A = \op{Diag}(\i \phi, - \i \phi ) \d x ,
\end{equation}
where $x$ is the periodic direction. 

Let us write down the TQ and QQ relations for $SL_2$.  The TQ relation is 
\begin{equation} 
\mbf{T}(z,\phi) \mbf{Q}_{\pm}(z,\phi) = (z - \tfrac{1}{2}\hbar)^N \mbf{Q}_{\pm}(z +\hbar,\phi) +  (z + \tfrac{1}{2}\hbar)^N \mbf{Q}_{\pm}(z -\hbar,\phi)\label{eqn_tq} .
 \end{equation}
In addition,  if $\mbf{T}_j^+(z,\phi)$ is the transfer matrix for the Verma module of highest weight $j$, we have the QQ relation
\begin{equation} 
 	2 i \op{sin} (\phi/2) 	\mbf{T}^+_{j-\frac{1}{2}} (z) = \mbf{Q}_+(z + \hbar j )  \mbf{Q}_- ( z - \hbar j  ).\label{eqn_qq} 
 \end{equation}
 Here $\phi$ is the twist parameter.

A final identity we will need is that $Q_+$ and $Q_-$ are related by the Weyl group of $SL_2(\C)$:
\begin{equation} 
	w \mbf{Q}_+(z,\phi) w = \mbf{Q}_-(z,-\phi). 
\end{equation}
This is equation (3.55) of \cite{Bazhanov:2010ts}.

\section{Q-operators as 't Hooft lines}
\label{sec:q}

Let us introduce the particular 't Hooft lines that we will match with
the Q-operators for the group $SL_2$.  We let
$\mbf{H}_{\pm \half}(z_0)$ be the 't Hooft lines of fractional charge
$\pm \half$ at $z_0$ and of charge $\mp \half$ at $\infty$.  By
definition, the field sourced by the 't Hooft line has a singularity
which near both $z_0$ and $\infty$ is given by gluing the trivial
field configurations along $y \le 0$ and $y \ge 0$ using the singular
gauge transformation
\begin{equation} 
	\begin{pmatrix}
		(z-z_0)^{\pm \half} & 0 \\
		0 & (z - z_0)^{\mp \half} 
	\end{pmatrix}.
\end{equation}
There are two important points to note about this. First, this gauge
transformation has a branch cut. This reflects the fact that we are
dealing with fractional 't Hooft lines, which live at the end of a
Dirac string.

Since gauge transformations are broken at $z = \infty$, the operators
$\mbf{H}_{\pm \half}$ are \emph{not} equivalent.

Our proposal, for $G = SL_2$, is that the two operators
$\mbf{Q}_+(z)$, $\mbf{Q}_-(z)$ are related to the operators
$\mbf{H}_{\pm \half}(z)$ by an equation of the form
\begin{equation} 
  \mbf{Q}_{\pm}(z) = G(z) \mbf{H}_{\pm\half}(z) 
\end{equation}
for a normalizing function $G(z)$.  

If 't Hooft lines are to be related to Q-operators in this way, then
the TQ and QQ relations should be uplifted to relations involving line
defects.

The TQ relation says that the Wilson line, which becomes $T$, when
fused with an 't Hooft line $Q$, must become a sum\footnote{In fact
  there is a short exact sequence of line defects, rather than an
  equality: at the level of the trace this distinction is irrelevant.}
of two 't Hooft lines, with certain prefactors.

The prefactors, however, are problematic. For the TQ equation
\eqref{eqn_tq} to hold in the category of line defects, the prefactors
$(z \pm \half \hbar)^N$ must be represented by a line defect, say $L$,
which, when we cross it with the $N$ fundamental Wilson lines at
$z = 0$, yields the factor $(z \pm \tfrac{1}{2} \hbar)^N$.  Such a
line defect does not exist with gauge group $SL_2$. Indeed, matching
magnetic charges on both sides of the TQ relation, this putative
line defect must have no magnetic charge, and so be a Wilson line. The
exchange of gluons between two Wilson lines goes like
$1 + O(\hbar/z)$, so can not give the factor of $z\pm \half \hbar$.

The solution to this problem lies in the fact that the T-operator
that arises from four-dimensional Chern-Simons theory has a very
particular normalization, arising from the quantum determinant. Let us
explain how this arises, and how it solves the problem just discussed.

\subsection{Normalizing the T-operator}

In the standard theory of integrable spin chains, the L-operator is a
local quantity from which one can build the transfer matrix.  In the
language of \cite{Costello:2017dso} the L-operator arises from a
fundamental Wilson line of $SL_2(\C)$ at $z = 0$ crossing some other
Wilson line at $z$.  In \cite{Costello:2017dso}, the gauge theory
setup forces one to have an L-operator normalized so that the quantum
determinant is $1$.  In standard references, the L-operator does not
have this normalization.  We will refer to the unnormalized
L-operator as $\til{L}$, and the corresponding transfer matrix as
$\mbf{\til{T}}$, whereas the normalized L-operators and transfer
matrices will be $L$ and $\mbf{T}$.

There is a Wilson line associated to any representation of
$\mf{sl}_2$.  Let $\rho(e)$, $\rho(f)$, $\rho(h)$ be the matrices
defining this representation, where as usual $[e,f] = h$,
$[h,e] = 2 e$, $[h,f] = 2 f$.  Then the unnormalized L-operator
corresponding to this Wilson line is
\begin{equation} 
	\til{L}(z) =  \begin{pmatrix}
		z + \hbar \tfrac{1}{2} \rho(h) &\hbar  \rho(f) \\
		\hbar \rho(e) & z - \hbar \tfrac{1}{2} \rho(h)
	\end{pmatrix}.
\end{equation}
The normalized L-operator is
\begin{equation} 
	L(z) = F(z,j) \begin{pmatrix}
		z + \hbar \tfrac{1}{2} \rho(h) &\hbar  \rho(f) \\
		\hbar \rho(e) & z - \hbar \tfrac{1}{2} \rho(h)
	\end{pmatrix} .
\end{equation}
If we have a highest-weight representation of $\mf{sl}_2(\C)$ of
weight\footnote{Our conventions are such that $h$ acts by $2j$ on the
  highest weight vector.} $j$, then, as we show in
appendix~\ref{sec:appendix_prefactor}, the prefactor $F(z,j)$ must
satisfy the difference equation
\begin{equation} 
  F(z,j) F(z+\hbar,j) =  \frac{1}{z^2 + \hbar z  -  \hbar^2 j(j+1) }. \label{eqn_difference_prefactor}
 \end{equation}
This equation states that the quantum determinant is one. 

For integral $j$, this equation has a unique rational solution, which
we can write as a ratio of $\Gamma$-functions
\begin{equation} 
	F(z,j)  
	= \frac{1}{2  \hbar} \frac{ \Gamma\left(\frac{1}{2 \hbar} (z+\hbar(j+1) \right) \Gamma\left(\frac{1}{2 \hbar} (z -  \hbar j )  \right)} 
	{\Gamma\left(\frac{1}{2 \hbar} (z+\hbar(j+2) \right)  \Gamma\left(\frac{1}{2 \hbar} (z-\hbar j+ \hbar) \right)} \label{eqn_F}
\end{equation}
with a nice perturbative expansion in odd powers of $\hbar / z$.

For all $j$, there is a unique solution to
\eqref{eqn_difference_prefactor} which is $z^{-1}$ times a series in
$\hbar/z$, and this is the the prefactor determined by gauge theory
considerations.  This perturbative solution is the asymptotic
expansion of \eqref{eqn_F} in a Stokes sector of width $2 \pi$ in the
$\hbar / z$ plane centered on the positive real axis, but it is not
odd under $z \to - z$. The alternative solution $- F(-z,j)$ agrees
with the perturbative solution in a Stokes sector of width $2 \pi$ in
the $\hbar / z$ plane centered on the negative real axis.

This shows that, for non-integral $j$, there are two natural
non-perturbative completions of the perturbative Wilson line, with the
normalizing factors $F(z,j)$ and $-F(-z,j)$.  (More generally, we
could multiply $F(z,j)$ by an appropriate trigonometric function of
$z/\hbar$ which does not affect the asymptotic expansion of $F(z,j)$
in the Stokes sector where it is defined)). The physical meaning of
this ambiguity is obscure, but perhaps it could be explained in the
string-theoretic embedding of the system \cite{Costello:2018txb,
  Ashwinkumar:2018tmm, Nekrasov:String-Math-2017}.  In any case, it is
only relevant in a non-perturbative completion of four-dimensional
Chern-Simons theory, which is outside of the scope of this paper.

The normalized and unnormalized T-operators in a spin chain system
with $L$ sites are related by
\begin{equation} 
  \mbf{T}(z,\phi) = F(z)^N \mbf{\til{T}}(z,\phi) 
\end{equation}
This T-operator is the one obtained by the perturbative analysis of
four-dimensional Chern-Simons theory, because the quantum determinant
relation can be derived from QFT considerations
\cite{Costello:2018gyb}.

\subsection{Normalizing the 't Hooft line}

Let 
\begin{equation} 
  \label{eqn_G}
  G(z)
  =
  \frac{1}{(2 \hbar)^{1/2}}
  \frac{ \Gamma\left(\frac{1}{2 \hbar} (z+\tfrac{\hbar}{2}) \right)}
  {\Gamma\left(\frac{1}{2 \hbar} (z+\tfrac{3\hbar}{2}) \right)}.
\end{equation}
Our proposed normalization for the relation between the Q-operators
and the 't Hooft lines is that
\begin{equation} 
  \mbf{H}_{\pm\frac{1}{2}} (z) =G(z)^N \mbf{Q}_{\pm}(z),
\end{equation}
where the 't Hooft line crosses $N$ Wilson lines, with periodic
boundary conditions.  Let us check that this normalization renders the
statement that the TQ and QQ relations are relations among line
defects plausible.

Let us first consider what form the QQ relation might take for 't
Hooft lines. The QQ relation for the normalized T-operator
$\mbf{T}^+_j$ for a Verma module of highest weight $j$ takes the form
\begin{equation}
			2 i \op{sin} (\phi/2) F(z,j-\tfrac{1}{2})^{-N} \mbf{T}^+_{j-\frac{1}{2}} (z) =  \mbf{Q}_{+}(z + \hbar j )  \mbf{Q}_{-}( z - \hbar j  ).
 \end{equation}
Writing this in terms of the 't Hooft lines, the equation is
\begin{equation} 
	\mbf{T}^+_{j-\tfrac{1}{2}}  =  \frac{F(z,j-\tfrac{1}{2})^N}{ G(z+\hbar j)^N G(z - \hbar j)^N}  \frac{1}{2 i \op{sin} (\phi/2)}  \mbf{H}_{\frac{1}{2}}(z + \hbar j )  \mbf{H}_{-\frac{1}{2}}( z - \hbar j  ).
\label{eqn_qq_prefactor}
\end{equation}
This is supposed to come from an identity between line operators in
four-dimensional Chern-Simons theory. For this to hold, the prefactor
$F(z,j-\half)^N/G(z+\hbar j) G(z - \hbar j)$ must be independent of
$z$, as in a theory with a simple gauge group, we are not free to
multiply the transfer matrix associated to a line defect by a function
of $z$. Fortunately, we have
\begin{equation} 
  G(z+\hbar j) G(z - \hbar j) = F(z,j-\tfrac{1}{2}). \label{eqn_G_relation}  
\end{equation}
Further, the factor $(2 i \op{sin}(\phi/2))^{-1}$ in
\eqref{eqn_qq_prefactor} arises from the collision of the two 't Hooft
lines at $z = \infty$. We will derive the QQ relation in this form
from first principles in section~\ref{sec:QQ}.

Note that Stirling's formula implies that $G(z)$ has an asymptotic
expansion\footnote{This expansion is valid if $\hbar$ is a small
  positive real number and $\abs{\op{Arg}(z)} < \pi$.} as
$\hbar \to 0$ which is
\begin{equation} 
G(z) = z^{-1/2} (1 + C_2 \hbar^2 z^{-2} + C_3 \hbar^3 z^{-3} + \dotsb)
 \end{equation}
 for certain constants $C_i$. That is, $G(z)$ is a series in $\hbar / z$ times $z^{-1/2}$.  The factor of $z^{-1/2}$ introduces a branch cut, consistent with the fact that we are working with a fractional 't Hooft line, which lives at the end of a Dirac string.

What about the TQ relation? In the usual normalization, this takes the form
\begin{equation} 
 \mbf{\til{T}}(z,\phi) \mbf{Q}_{\pm}(z,\phi) 
  = (z - \tfrac{1}{2} \hbar)^N \mbf{Q}_{\pm}(z +\hbar,\phi) +  (z + \tfrac{1}{2}\hbar)^N \mbf{Q}_{\pm}(z -\hbar,\phi).
\end{equation}
As we have discussed, the factors of $(z \pm \tfrac{1}{2} \hbar)^N$
can not be present if we would like a relation that holds at the level
of line defects.  For TQ relation to make sense in our context, we
need the normalization factor relating the 't Hooft lines to the
Q-operators and that relating $\mbf{T}$ to $\mbf{\til{T}}$ to cancel
the factors of $(z \pm \tfrac{1}{2} \hbar)^N$.  The normalization we
have derived does exactly this.
 
Inserting the prefactors relating the 't Hooft lines and the
Q-operators, the normalized version of the TQ relation that we need to
prove is
\begin{multline} 
  F(z,\half)^{-N} G(z)^{-N} \mbf{T}(z,\phi) \mbf{H}_{\frac{1}{2}}(z,\phi) 
  \\
  =  G(z + \hbar)^{- N} (z - \half\hbar)^N \mbf{H}_{\frac{1}{2}}(z + \hbar , \phi) + G(z - \hbar)^{- N} (z + \half\hbar)^N   \mbf{H}_{\frac{1}{2}}(z - \hbar,\phi).
 \end{multline}
 Using the fact that $F(z,\half) = G(z + \hbar)G(z - \hbar)$, and
 multiplying both sides by $G(z+\hbar)^{N} G(z-\hbar)^{N} G(z)^{N}$ we
 find that the TQ relation takes the form
\begin{multline} 
  \mbf{T}(z,\phi) \mbf{H}_{\frac{1}{2}}(z,\phi)  
  =   G(z - \hbar)^{N}G(z)^{N}  (z - \half\hbar)^N \mbf{H}_{\frac{1}{2}}(z + \hbar , \phi) \\ + G(z)^{N} G(z + \hbar)^{N} (z + \half\hbar)^N   \mbf{H}_{\frac{1}{2}}(z - \hbar,\phi).
 \end{multline}
Using the difference equation
\begin{equation} 
	G(z)^{-1} G(z+\hbar)^{-1}= z + \tfrac{1}{2} \hbar  \label{eqn_g_difference}
 \end{equation}
 we find that, expressed in terms of the operators
 $\mbf{H}_{\pm\frac{1}{2}}$, the normalized TQ relation takes the form
 \begin{equation} 
   \mbf{T}(z,\phi) \mbf{H}_{\frac{1}{2}}(z,\phi)  =   \mbf{H}_{\frac{1}{2}}(z + \hbar , \phi) +   \mbf{H}_{\frac{1}{2}}(z - \hbar,\phi). \label{eqn_th}
  \end{equation}
As desired, the factors of $(z \pm \tfrac{1}{2} \hbar)^N$ have been cancelled.

What we have shown is that, when we introduce the normalizing factor
in the T-operator required by the quantum determinant relation, the TQ
relation takes the simpler form \eqref{eqn_th}.  This is essential for
our story to work, because only an equation with integral coefficients
can come from a relation between line defects.

Note that, as in the case of Wilson lines for $SL_2$, there are several possible non-perturbative completions of the normalization factor $G(z)$.  One trivial variation is to simply use $-G(z)$, which will of course still satisfy the TQ and QQ relations.  This is simply the parity-reversal of the 't Hooft line.  More interestingly,  we can replace $G(z)$  by  $\i G(-z)$.  This still satisfies   \eqref{eqn_g_difference} and so will satisfy the TQ relation.   If we use $\i G(-z)$ in place of $G(z)$, then the QQ relation will still hold as long as we replace $F(z)$ by $-F(-z)$ in the normalization of the Wilson line. 

As in the analysis of Wilson lines, $G(z)$ and $\i G(-z)$ have the same asymptotic expansion, as one can see from the fact that
\begin{equation} 
	\i G(-z) = G(z) \frac{1  -\i e^{\i \pi z/ \hbar}  }{1+ \i e^{\i \pi z/\hbar} }. 		 
\end{equation}
In the region where $\hbar \to 0$ and $z/\hbar$ is close to the positive imaginary axis, the two expressions coincide up to exponentially suppressed terms. 

This suggests there are two (or more) natural non-perturbative
completions of the 't Hooft line, just as there are two
non-perturbative completions of the Wilson line.  If we take the pair
of Q-operators built from $G(z)$, their product yields the Wilson line
normalized by $F(z,j)$.  Similarly, the Q-operators built from
$\i G(-z)$ yields the Wilson line normalized by $-F(-z,j)$.  However,
this is not all we can do: if we multiply the Q-operator normalized by
$G(z)$ with that normalized by $\i G(-z)$, we get the T-operator
normalized by $F$ times a trigonometric function of $z/\hbar$.  Again,
we don't have a good physical understanding of the various
non-perturbative completions of 't Hooft and Wilson lines, but it
would be fascinating to understand it from a string-theoretic UV
completion of the theory.

\section{The TQ relation from the Witten effect}
\label{sec:Witten}

We will now show that our formulation of the TQ equation in terms of
't Hooft lines \eqref{eqn_th} follows from the Witten effect
\cite{Witten:1979ey}. Let us first recall some background on the
Witten effect.

Let $G$ be the gauge group, and $T$ the maximal torus. Let
$\Gamma = \Hom(U(1), T)$ be the coweight lattice, and
$\Gamma^\vee = \Hom(T,U(1))$ be the weight lattice.  The magnetic
charge of a line operator lives in $\Gamma$, and the electric charge
lives in $\Gamma^\vee$.  The Killing form gives a map
$\rho\colon \Gamma \to \Gamma^\vee$.

In the case of $G = SU(2)$, the electric and magnetic charge lattices
are both $\Z$, but the map from the magnetic to the electric charge
lattices is given by multiplication by $2$.  We can take a basis for
the magnetic charge lattice to be associated to the element
$h \in \mf{sl}_2(\C)$.

The Witten effect states the following. Suppose we have a
four-dimensional gauge theory with gauge group $G$, and a line defect
with magnetic charge $m \in \Gamma$.  Then, the electric charge is not
in $\Gamma^\vee$; rather, it is in
\begin{equation} 
  \Gamma^\vee + \rho(m) \theta / 2 \pi  ,
\end{equation}
where $\theta$ is the $\theta$-angle.  In particular, if we
continuously vary $\theta$ to $\theta'$, then the electric charge of a
line defect of magnetic charge $m$ shifts by
$\rho(m) (\theta' - \theta) / 2 \pi$.

The Lagrangian of four-dimensional Chern-Simons theory, in the
normalization we are using, is
\begin{equation} 
  \frac{1}{2 \hbar \pi} \int \d z \, \ChS(A)
  = \frac{1}{4\hbar \pi}
    \int \d z \left(\ip{A, \d A} + \frac{1}{3} \ip{A,[A,A]}\right).
 \end{equation}
By integration by parts, we can rewrite this as 
\begin{equation} 
\frac{-1}{4 \hbar \pi} \int z \op{Tr} F(A)^2.   
 \end{equation}

The usual $\theta$-angle in a four-dimensional gauge theory is defined so that we add on a term
\begin{equation} 
- \frac{\theta}{8 \pi^2}  \int \op{Tr} F(A)^2  .
 \end{equation}

Thus, the Lagrangian describes varying $\theta$-angle. If we move an 't Hooft line from $z$ to $z_0 + \hbar$, then the $\theta$-angle seen by the 't Hooft line shifts by $2 \pi$.  This, combined with the Witten effect, tells us that a shift $z \to z + \hbar$ in the position of a line defect of magnetic charge $m$ shifts the electric charge $e$ to $e + \rho(m)$.   

We are interested in this in the case when the gauge group is $SU(2)$, and we have a fractional 't Hooft line at $z$ of magnetic charge $\tfrac{1}{2}$.  We call this line $\mbf{H}_{\frac{1}{2}}(z)$.  We see that the 't Hooft line $\mbf{H}_{\frac{1}{2}}(z + \hbar)$ is the same as the dyonic line at $z$ of magnetic charge $\tfrac{1}{2}$ and electric charge $\rho(\tfrac{1}{2}) = 1$.   

Let us now derive the TQ relations.  We propose that the operator
$\mbf{Q}_+(z)$ is given by the fractional 't Hooft line
$\mbf{H}_{\frac{1}{2}}(z)$ of charge $1/2$, up to the factor
$G(z)$. The T-operator is, of course, associated to a Wilson line.  On
the 't Hooft line, gauge symmetry is broken to the Borel subgroup of
$SL_2$.\footnote{ The easiest way to see this is to note that, viewing
  the 't Hooft line as a line defect of charge $1$ for $PSL_2$, the
  corresponding orbit in the affine Grassmannian is $\CP^1$.  The
  stabilizer of a point in $\CP^1$ is the Borel subgroup.}

This means that bringing a fundamental Wilson line to the 't Hooft line gives a dyonic line obtained by coupling the fundamental representation of $SL_2$ to the Borel gauge symmetry. As a representation of the Borel, the fundamental representation of $SL_2$ is an extension of the charge $1$ and $-1$ representations of the Cartan. At the level of traces,  this representation is equivalent to the sum of the charge $\pm 1$ representations of the Cartan. 

The Witten effect tells us that the 't Hooft line $\mbf{H}_{\frac{1}{2}}(z)$ of charge $\half$ coupled to a representation of charge $\pm 1$ is the 't Hooft line $\mbf{H}_{\frac{1}{2}}(z \pm \hbar)$ with shifted value of the spectral parameter.

We conclude that
\begin{equation} 
\mbf{T}(z) \mbf{H}_{\frac{1}{2}}(z) = \mbf{H}_{\frac{1}{2}}(z + \hbar) + \mbf{H}_{\frac{1}{2}}(z - \hbar). 
 \end{equation}
 This is precisely the TQ relation \eqref{eqn_th}, when we identify as before $\mbf{H}_{\frac{1}{2}}(z) = G(z)^N \mbf{Q}_+(z)$. 

This concludes our derivation of the TQ relation from the Witten effect.

\subsection{The Witten effect in four-dimensional Chern-Simons theory}

We have derived the TQ relation from the Witten effect.  Although the
Witten effect is very robust, a cautious reader might be concerned
about the application of the Witten effect in the non-standard gauge
theory we are using.  In this section we will verify the Witten effect
directly to leading order, by a Feynman diagram computation.

For simplicity, we will present the analysis for the group $G = SL_2$
and a fractional 't Hooft line of charge $\half$, although it is not
difficult to generalize to an arbitrary group.

On the 't Hooft line, the gauge group is broken to the Borel
$B \subset SL_2$.  Since there is a homomorphism from the Borel $B$ to
the Cartan $H$, we can couple a Wilson line in a rank one
representation of $H$ to give a (classically) gauge-invariant
Wilson-'t Hooft line. A rank one representation of $H$ is specified by
the weight, which is the electric charge of the Wilson-'t Hooft line
and which we denote by $e$.

We let $\mbf{H}_{(e,\half)}(z)$ be the Wilson-'t Hooft line of electric charge $e$ and magnetic charge $\half$.

The Witten effect in four-dimensional Chern-Simons theory states that
\begin{equation} 
\mbf{H}_{(e,\frac{1}{2})}(z) = \mbf{H}_{\frac{1}{2}}(z+ e \hbar). 
 \end{equation}
Let us derive this, to leading order in $\hbar$.  

Consider an 't Hooft line wrapping $y = 0$, $z = 0$.   
The field sourced by an 't Hooft line is, essentially by definition, the field given by gluing the trivial field configuration at $y \le 0$ and $y \ge 0$ by the singular gauge transformation $\op{Diag}(z^{\half}, z^{-\half})$ at $y = 0$.   This means that the effect of an 't Hooft line wrapping $y = 0$, $z = 0$ on a Wilson line wrapping $x = 0$, at $z$ is simply the matrix 
\begin{equation} 
	 M(z) = \begin{pmatrix}z^{1/2} & 0 \\
0 & z^{-1/2} 
\end{pmatrix} ( 1 + O(\hbar /z ) ).
 \end{equation}

Adding electric charge to the 't Hooft line means that we add the term
\begin{equation} 
	e \int_{y = 0, z = 0} A_h 
\end{equation}
to the Lagrangian of the theory. To leading order in $\hbar$, the effect of this modification on a Wilson line passing the 't Hooft line will be given by the exchange of a single gluon between the Wilson line and the dyonic line.  We can compute this by using the same Feynman diagrammatics as in \cite{Costello:2017dso}. 

Only the $h$-component of the gauge field is coupled when we deform the 't Hooft line to a Wilson-'t Hooft line. Therefore the group theory factor in the propagator connecting the Wilson and dyonic lines will have $h$ on each end.  

The analytic factor is identical to that considered in the calculations in \cite{Costello:2017dso} of the single gluon exchange between two Wilson lines.  In \cite{Costello:2017dso}, we found that the analytic factor is $\hbar / z$. The component of the quadratic Casimir involving $h \in \mf{sl}_2(\C)$ is $\tfrac{1}{2} h \otimes h$ (as usual $h = \op{Diag}(1,-1)$). We find that introducing electric charge $e$ to the 't Hooft line changes the matrix $M$ by 
\begin{equation} 
	M(z) \mapsto \left(1 +  \frac{e \hbar}{2 z} \begin{pmatrix} 1 & 0 \\ 0 & -1  \end{pmatrix}  \right) M(z) + O(\hbar^2). 
 \end{equation}

Now, it is obvious that
\begin{equation} 
e \hbar \partial_z M(z) = \frac{e \hbar}{2 z} h M(z) + O(\hbar^2). 
 \end{equation}
This verifies the Witten effect, to leading order in $\hbar$:
\begin{equation} 
\mbf{H}_{(1,\frac{1}{2})}(z) = \mbf{H}_{\frac{1}{2}}(z + \hbar) + O(\hbar^2).  
 \end{equation}

 It is more challenging to show that this result holds to all orders
 in $\hbar$. To see this, we will first use the fact, proved in
 section \ref{sec:quantizing_tHooft}, that the 't Hooft line has only
 one parameter, which we can take to be the position in the $z$-plane.

This implies that
\begin{equation} 
  \mbf{H}_{(e,\frac{1}{2})}(z)  = \mbf{H}_{\frac{1}{2}}(z + f(\hbar,e))
\end{equation}
for some unknown function $f$. Both the 't Hooft line and the dyonic
line placed at $z = 0$ are preserved by the symmetry of the system
which scales both $z$ and $\hbar$. This implies that the shift in $z$
on the right hand side of the equation must be linear in $\hbar$, so
that, by the argument we have already provided,
\begin{equation} 
  f(\hbar,e) = e \hbar. 
\end{equation}

\section{Phase spaces of 't Hooft operators}
\label{section_phase}
In this section we will carefully analyze the phase space of
Chern-Simons theory in the presence of 't Hooft lines at $0$ and at
$\infty$.  There are a number of subtle issues that arise when we
discuss 't Hooft lines that do not appear for Wilson lines, so we will
first discuss the phase space in the absence of 't Hooft lines.

\subsection{Moduli of bundles on $\CP^1$}
Recall that we treat four-dimensional Chern-Simons theory as a theory
on $\R^2 \times \CP^1$ with the one-form $\d z$ on $\CP^1$, where all
components of the gauge field and all gauge transformations go like
$1/z$ near $z = \infty$. In \cite{Costello:2017dso, Costello:2018gyb}
we analyzed the gauge theory perturbatively, working near the trivial
field configuration $A = 0$. In that case, the moduli space of
solutions to the equations of motion is a point.  The key point here
is that the trivial holomorphic bundle on $\CP^1$ admits no
infinitesimal deformations or automorphisms which are trivial at
$z = \infty$.

A perturbative analysis is not sufficient for the treatment of 't
Hooft lines, so here we must be more careful. If the group $G$ is
simply connected, then the moduli space of $G$-bundles on $\CP^1$ is
connected, and the only stable bundle is the trivial bundle.  In this
case, the previous analysis suffices.

If $G$ is not simply connected, then the moduli space of $G$-bundles
on $\CP^1$ has components labelled by $\pi_1(G)$, and a
non-perturbative treatment of four-dimensional Chern-Simons theory
must take account of these components.  For example, for the group
$PSL_2$, the moduli of bundles has two components, one containing the
trivial bundle and one containing $\Oo \oplus \Oo(1)$.  In the
non-trivial component, the moduli\footnote{As a stack, the moduli
  space is bigger and contains bundles such as
  $\Oo(1) \oplus \Oo(-2)$.  However, these do not contribute to the
  path integral, as the bundles isomorphic to $\Oo \oplus \Oo(1)$ form
  an open substack.}  of bundles trivialized at $\infty$ is $\CP^1$ is
$\CP^1$.

To see this, note that the group $PSL_2$ acts on the moduli space of
bundles trivialized at $\infty$. The stabilizer of $\Oo \oplus \Oo(1)$
consists of those projective automorphisms of the fibre at
$z = \infty$ which extend to bundle automorphisms on all $\CP^1$.
This is the Borel subgroup $B \subset PSL_2$.  Thus, the moduli of
bundles is $PSL_2 / B = \CP^1$.  More generally, take any simple Lie
group $G$ of adjoint type.  The components of the moduli space of
$G$-bundles on $\CP^1$, trivialized at $\infty$, consist of the
trivial bundle, plus a component for each Weyl orbit of minuscule
coweights of $G$.  Given any minuscule coweight $\mu$, we form a
bundle on $\CP^1$ whose transition function around $\infty$ is
$z^{\mu}$.  This bundle, together with its trivialization at $\infty$,
lives in a moduli space which is $G/P_{\mu}$, where
$P_{\mu} \subset G$ is a parabolic subgroup associated to $\mu$.  (The
Lie algebra $\mf{p}_{\mu}$ is defined to be the subalgebra of $\mf{g}$
of elements of charge $\ge 0$ under $\mu$, and $P_{\mu}$ is the
exponential of $\mf{p}_{\mu}$.)

To see that the moduli space is $G/P_{\mu}$, we note that $G$ acts on
this moduli space by changing the framing at infinity, and at the
level of the Lie algebra, a symmetry at $\infty$ extends to a bundle
automorphism on all of $\CP^1$ if and only if it is if non-negative
charge under $\mu$.

In \cite{Costello:2013zra,Costello:2017dso} it was shown that spin
chain systems arise by placing Wilson lines in four-dimensional
Chern-Simons theory on $\R^2 \times \CP^1$.  The analysis here only
used the trivial bundle on $\CP^1$, and did not take into account the
other components of the moduli space. Because of this, we conclude
that this analysis is complete \emph{only} in the case that $G$ is
simply connected.  When $G$ is not simply connected, the other
components of the moduli space will give rise to other sectors of the
Hilbert space of the spin chain system, which have not been analyzed.

Since our goal in this paper is to make contact with the Q-operators
in ordinary spin chain systems, we see that we are forced to always
work with the simply-connected form of the group.

\subsection{Solutions to the equations of motion in the presence of an 't Hooft line} \label{sec:phasespace}

Consider the four-dimensional Chern-Simons setup on
$\R^2 \times \CP^1$ for a simply connected group.  Place an 't Hooft
line wrapping $y = 0$, $z = 0$ of charge $\mu$. This may be a
fractionally charged 't Hooft line, in which case we view it as living
at the end of a Dirac string at $y = 0$ stretched between $z = 0$ and
$z = \infty$.

In the region $y < 0$, since $G$ is a simply-connected group, the $G$-bundle on $\CP^1$ must be trivial, and in fact trivialized because the trivialization at $z = \infty$ extends across uniquely all $\CP^1$.The same holds for $y > 0$.  This means that we can form a gauge-invariant classical observable  by measuring the parallel transport of the gauge field from $y  \ll 0$ to $y \gg 0$:
\begin{equation} 
	L(z) = \op{PExp} \int_{y} A_y (z) \in G. 
\end{equation}
This is a holomorphic $G$-valued function of $z$, with possible poles and zeroes at $z = \infty$, $z = 0$ arising from the 't Hooft lines at $0$ and $\infty$.  

The space of solutions to the equations of motion will be a sub-space of the space of holomorphic maps from $\C^\times$ to $G$, specified by the allowed poles at $0$ and $\infty$.  By expanding in Laurent series near $z = 0$, we will view it as a subgroup of the loop group $G((z))$.

As we have seen, the singularity in the 't Hooft line of charge $\mu$ at $z = 0$ means that the parallel transport takes the form
\begin{equation} 
	G[[z]] z^{\mu'} G[[z]] \subset G((z)),
\end{equation}
where $\mu$ is dominant and $\mu' \le \mu$. 

At $z = \infty$, the singularity from an 't Hooft line of charge $\eta$ means the parallel transport takes the form 
\begin{equation} 
	G_0[z^{-1}] z^{-\eta} G_0[z^{-1}],
\end{equation}
where $G_0[z^{-1}]$ is the group of polynomial maps from a neighbourhood of $\infty$ to $G$ which takes the value $1$ at $\infty$. 

We thus conclude that the phase space of the 't Hooft line is the set of operators $L(z) \in G((z))$ which satisfy both of these constraints.  It is known \cite{Braverman:2016pwk} that if $G$ is an ADE group, this space is the Coulomb branch of an ADE quiver gauge theory. This gives us a perspective on quantizing these phase spaces that we will explore in section \ref{sec:coulomb}.

In the case of fractional 't Hooft lines, we will also allow
L-operators in $G(z^{\tfrac{1}{n}})$ where $n$ is the rank of the
center of $G$, and the branch cuts live in the center.  (We could also
of course pass to the adjoint form of the group, in which case the 't
Hooft line is no longer fractional, but we prefer not to do that for
the reasons mentioned above).

\subsection{Poisson bracket on the phase space}
The Poisson bracket on the phase space of the theory in the presence of an 't Hooft line can be computed using the semiclassical version of the RTT relation.  The idea is the following. Fix an 't Hooft line at $z = 0$ crossed by a Wilson line in some representation $R$ at $z$.  Let $A$ be the algebra of functions on the phase space of the 't Hooft line.   The L-operator is an element
\begin{equation} 
	L(z) \in A \otimes \op{End}(R). 
\end{equation}
The matrix entries of $L(z)$ are thus entries in the algebra $A$.  These matrix entries $L^i_j(z)$ are obtained by asking that the Wilson line has initial state $i$ and final state $j$.

\begin{figure}[htbp]
\begin{center}
\begin{tikzpicture}
\node(In) at (0,4) {$\langle i|$};
\node(Out) at (0,0) {$|j \rangle$};
\node(HLeft) at (-2,2) {};
\node(HRight) at (2,2) {};
\draw (HLeft) to (HRight); 
\draw[very thick, color=white] (In) to (Out);
\draw (In) to (Out);  
\node at (0.2,3) {$z$};
\end{tikzpicture}
	\caption{\small{A vertical Wilson line, equipped with an incoming state $\langle i |$ and an outgoing state $|j \rangle$, gives rise to a function $L^i_j(z)$ on the phase space of the  horizontal 't Hooft line. }}
\label{figure_crossing}
\end{center}
\end{figure}
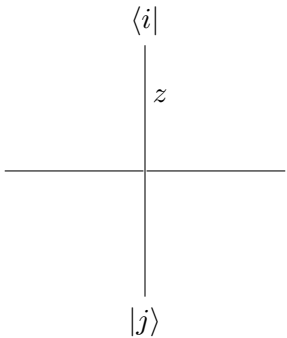

These functions satisfy some commutation relations (or, semiclassically, Poisson brackets) which arise from the R-matrix appearing when Wilson lines cross.  These commutation relations are called the RTT or RLL relations, and were derived from a field theory analysis in \cite{Costello:2018gyb}.  We will review the analysis and apply it to understand the Poisson bracket on the phase space of 't Hooft lines.  

\begin{figure}
\begin{center}
\begin{tikzpicture}
\begin{scope}
\node(In1) at (0,5) {$\langle i|$};
\node(Out1) at (1,0) {$|j \rangle$};
\node(In2) at (1,5) {$\langle k|$};
\node(Out2) at (0,0) {$|l \rangle$};
\node(HLeft) at (-2,2) {};
\node(HRight) at (3,2) {};
\draw[->] (HRight) to (HLeft); 
\draw[very thick, color=white] (In1) to (Out2);
\draw[very thick, color=white] (In2) to (Out1);
	\draw[rounded corners] (In1.south) to (0,3.5) to (1,2.1) to  (Out1);
\node at (1.2,1) {$z$};
	\draw[  rounded corners] (In2.south) to (1,3.5) to  (0,2.1) to  (Out2);
% \draw[rounded corners] (In2.south) to (0,2.1) to  (Out2);
\node at (-0.5,1) {$z'$};
\end{scope}

\node at (4,2) {$=$};

	\begin{scope}[shift={(7,0)}]
\node(In1) at (0,5) {$\langle i|$};
\node(Out1) at (1,0) {$|j \rangle$};
\node(In2) at (1,5) {$\langle k|$};
\node(Out2) at (0,0) {$|l \rangle$};
\node(HLeft) at (-2,3.5) {};
\node(HRight) at (3,3.5) {};
\draw[->] (HRight) to (HLeft); 
\draw[very thick, color=white] (In1) to (Out2);
\draw[very thick, color=white] (In2) to (Out1);
	\draw[rounded corners] (In1.south) to (0,3.5) to (1,2.1) to  (Out1);
\node at (1.2,1) {$z$};
	\draw[  rounded corners] (In2.south) to (1,3.5) to  (0,2.1) to  (Out2);
% \draw[rounded corners] (In2.south) to (0,2.1) to  (Out2);
\node at (-0.5,1) {$z'$};
\end{scope}

\end{tikzpicture}
\caption{\small{Two vertical Wilson lines are ``bent'' to cross each other, above or below a given horizontal 't Hooft line line. }}
\label{figure_bent} 
\end{center}
\end{figure}
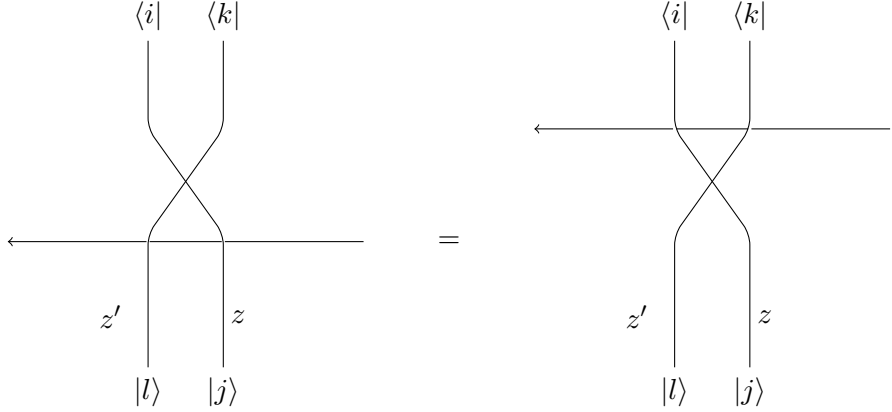

The commutation relation given is described diagrammatically from figure \ref{figure_bent}.  The diagram holds simply because the Wilson lines do not touch the 't Hooft line or each other in the four-dimensional space-time, and so can be freely moved. Translated into symbols, the relation becomes 
\begin{equation}
\sum_{r,s} R^{i k}_{r s}(z - z') L^r_j(z) L^s_l(z') = \sum_{r,s} L^i_r(z) L^k_s(z')R^{rs}_{jl}(z -z').
\end{equation}
In interpreting this equation, note that the entries
$R^{ij}_{rs}(z-z')$ of the R-matrix are scalar functions, whereas
the entries $L^r_j$ of the T-operator are functions on the phase space
in the presence of the 't Hooft line.

To understand the Poisson bracket between 't Hooft lines, we should take the semiclassical $r$-matrix, under the expansion 
\begin{equation} 
  R(z) = 1 + \hbar \frac{1}{z} c + O(\hbar^2),
\end{equation}
where $c \in \mf{g} \otimes \mf{g}$ is the quadratic Casimir. In the
above formula, we should use $c$ acting on the representation
$R \otimes R$, which we write $c^{ik}_{rs}$.

By examining the order $\hbar$ term in the RLL relation, we find the
Poisson bracket between the operators $L^i_j(z)$ is
\begin{equation}
	\{L^i_j(z), L^k_l(z')\} = 	\sum_{r,s}  \frac{1}{z - z'}  c^{i k}_{r s}L^r_j(z) L^s_l(z') - \frac{1}{z-z'}  L^i_r(z) L^k_s(z')c^{rs}_{jl}.
\end{equation}
This expression, in general, describes the Poisson bracket on the
phase space of the 't Hooft line, where we put arbitrary charge $\mu$
at $0$ and $\eta$ at $\infty$.  It is shown in \cite{Kamnitzer2020}
that in the case of an ADE group, this Poisson bracket matches the
natural Poisson bracket on the Coulomb branch of a three-dimensional
$\mathcal{N}=4$ quiver gauge theory.  We will discuss the quiver gauge
theory picture in more detail later.

\section{Oscillator realizations of Q-operators}

In the sequence of papers \cite{Bazhanov:2010ts,Bazhanov:2010jq} it
was shown that the Q-operators for the group $SL_n$ can be realized by
a more fundamental object: an L-operator valued in an oscillator
representation. For the group $SL_2$, the L-operator takes the form
\begin{equation}
  L(z) = \frac{1}{\sqrt{z}} \begin{pmatrix} z+b c & b \cr c & 1\end{pmatrix},
\end{equation}
where $b$, $c$ are the generators of a Weyl algebra, satisfying the
commutation relations $[b,c] = 1$.

In this section we will show how this L-operator arises from the
analysis of four-dimensional Chern-Simons theory. We will show that
the phase space for the group $SL_2$, in the presence of an 't Hooft
line of charge $(\half,-\half)$ at $0$ and $(-\half,\half)$ at
$\infty$, is $\C^2$ with Darboux coordinates $b$, $c$.  Then we will
calculate the L-operator that arises when we cross a Wilson line at
$z$ with these 't Hooft operators, and we will find the result written
above.

Our analysis also yields the Q-operators for $SL_n$ as constructed
in \cite{Bazhanov:2010jq}, and for $SO_n$ as analyzed in
\cite{Frassek:2020nki}.  Our approach gives a uniform construction of
oscillator-valued L-operators associated to 't Hooft lines of
minuscule charge for any simple group with a minuscule coweight,
namely the classical groups and $E_6$, $E_7$.

\subsection{Phase space for the minuscule coweight of  $\mathfrak{sl}_2$}
Let us start by analyzing the phase space for a pair of 't Hooft
operators for $SL_2$ of charge $\mu = (\half,-\half)$ at $0$ and
$-\mu$ at $\infty$.  As we discussed earlier, the boundary conditions
at $z = \infty$ on the region $y < 0$ and $y > 0$ force the bundle to
be trivial, so that parallel transport from $y < 0$ to $y > 0$ is a
well-defined matrix valued function of $z$.

We claim that the permitted singularities at $0$ and $\infty$ show
that the possible matrices are of the form
\begin{equation}
L(z) = \frac{1}{\sqrt{z}} \begin{pmatrix} z+b c & b \cr c & 1\end{pmatrix}
\end{equation}
for arbitrary $b$, $c$. Thus, the phase space is a copy of $\C^2$.

The proof of this is an explicit computation. To perform it, we work
in $PSL_2$ instead of $SL_2$ and drop the normalizing prefactor, so
that $z^{\mu} = \op{Diag}(z,1)$.  Suppose that there are elements
$A(z), B(z)$ in $PSL_2[[z]]$ so that
\begin{equation} 
	A(z) z^{\mu}  B(z) = L(z). 
\end{equation}
Then, clearly, $L(z)$ has no poles at $z = 0$, and $M(0)$ is a matrix
of rank $1$.

Similarly, at $z = \infty$, we have
\begin{equation} 
  \til{A}(z) z^{\mu}  \til{B}(z) = L(z)  ,
\end{equation}
where $\til{A}(z),\til{B}(z)$ go to the identity at $z = \infty$. From
this we see that $L(z)$ has at most a first order pole at infinity,
and the polar part is $\op{Diag}(z,0)$.  Further, the term regular at
infinity is of the form
\begin{equation} 
	\begin{pmatrix}0 & 0\\
		0 & 1
		\end{pmatrix} 
		+ \til{A}_{-1} \begin{pmatrix} 1 & 0 \\ 0 & 0 \end{pmatrix} + 	 \begin{pmatrix} 1 & 0 \\ 0 & 0 \end{pmatrix} \til{B}_{-1}  	 ,
\end{equation}
where $\til{A}_{-1}$, $\til{B}_{-1}$ are the coefficients of $z^{-1}$.   Therefore the coefficient of $z^0$ in $L(z)$ is of the form
\begin{equation} 
	L_0 = \begin{pmatrix}
		a & b \\
		c & 1
	\end{pmatrix}.
\end{equation}
The constraint that $L_0$ has determinant $0$ forces $a = bc$, so that
\begin{equation} 
	L(z) = \begin{pmatrix}
		z + bc & b \\
		c & 1 
		\end{pmatrix}
\end{equation}
as desired.

We have shown that the 't Hooft line, with charge $(\half,-\half)$ at
$0$ and $(-\half,\half)$ at $0$, has phase space $\C^2$ and classical
L-operator as above.

This is a classical limit of the L-operator which gives rise to
$Q$-functions \cite{Bazhanov:2010ts} in $\mathfrak{sl}_2$ spin chains.
For this L-operator, $b$ and $c$ live in an oscillator algebra with
commutation relations $[b,c] = \hbar$.  We will show shortly that for
the 't Hooft line, the Poisson bracket is $\{b,c\} = 1$.

This computation is a proof, to leading order, that the Q-operator is
given by the 't Hooft line.

\subsection{A general minuscule coweight}

Now let us consider an 't Hooft line of charge $\mu$ where $\mu$ is a
minuscule coweight in the adjoint form (necessarily in the simply
connected form $\mu$ is fractional).  We place charge $-\mu$ at
$\infty$.  We will find that the phase space is, as in the case of
$SL_2$ discussed above, a symplectic vector space, which quantizes
into an oscillator algebra.  The L-operator obtained from a Wilson
line crossing the 't Hooft line has a very simple general form.

Let us first recall the definition of minuscule coweight
\cite{MR1795753}.  A minuscule coweight $\mu$ in the Cartan of a
simple Lie algebra $\g$ is characterized by the fact that the only
eigenvalues of $\mu$ on $\g$ are $-1$, $0$, $1$. We denote the $\pm 1$
eigenspaces by $\mf{n}^{\pm}$, and the $0$ by $\mf{l}_{\mu}$.  Note
that $[\mf{n}^{\pm 1}, \mf{n}^{\pm 1}] = 0$.  We can form a parabolic
subalgebra $\mf{p} = \mf{l} \oplus \mf{n}^+$ consisting of the charge
$\ge 0$ elements of $\mf{g}$.  The exponential $P$ of $\mf{p}$ is a
parabolic subgroup of $G$, and the quotient space $G / P$ plays an
important role in the theory.

Minuscule coweights are closely related to symmetric spaces, since we
get a $\Z/2$ grading on $\mf{g}$ where given by reducing the
eigenvalues of $\mu$ modulo $2$.  The symmetric spaces associated to
minuscule coweights are Hermitian.

For example, the minuscule coweights of $PSL_n$ are of the form 
\begin{equation} 
	\mu_k =
		\op{Diag}(\underbrace{z, \dots, z}_{k \text{ times}} ,\overbrace{ 1,\dots,1}^{n-k \text{ times}})	
\end{equation}
for $k = 1$, $\dotsc$, $n-1$.  When viewed as fractional cocharacters
of $SL_n$, they are
\begin{equation} 
	\mu_k = 	\op{Diag}(\underbrace{z^{1-\tfrac{k}{n}}, \dots, z^{1-\tfrac{k}{n}  }}_{k \text{ times}} ,\overbrace{ z^{\tfrac{-k}{n}} ,\dots,z^{\tfrac{-k}{n}}}^{n-k \text{ times}})	.
\end{equation}

The $0$ eigenspace of the minuscule coweight we call $\mf{l}_{\mu}$;
it is a Levi factor of the parabolic $\mf{p}_{\mu}$ consisting of the
$0$ and $1$ eigenspaces.  The Levi factor corresponding to $\mu_k$ is
\begin{equation} 
  \mf{l} = \mf{sl}_k \oplus \mf{sl}_{n-k} \oplus \C \cdot \mu_k ,
\end{equation}
where here we view $\mu_k$ as an element of the Cartan Lie algebra
$\mf{h}$.  The parabolic subalgebra $\mf{p}$ is the subalgebra of
block-upper triangular matrices
\begin{equation} 
	\mf{p} = \mf{sl}_k \oplus \mf{sl}_{n-k} \oplus \C \cdot \mu_k \oplus \op{Hom}(\C^k, \C^{n-k}) . 
\end{equation}
This has the feature that, when we exponentiate to a subgroup $P \subset SL_n$, we have 
\begin{equation} 
	SL_n / P = \op{Gr}(k,n). 
\end{equation}

Now let us analyze the phase space of an 't Hooft line with charge
$\mu$ at $0$ and $-\mu$ at $\infty$, where $\mu$ is minuscule. The
singularity at $0$ means that we can write
\begin{equation} 
	L(z) = A(z) z^{\mu} B(z) ,
\end{equation}
where $A(z)$, $B(z)$ are regular at $0$. We can use the decomposition $\g = \mf{n}^{-} \oplus \mf{l} \oplus \mf{n}^+$ to decompose the elements $A(z)$, $B(z)$:
\begin{align} 
	A(z) &= e^{a^+(z)} A_0(z) e^{a^-(z)} , \\
	B(z) &= e^{b^+(z)} B_0(z) e^{b^-(z)}  ,
\end{align}
where $a^{\pm}(z)$, $b^{\pm}(z)$ are valued in $\mf{n}^{\pm}$, and $A_0$, $B_0$ are valued in $L$, the exponentiation of $\mf{l}$.  Now, 
\begin{equation} 
	z^{\mu} e^{a^{\pm}(z)}  = e^{z^{\pm 1} a^{\pm} (z)}  z^{\mu}. 
\end{equation}
Also, $z^{\mu}$ commutes with $A_0(z)$, $B_0(z)$. Then we can write
\begin{equation}
	\begin{split}
		L(z)& =  e^{a^+(z)} A_0(z) e^{a^-(z)} z^{\mu}   e^{ b^+(z)}    B_0(z) e^{b^-(z)}\\
		&= e^{a^+(z)} A_0(z) e^{a^-(z)}   e^{z b^+(z)} z^{\mu}    B_0(z) e^{b^-(z)}.
	\end{split}
\end{equation}
Since $	e^{a^-(z)} e^{z b^+(z)}$ is an element of $G[[z]]$ which is in $\op{Exp} \mf{n}_-$ at $z = 0$, we can write it in the form 
\begin{equation} 
	e^{a^-(z)} e^{z b^+(z)}=  e^{z \til{b}^+(z)} M_0(z) e^{\til{a}^-(z)} ,
\end{equation}
where $M_0(z)$ is valued in $L$ and is the identity at $z= 0$.
Further, we can move elements $M_0(z)$, $B_0(z)$, $A_0(z)$ past the
elements $e^{a^{\pm}(z)}$, $e^{b^{\pm}(z)}$ at the price of
conjugating them by $B_0(z)$ or $A_0(z)$.  Thus, by re-defining
$a^{\pm}(z)$, $b^{\pm}(z)$ we find that $L(z)$ can be written in the
form
\begin{equation} 
	L(z) =  e^{a^+(z)+ z b^+(z)}  z^{\mu} C_0(z)  e^{ z a^-(z) + b^-(z)}  
\end{equation}
Here, all expressions are regular at $0$, so that we can absorb
$z b^+(z)$ into $a^+(z)$ and $z a^-(z)$ into $b^-(z)$ to give an
expression of the form
\begin{equation} 
  L(z) =  e^{a^+(z)}  z^{\mu} C_0(z)  e^{ b^-(z)}  .
\end{equation}

Using the same analysis with the decomposition 
\begin{equation} 
	L(z) = \til{A}(z) z^{\mu} \til{B}(z) ,
\end{equation}
where $\til{A}(z)$, $\til{B}(z)$ take value $1$ at $\infty$
we get
\begin{equation} 
	L(z) =  e^{z \til{b}^+(z)}  z^{\mu} \til{C}_0(z)  e^{ z \til{a}^-(z) }   ,
\end{equation}
where $\til{a}^-(z)$, $\til{b}^+(z)$ vanish at $z = \infty$.

In the analysis at $0$ and $\infty$, we have decomposed $L(z)$ as a
product of an element of $N^+((z))$, $ L((z))$, and $ N^-((z))$.  This
decomposition is unique.  Since
\begin{equation} 
	 e^{z \til{b}^+(z)}  z^{\mu} \til{C}_0(z) e^{ z \til{a}^-(z) }  =   e^{a^+(z)}  z^{\mu} C_0(z))  e^{ b^-(z)}  ,
\end{equation}
we have
\begin{align}
	 z \til{b}^+(z) & = a^+(z) , \\
	z \til{a}^-(z)  & = b^- (z) , \\
	\til{C}_0(z) &= C_0(z).  
\end{align}
Imposing the constraints that $\til{C}_0(z)$ is the identity at
$\infty$ and $C_0(z)$ is regular at $0$, we find that $C_0(z) = 1$.

Next, $\til{a}^{\pm}(z)$, $\til{b}^{\pm}(z)$ have series expansions in
$\mf{n}^{\pm}$ as series in $z^{-1}$, with no constant term.
Similarly, $a^{\pm}(z)$, $b^{\pm}(z)$ have series expansions in $z$,
where the constant term is allowed. Identifying both sides, we find
that $a^{+}(z)$ and $b^-(z)$ are constant. We let $X^+$, $X-$ be the
values of $a^+(z)$, $b^-(z)$ at $0$.  We find that we can write $L(z)$
uniquely as
\begin{align} 
	L(z) &= e^{X^+} z^{\mu} e^{X^-} \\
	&= e^{X^- / z} z^{\mu} e^{X^+/ z}.
\end{align}

Let us now compute this for minuscule coweights of $\mf{sl}_n$. We
will focus on those which are in the Weyl orbits of
$\mu = (1,0,\dots,0)$, and, because the L-operator becomes conjugate
by a Weyl transformation when we choose a different element in the
Weyl orbit of $\mu$, we only need to compute the L-operator for this
element.  In this case, $\mf{n}^+$ consists of those matrices $M^i_j$
whose only non-zero entries are $M^1_j$, $j > 1$; and $\mf{n}^-$
consists of the matrices with non-zero entries $M^j_1$, $j > 1$.

The decomposition of $L(z)$ as $e^{X^+} z^{\mu} e^{X^-}$ is
\begin{equation}
  \label{gl_n_L} 
  \begin{split}
    L(z) &=  
    \begin{pmatrix}
      1 & b_1 & \dots & b_{n-1}\\
      0 & 1 & \dots & 0 \\
      \vdots & \vdots & & \vdots \\
      0 & 0 & \dots & 1  
    \end{pmatrix}
    \begin{pmatrix}
      z & 0 & \dots & 0\\
      0 & 1 & \dots & 0 \\
      \vdots & \vdots & & \vdots\\
      0 & 0 & \dots & 1  
    \end{pmatrix}	
    \begin{pmatrix}
      1 & 0 & \dots & 0\\
      c_1 & 1 & \dots & 0 \\
      \vdots & \vdots & & \vdots\\
      c_{n-1} & 0 & \dots & 1  
    \end{pmatrix}\\
    &= \begin{pmatrix}
      z + \sum b_i c_i & b_1 & \dots&  b_{n-1} \\
      c_1 & 1 & \dots & 0\\
      \vdots & \vdots & & \vdots \\
      c_{n-1} & 0 & \dots & 1
    \end{pmatrix}.    
  \end{split}
\end{equation}
This is the same L-operator (in the semiclassical limit) as derived in
\cite{Bazhanov:2010jq}.  This paper also included more general
L-operators, but these are simply obtained from the one we have
derived by a Weyl group transformation which reorders the rows and
columns. From our perspective, they are obtained by applying a Weyl
group element to the charge at $\infty$, and so replacing it by a
different coweight in the same Weyl orbit.

\subsection{The Poisson bracket on the phase space for $\mf{gl}_n$}
Let us consider the Poisson bracket on the phase space in the case of the coweight of $\mf{gl}_n$ of charge $(1,0,\dots,0)$.  We have seen that the L-operator \eqref{gl_n_L} is the same as that studied by Bazhanov et al \cite{Bazhanov:2010jq}, with $L^i_1(z) = b_i$, $L^1_j (z)= c_j$.  The quadratic Casimir for $\mf{gl}_n$ is
\begin{equation} 
	c^{ij}_{kl} = \delta^i_l \delta^j_k. 
\end{equation}
Therefore the Poisson bracket takes the form
\begin{equation} 
	(z-z')	\{L^i_1(z), L^1_j(z')\} = L^1_1(z) L^i_j(z') - L^i_j(z) L^1_1(z'). 
\end{equation}
Since $L^1_1(z) = z + b_i c^i$, and $L^i_j (z) = \delta^i_j$ if $i,j > 1$, we find
\begin{equation} 
	\{b_i, c_j\} = \delta^i_j 
\end{equation}
as in \cite{Bazhanov:2010jq}.

\subsection{The Poisson bracket in general}

Now let us return to the case of a general minuscule coweight $\mu$ in
a simple group $G$ in which we have a minuscule coweight $\mu$ at $0$
and we have $-\mu$ at $\infty$.  In this case, it is convenient to
take the Wilson line to live in the adjoint
representation.\footnote{Note that outside of type $A$, the adjoint
  representation does not lift to a quantum Wilson line.  The anomaly,
  as analyzed in \cite{Costello:2017dso}, occurs at two loops and does
  not affect the Poisson bracket computation.}  We decompose $\g$ as
above into $\mf{n}^+ \oplus \mf{l} \oplus \mf{n}^-$ and take a basis
$X^{m}$ for $\mf{n}^+$, and $Y_{n}$ for $\mf{n}^-$.  We then write, as
before
\begin{equation} 
	L(z) = e^{ b_{m} X^m}  z^{\mu} e^{c^n Y_n}. 
\end{equation}
The $b_m$, $c^n$ are coordinates on the phase space. We choose a basis so that the non-degenerate pairing between $\mf{n}^+$ and $\mf{n}^-$ coming from the Killing form is $\delta^n_m$.

The coordinates $b_m$, $c^n$ can be recovered by studying how $L(z)$
acts on the elements $X^m$, $Y_n$ and $\mu$ in the adjoint
representation:
\begin{align} 
	\ip{\mu,	L(z) Y_n} &=  z^{-1}\ip{\mu,  b_m[X^m,Y_n]} = - z^{-1}  b_n , \\
	\ip{X^m, L(z) \mu} &= z^{-1}  c^m ,
\end{align}
where $\ip{-,-}$ is the Killing form. 

We will write the matrix elements of the components of $L(z)$ which
act on the space $\mf{n}_- \oplus \C \cdot \mu$, using the index $0$
to indicate $\mu$.
\begin{align}
	\ip{\mu,\mu} L_n^0(z)& = -z^{-1}  b_n , \\
	L^n_0(z) &= z^{-1} c^n , \\
	L^0_0(z)&= 1 - c^m b_m  z^{-1}  \ip{\mu,\mu}^{-1} , \\
	L^m_n &= z^{-1}\delta^m_n .
\end{align}

The quadratic Casimir has three terms, living in
$\mf{l} \otimes \mf{l}$ and $\mf{n}^{\pm} \otimes \mf{n}^{\mp}$. Only
the last two terms will enter the Poisson bracket relation. These
terms are
\begin{equation} 
	  X^m \otimes Y_m +  Y_m \otimes X^m. 
\end{equation}
Acting on the elements $Y_m$, $\mu$, we have 
\begin{equation} 
  c^{m0}_{0n} = \frac{1}{\ip{\mu,\mu}}\delta^m_n,
  \qquad
  c^{0m}_{n0} =   \frac{1}{\ip{\mu,\mu}}\delta^m_n .
\end{equation}

The Poisson bracket relation in this case becomes 
\begin{align}
  \ip{\mu,\mu} (z -z')	\{L^0_n(z), L^m_0(z')\}
  &= 	L^m_n(z) L^0_0(z') -  L^0_0(z) L^m_n(z') \\
  \delta^m_n \frac{1}{z z'} (z' - c^k b_k)
  - \delta^m_n \frac{1}{z z'} (z - c^k b_k)
  &= \delta^m_n   \frac{z'-z}{z z'}.
\end{align}
From which we conclude, 
\begin{equation} 
	\{b_n, c^m\} = \delta^m_n. 
\end{equation}

\subsection{Fundamental coweight of $SO(2n)$}

There are three Weyl orbits of minuscule coweights for $SO(2n)$.  The first is associated to an embedding
\begin{equation} 
	\mu\colon SO(2) \to SO(2n) 
\end{equation}
corresponding to a choice of an orthogonal decomposition $\C^{2n} = \C^2 \oplus \C^{2n-2}$ of the vector representation. There are $2n$ elements in the Weyl orbit of this minuscule coweight. Indeed, if we identify the maximal torus of $SO(2n)$ with $SO(2)^n$ under the obvious embedding, the minuscule coweights in this Weyl orbit are of the form $(0,\dots,0,\pm 1,0,\dots,0)$.  

The  Levi subgroup in this case is
\begin{equation} 
	\mf{l} = \mf{so}(2) \oplus \mf{so}(2n-2) \subset \mf{so}(2n). 
\end{equation}
The subalgebras $\mf{n}_+$, $\mf{n}_-$ are the elements of
$\mf{so}(2n)$ of charge $\pm 1$ under the action of $SO(2)$. Take an
orthonormal basis $x_1$, $x_2$, $y_1$, $\dotsc$, $y_{2n-2}$ of the
vector representation $\C^{2n}$. The adjoint representation is
$\wedge^2 \C^{2n}$. The subspaces $\mf{n}_{\pm}$ are each of dimension
$2n-2$ and are spanned by $x_{\pm}\wedge y_j$, $j = 1$, $\dotsc$,
$2n-2$, where $x_{\pm} = \frac{1}{\sqrt{2}}(x_1 \pm \i x_2)$.

The L-operator for a Wilson line crossing a minuscule 't Hooft line
is, according to our calculations earlier,
\begin{equation} 
	L(z)=	e^{ (x_+\wedge y_j) b_j} z^{\mu} e^{(x_- \wedge y_k) c_k }  ,
\end{equation}
where $\{b_j,c_k\} = \delta_{jk}$, $j = 1$, $\dots$, $2n-2$.

We will calculate this in the vector representation.  First note that
\begin{align}
	e^{(x_+ \wedge y_j) b_j} x_- &= x_- - y_j b_j - \tfrac{1}{2} x_+ b_j b_j, \\
  e^{(x_+ \wedge y_k) b_k} y_j &= y_j  + x_+ b_j 
\end{align}
and similarly with $x_+$ and $x_-$ switched.  From this we have
\begin{align}
  L(z) x_+ &= z x_+ - y_j c_j - x_+ b_j c_j
             - z^{-1} \tfrac{1}{2} x_- c_j c_j
             + z^{-1} \tfrac{1}{2} y_k b_k c_j c_j
             + z^{-1} \tfrac{1}{4} x_+ b_j b_j c_k c_k,
  \\ 
  L(z) x_- &= z^{-1} ( x_- - b_j y_j - \half x_+ b_j b_j),
  \\
  L(z) y_j &= y_j + x_+ b_j +  z^{-1}c_j
             \left( x_- + y_k b_k + \tfrac{1}{2} x_+ b_j b_j \right).
\end{align}
This matrix is the same as the L-operator studied in
\cite{Ferrando:2020vzk}, equation (4.1), up to a change of basis and
multiplying by an overall factor of $z$. They use a basis of the form
$x_{\pm}^1$, $\dotsc$, $x_{\pm}^r$ with $x_{\pm}^1 = x_{\pm}$,
$x_{\pm}^k = y_{2k-1} \pm \i y_{2k}$.

\subsection{Spinor coweights of $SO(2n)$}

The other two Weyl orbits of minuscule coweights for $SO(2n)$ are
associated to the two spinor nodes of the Dynkin diagram.  Choose a
basis $x_1$, $\dotsc$ ,$x_n$, $y^1$, $\dotsc$, $y^n$ of the vector
representation $\C^{2n}$ under which $\ip{x_i, y^j} = \delta_i^j$. As
usual we identify $\mf{so}(2n)$ with $\wedge^2 \C^{2n}$. A basis for
the coweight lattice is given by the expressions $x_1 \wedge y^1$,
$\dotsc$, $x_n \wedge y^n$. This is the standard basis of the maximal
torus associated to the natural embedding $SO(2)^n \subset SO(2n)$,
where each copy of $SO(2)$ rotates one of the planes spanned by $x_i$,
$y_i$.

In this basis, the spinorial minuscule coweights are  given by 
\begin{equation} 
	\sum (-1)^{k_i} \half (x_i \wedge y^i) ,
\end{equation}
where $k_i = 0$, $1$.  These are minuscule, because the basis vectors
$x_i$, $y_j$ have charge $\pm \half$ so that every element in the
adjoint representation $\wedge^2 \C^{2n}$ has charge $0$ or $\pm 1$.

Thus, there are $2^n$ minuscule coweights of this form. An element of
the Weyl group can switch an even number of the signs, so that there
are two Weyl orbits of minuscule coweights, each with $2^{n-1}$
elements.  The two Weyl orbits are related by the outer automorphism
of $SO(2n)$.

We will calculate the L-operator associated to the coweight
$\sum \half x_i \wedge y^i$, with all signs $+1$.  The L-operator for
other coweights in the same Weyl orbit can be determined by
conjugating by an element of the Weyl group, and the L-operator for
elements in the other Weyl orbit can be obtained by conjugating with
the matrix in $O(2n)$ which switches $x_1$ and $y^1$.

For the coweight $\sum \half x_i \wedge y^i$, the elements
$x_i \wedge x_j$ are of charge $+1$, $x_i \wedge y^j$ are of charge
$0$, and $y^i \wedge y^j$ are of charge $-1$. Thus, the Levi $\mf{l}$
is $\mf{gl}(n)$, and $\mf{n}_{\pm}$ are the exterior squares of the
fundamental and antifundamental representations of $\mf{gl}_n$, and
are of dimension $n \choose 2$.  As usual, we introduce oscillators
$b^{rs} \in \mf{n}_+^\vee$, $c_{rs} \in \mf{n}^\vee_-$, which are
antisymmetric in their indices and satisfy
$\{b^{rs}, c_{mn}\} = \delta^m_r \delta^n_s - \delta^m_s \delta^n_r$.

The L-operator is
\begin{equation} 
	L(z) = e^{x_r\wedge x_s b^{rs}} z^{\mu} e^{y^m \wedge y^n c_{mn}},
\end{equation}
where $z^{\mu} x_r = z^{\half} x_r$, $z^{\mu} y^r = z^{-\half} y^r$.

Thus,
\begin{align}
	L(z) y^r &= z^{-\half}(y^r +  2 b^{sr} x_s) ,\\
	L(z) x_r &=  z^{ \half } x_r +  z^{-\half} 2 c_{sr} y^s + 4z^{-\half}x_s b^{sm} c_{mr}  .
\end{align}
Multiplying by $z^{1/2}$ and re-ordering the basis elements we find
the expression computed in \cite{Ferrando:2020vzk}, equation (4.11).

\subsection{A sketch of the L-operator for the minuscule coweight of $E_6$}

So far, we given a presentation of the L-operators associated to
minuscule 't Hooft lines of $SL_n$ and $SO(2n)$, recovering known
expressions for Q-operators.  As a final example in this section, we
will present something a little new, which is the L-operator for the
minuscule coweight of $E_6$.  There is only one Weyl orbit of
minuscule coweight in this case. The Levi factor is
$\mf{l} = \mf{so}(10) \oplus \mf{so}(2)$, and the subspaces
$\mf{n}_{\pm}$ are the two spin representations of $\mf{so}(10)$. We
choose our coweight so that $\mf{n}_{\pm} = S_{\pm}$.  The $27$
dimensional representation of $E_6$ decomposes, under
$\mf{so}(10) \oplus \mf{so}(2)$, as the vector representation of
$\mf{so}(10)$ of charge $2/3$ under $\mf{so}(2)$; the spin
representation $S_+$ of charge $-1/3$; and the trivial representation
of charge $-4/3$.

Let $x_i$ be a basis for for the vector representation of
$\mf{so}(10)$, and $\psi_\alpha$, $\psi^\alpha$ bases for the spin
representations $S_{\pm}$. A basis for the $27$ dimensional
representation of $E_6$ is given by $v$, in the trivial
representation, $x_i$, and $\psi_\alpha$.  A basis for the adjoint
representation of $E_6$ is $x_i \wedge x_j$, $\Psi_\alpha$,
$\Psi^\alpha$.  The action the subspace of $E_6$ spanned by
$\Psi_\alpha$, $\Psi^\beta$ on the $27$ dimensional representation is
given by
\begin{align} 
	\Psi_\alpha \cdot v &= \psi_\alpha , \\
	\Psi_\alpha \cdot \psi_\beta &= \Gamma^i_{\alpha \beta} x_i , \\
	\Psi^\alpha \cdot x_i &= \Gamma_i^{\alpha \beta} \psi_\beta ,\\
	\Psi^\alpha \cdot \psi_\beta &= \delta^{\alpha}_{\beta} v.
\end{align}
Here $\Gamma$ is the ten dimensional $\Gamma$ matrix intertwining the
vector representation with $S_+ \otimes S_-$ (appropriately
normalized).

For the L-operator, there are oscillators $b^\alpha$, $c_\beta$
satisfying $\{b^\alpha,c_\beta\} = \delta^{\alpha}_{\beta}$. The
L-operator takes the form
\begin{equation} 
	L(z) = e^{\Psi_\alpha b^\alpha} z^{\mu} e^{\Psi^\beta c_\beta}, 
\end{equation}
where $z^{\mu} v = z^{-4/3} v$, $z^{\mu} x_i = z^{2/3} x_i$,
$z^{\mu} \psi_\alpha = z^{-4/3} \psi_\alpha$.  We will not attempt to
give a more explicit form for $L$, leaving it to the interested reader
to compute further.

We have given a uniform presentation of the semiclassical L-operator associated to a minuscule 't Hooft line (with opposite charge at $\infty$). In each case, we find that the L-operator lives in the semiclassical limit of an oscillator algebra.  These oscillator representations generalize those studied in type A by \cite{Bazhanov:2010jq} and for type D  in \cite{Ferrando:2020vzk}.

\section{'t Hooft lines and the shifted Yangian}
\label{section_shifted}

It has recently become clear \cite{Frassek:2020lky} that Q-operators
arise from representations of a certain shifted Yangian. One can view
the starting point for this development as the work of Bazhanov et al
\cite{Bazhanov:2010ts}, reviewed above, where it was shown that the
Q-operator arises from an L-operator valued in an oscillator algebra.
These authors viewed this L-operator as providing an unusual
representation of the Yangian.  This, however, is an error in
terminology: the Yangian algebra has an RLL presentation where the
L-operator is required to have leading term, as a series in $1/z$, the
identity.  As we have seen, the L-operators giving rise to Q-operators
have leading term $z^{\mu}$, where $-\mu$ is the charge of the 't
Hooft line at $\infty$.

In \cite{Frassek:2020lky} (see also \cite{Zhang:2018sej}) it was shown
that the correct interpretation of the L-operators of
\cite{Bazhanov:2010ts} is that they provide a representation of the
antidominant shifted Yangian.

Here, we will review their work and explain how the antidominant
shifted Yangian appears naturally in four-dimensional Chern-Simons
theory with an 't Hooft line at $\infty$. We will also propose a
conjectural RLL description of the antidominant shifted Yangian for
all groups except $E_8$, generalizing the description of the shifted
Yangian given in type $A$ in \cite{Frassek:2020lky} and of the
ordinary Yangian in \cite{Costello:2018gyb}.

First, let us make a comment on the terminology. In the literature
\cite{Brundan:2004ca,MR3248988}, there is a shifted Yangian associated
to any coweight $\mu$ of the Lie algebra.  Somewhat confusingly, this
construction is \emph{not} covariant under the action of the Weyl
group.  This means that the shifted Yangian associated to two
coweights in the same Weyl orbit are \emph{not} isomorphic.  This is
in contrast to the gauge-theory constructions of this paper, where
everything is always covariant under the $G$ global symmetry.

Our construction will give an algebra associated to any coweight
$\eta$, which is isomorphic to the antidominant shifted Yangian
associated to the antidominant coweight in the Weyl orbit of $\eta$.

Let us now turn to our description of the shifted Yangian.  Let $G$ be
a simple and simply-connected group, and let $\eta$ be a (possibly
fractional) coweight.  Consider four-dimensional Chern-Simons theory
with an 't Hooft line of charge $\eta$ at $z = \infty$, $y = 0$.
Consider an arbitrary collection of line defects at generic values of
$z$, again at $y = 0$.  To absorb the monopole charge at $\infty$,
these can not be pure Wilson lines, but must be Wilson-'t Hooft lines.

Let us compactify this four-dimensional system to two dimensions along
the $\CP^1$ with coordinate $z$. In the region $y < 0$ and $y > 0$,
the result is the trivial theory.  This is because we are working with
a simple group, so that the only solution to the equations of motion
is the trivial one, and all fields are infinitely massive.  Along the
line $y = 0$ we find an effective quantum mechanical system with some
algebra of operators $A$.

We will show that $A$ acquires a homomorphism from an
infinite-dimensional algebra called the antidominant shifted Yangian.

To do this, we will follow the analysis of \cite{Costello:2018gyb},
and consider a Wilson line along $x = 0$ at some value of $z$ near
$\infty$.  We must specify the representation this Wilson line lives
in.  For the classical groups $SL_n$, $SO_n$, $Sp_n$ we take the
vector representation. For $G_2$ we take the $7$-dimensional
representation, $F_4$ the $\mbf{26}$, $E_6$ the $\mbf{27}$, and $E_7$
the $\mbf{56}$. The Wilson line corresponding to each of these
representations exists at the quantum level. Further, in each case the
group can be realized as the group of invertible matrices acting on
the chosen representation which preserves certain tensors. As
explained in \cite{Costello:2018gyb}, in each case the tensors lift to
junctions of Wilson lines.

Now consider placing an 't Hooft line at $y = 0$ at $z = \infty$, a
Wilson line at $x = 0$ at $z$ near $\infty$, and some arbitrary line
defects at other points in the $z$-plane, along $y = 0$.  As above,
algebra of operators on the horizontal line defects is $A$.  A
vertical line defect, with incoming state $\langle i |$ and outgoing
state $|j \rangle$, gives an operator $L^i_j(z) \in A$ as in Figure
\ref{figure_crossing2}.

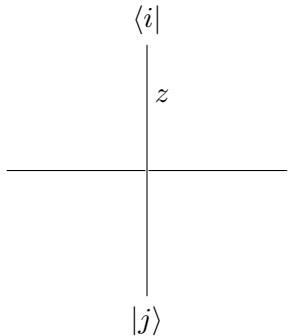
\begin{figure}[htbp]
\begin{center}
\begin{tikzpicture}
\node(In) at (0,4) {$\langle i|$};
\node(Out) at (0,0) {$|j \rangle$};
\node(HLeft) at (-2,2) {};
\node(HRight) at (2,2) {};
\draw (HLeft) to (HRight); 
\draw[very thick, color=white] (In) to (Out);
\draw (In) to (Out);  
\node at (0.2,3) {$z$};
\end{tikzpicture}
	\caption{\small{A vertical Wilson line, equipped with an incoming state $\langle i |$ and an outgoing state $|j \rangle$, gives rise to an operator $L^i_j(z) \in A$ in the quantum algebra of the  horizontal 't Hooft line. }}
\label{figure_crossing2}
\end{center}
\end{figure}

The presence of the 't Hooft line constrains the behaviour of
$L^i_j(z)$ near $z = \infty$.  In the absence of the 't Hooft line, we
have $L^i_j (z) = \delta^i_j 1_A + O( z^{-1})$.  In the presence of
the 't Hooft line at $\infty$, this is changed as follows.  Assume
that we have a basis of the representation $R$ where the Wilson line
lives in which the coweight $\eta$ of the 't Hooft line acts on the
$i$th basis element by $\eta_i$.  We ask that, near $z = \infty$,
there is an expansion of the form
\begin{equation} 
	L^i_j(z) =  \alpha^i_k(z)   \delta^k_l z^{-\eta_k} \beta^l_j(z) ,
\end{equation}
where $\alpha^i_k(z)$, $\beta^l_j(z)$ are series in $1/z$ with entries
in $A$, whose leading term is $\delta^i_j 1_A$.

We can think of this expression as saying that $L^i_j(z)$ looks, near
$\infty$, like the monopole singularity $z^{\mu}$, multiplied on
either side by ``perturbative'' contributions arising from gluon
exchange between the vertical Wilson line and whatever horizontal line
defects we have.  The gauge theory construction implies that, in
perturbation theory, this is the most general possible form of $L(z)$.
Indeed, we are clearly free to multiply on the left and the right by
the monopole singularity $z^{\mu}$, and further any small variation of
the monopole singularity $z^{\mu}$ coming from quantum effects can be
absorbed into a left or right multiplication by $\alpha(z)$,
$\beta(z)$ as above.

For example, with a coweight of charge $(\half,-\half)$ for $SL_2$,
the operator $L^i_j(z)$ will have expansion at $\infty$,
\begin{equation} 
	L(z) = G(z)\begin{pmatrix}
		z + l^1_1[-1] + l^1_1[0] z^{-1} + \dots & l^1_2[0] z^{-1} + \dots \\
	l^{2}_1[0] z^{-1} + \dots & 1 + l^2_2[1] z^{-2} + \dots 
	\end{pmatrix}	 ,
\end{equation}
where the overall normalizing factor $G(z)$ will be determined by the
requirement that the quantum determinant is one, as we will see later.
Thus, the generators of the shifted Yangian algebra we find are
$l^1_1[k]$, $k \ge -1$; $l^2_2[k]$, $k \ge 1$; and $l^1_2[k]$,
$l^2_1[k]$ for $k \ge 0$.

The gauge-theory construction implies that the generating function
$L(z)$ satisfies the RLL relation, as we have already discussed
semiclassically (see Figure \ref{figure_bent}). As explained in this
diagram, this is a consequence of the fact that Wilson lines crossing
above or below the 't Hooft line give the same elements of the algebra
$L$. The algebraic form of the relation is
\begin{equation}
\sum_{r,s} R^{i k}_{r s}(z - z') L^r_j(z') L^s_l(z) = \sum_{r,s} L^i_r(z) L^k_s(z')R^{rs}_{jl}(z -z') .
\end{equation}

It was shown in \cite{Frassek:2020lky} that, with gauge group $GL_n$,
this RLL relation, together with the boundary condition on the
behaviour of $L(z)$ at $\infty$, gives rise to the shifted Yangian for
$\mf{gl}_n$.

\subsection{Extra relations from other groups}
To understand the algebra we get for other groups, we need to
introduce extra relations that come from junctions of Wilson lines
\cite{Costello:2018gyb}. Suppose we have $k$ Wilson lines in the gauge
theory with gauge Lie algebra $\mf{g}$, in representations $V_1$,
$\dotsc$, $V_k$.  Suppose that there is a $\mf{g}$-invariant element
of the tensor product $V_1 \otimes \dots \otimes V_k$.  Classically,
this provides a junction between the Wilson lines, when they are all
placed at the same value of the spectral parameter.  For example, if
we have $n$ copies of the fundamental representation of $\mf{sl}_n$,
then the determinant defines an invariant vector in $V^{\otimes n}$
and so a (classical) junction between Wilson lines.

In \cite{Costello:2018gyb}, we analyzed when we can lift these
junctions to the quantum level.  The general story is involved: some
times you can, and sometimes there is an anomaly.  When the junctions
can be lifted to the quantum level, the spectral parameters of the
various lines have to have certain very special shifts.  For example,
for the determinant junction between $n$ fundamental Wilson lines of
$\mf{sl}_n$, there is a lift to a junction at the quantum level
provided that the spectral parameters of the $n$ lines are
$z$, $z+2\hbar$, $\dotsc$, $z + 2(n-1)\hbar$.

Consider a simple group presented as a matrix group preserving some
additional invariant tensors. For example, we can present $SO(n)$ as
the group of automorphisms of $\C^n$ preserving a symmetric pairing,
or $G_2$ as the group of automorphisms of $\C^7$ preserving an
antisymmetric cubic tensor and a non-degenerate symmetric pairing.
We can attempt to lift this presentation of the group to the the gauge
theory setting, by lifting the invariant tensors defining the group to
junctions between the defining Wilson line.

In \cite{Costello:2018gyb}, we succeeded in doing this for all groups
except $E_8$, with particular shifts.  The results are as follows
(here $\sh^\vee$ is the dual Coxeter number).
\begin{enumerate} 
\item For type $A$, the determinant junction lifts with the spectral
  parameters $z$, $z + 2\hbar$, $\dotsc$, $z + 2(n-1) \hbar$.
  
\item For $\mf{so}_n$, $\mf{sp}_n$ the junction defined by the
  symmetric or antisymmetric pairing lifts to the quantum level with
  spectral parameters $z$, $z + \hbar \sh^\vee$.
  
\item For $G_2$, the group is defined by preserving the quadratic and
  cubic invariant tensors on $\C^7$. The corresponding junctions exist
  at the quantum level with shifts $z$, $z + \hbar \sh^\vee$ for the
  quadratic vertex and $z$, $z+\frac{2}{3} \hbar \sh^\vee$,
  $z + \frac{4}{3} \hbar \sh^\vee$ for the cubic vertex.

\item The group $F_4$ is the automorphisms of $\C^{26}$ preserving a
  symmetric pairing and a cubic tensor. The vertices lift to the
  quantum level with spectral parameters $z$, $z + \hbar \sh^\vee$ for
  the quadratic vertex and $z$, $z+\frac{2}{3} \hbar \sh^\vee$,
  $z + \frac{4}{3} \hbar \sh^\vee$ for the cubic vertex.

\item For $E_6$, the group is the automorphisms of $\C^{27}$ which
  preserve a cubic tensor.  This vertex lifts to the quantum level
  with spectral parameters $z$, $z+\frac{2}{3} \hbar \sh^\vee$,
  $z + \frac{4}{3} \hbar \sh^\vee$ for the cubic vertex.

\item For $E_7$, the group is the automorphisms of $\C^{56}$
  preserving a symplectic pairing and a quartic tensor.  These
  vertices lift to the quantum level with the shifts $z$,
  $z + \hbar \sh^\vee$ for the quadratic vertex and $z$,
  $z+\frac{1}{2} \hbar \sh^\vee$, $z + \hbar \sh^\vee$,
  $z + \frac{3}{2} \hbar \sh^\vee$ for the quartic vertex.
\end{enumerate}
(We do not have a similar construction for $E_8$, because the smallest
representation of $E_8$, which is the adjoint, does not lift to a
quantum Wilson line).

The existence of the vertex between Wilson lines implies extra
relations, beyond the RLL relation, on the universal algebra that acts
on a line defect when we have an 't Hooft line at $\infty$. The extra
relation is sketched in Figure \ref{figure_identity}: we can
contemplate moving an 't Hooft line past a network of Wilson
lines. This will have no effect, by topological invariance of the
theory in the $x$-$y$ plane, leading to a relation among the algebra
elements in the series $L^i_j(z)$.
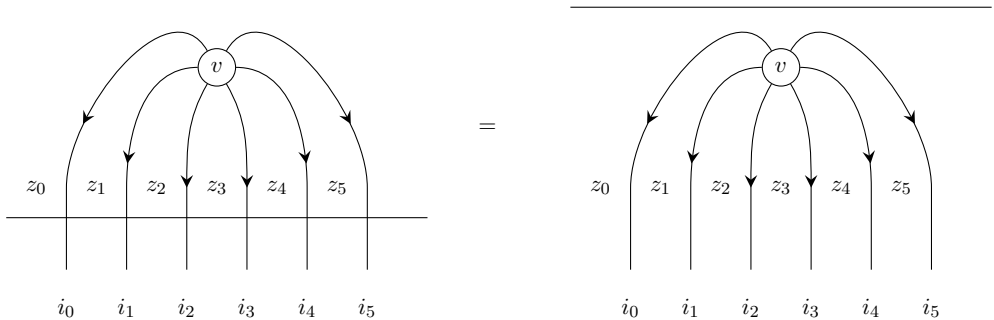
\begin{figure}
\begin{center}
\begin{tikzpicture}
\begin{scope}[scale=0.8, every node/.style={transform shape}]
\node[draw, circle] (central) at (0:0) {$v$};
\node at (-3,-2) {$z_0$};
\node at (-2,-2) {$z_1$};
\node at (-1,-2) {$z_2$};
\node at (0,-2) {$z_3$};
\node at (1,-2) {$z_4$};
\node at (2,-2) {$z_5$};
\node at (-2.5,-4) {$i_0$};
\node at (-1.5,-4) {$i_1$};
\node at (-0.5,-4) {$i_2$};
\node at (0.5,-4) {$i_3$};
\node at (1.5,-4) {$i_4$};
\node at (2.5,-4) {$i_5$};
\node(N1) at (-2.5,-3.5){};
\node(N2) at (-1.5,-3.5){}; 
 \node(N3) at (-0.5,-3.5){}; 
\node(N4) at (0.5,-3.5){};
\node(N5) at (1.5,-3.5){};
\node(N6) at (2.5,-3.5){};
\draw[-<-] (N1) to (-2.5,-2) to [out=90,in=120] (central);
\draw[-<-] (N2) to (-1.5,-2) to [out=90,in=180] (central);
\draw[-<-] (N3) to (-0.5,-2) to [out=90,in=240] (central);
\draw[-<-] (N4) to (0.5,-2) to [out=90,in=300] (central);
\draw[-<-] (N5) to (1.5,-2) to [out=90,in=0] (central);
\draw[-<-] (N6) to (2.5,-2) to [out=90,in=60] (central);
\draw (-3.5, -2.5) to (3.5,-2.5);
\node at (4.5,-1) {$=$};
\end{scope}

\begin{scope}[shift={(7.5,0)},scale=0.8,  every node/.style={transform shape}]
\node[draw, circle] (central) at (0:0) {$v$};
\node at (-3,-2) {$z_0$};
\node at (-2,-2) {$z_1$};
\node at (-1,-2) {$z_2$};
\node at (0,-2) {$z_3$};
\node at (1,-2) {$z_4$};
\node at (2,-2) {$z_5$};
\node at (-2.5,-4) {$i_0$};
\node at (-1.5,-4) {$i_1$};
\node at (-0.5,-4) {$i_2$};
\node at (0.5,-4) {$i_3$};
\node at (1.5,-4) {$i_4$};
\node at (2.5,-4) {$i_5$};
\node(N1) at (-2.5,-3.5){};
\node(N2) at (-1.5,-3.5){}; 
 \node(N3) at (-0.5,-3.5){}; 
\node(N4) at (0.5,-3.5){};
\node(N5) at (1.5,-3.5){};
\node(N6) at (2.5,-3.5){};
\draw[-<-] (N1) to (-2.5,-2) to [out=90,in=120] (central);
\draw[-<-] (N2) to (-1.5,-2) to [out=90,in=180] (central);
\draw[-<-] (N3) to (-0.5,-2) to [out=90,in=240] (central);
\draw[-<-] (N4) to (0.5,-2) to [out=90,in=300] (central);
\draw[-<-] (N5) to (1.5,-2) to [out=90,in=0] (central);
\draw[-<-] (N6) to (2.5,-2) to [out=90,in=60] (central);
\draw (-3.5, 1) to (3.5,1);
\end{scope}
\end{tikzpicture}
\end{center}
\caption{\small{Topological invariance allows us to move the position
    of the horizontal 't Hooft line line defects past the collection
    of Wilson lines attached to the vertex without effecting the
    result\label{figure_identity} }}
\end{figure}

For example, for $\mf{sl}_n$, the relation is the quantum determinant relation 
\begin{equation} 
	 \sum_{k_r} \op{Alt}(k_0,\dots,k_{N-1}) L^{k_0}_{0}(z ) L^{k_1}_{1}(z + 2 \hbar  )\cdots L^{k_{N-1}}_{N-1}(z + 2(N-1) \hbar ) = 1. \label{equation_qdet}  
\end{equation}
For $\mf{so}_n$ or $\mf{sp}_{2n}$,  the extra relation we get is 
\begin{equation} 
	\omega_{kl}	L^k_i(z) L^l_j(z + \hbar \sh^\vee ) = \omega_{ij} ,
\end{equation}
where $\omega_{ij}$ is the invariant pairing on the vector representation (symmetric or antisymmetric as appropriate).

For the other groups, we have a similar presentation. For instance, for $E_7$, we have one extra relation which is 
\begin{equation} 
	\Omega_{ijkl} L^i_m (z) L^j_n(z + \half \hbar \sh^vee) L^k_o (z + \hbar \sh^\vee) L^l_p (z + \tfrac{3}{2} \hbar \sh^\vee ) = \Omega_{mnop}. 
\end{equation}

In each case, the RLL relation,  the extra relations coming from vertices between Wilson lines, and the boundary condition at $z = \infty$, define an associative algebra $Y^{\mu}(\g)$.
\begin{conjecture}
	The algebra $Y^{\mu}(\g)$ is isomorphic to the antidominant shifted Yangian. 
\end{conjecture}
At the classical level, this holds because of the results of
\cite{Kamnitzer2020}.  In general, a sufficiently strong uniqueness
theorem for the shifted Yangian as a quantization of its classical
limit (compatible with additional structures such as coproducts) will
prove that our algebra $Y^{\mu}(g)$ is isomorphic to the shifted
Yangian, but such a result seems not to be currently available.

\subsection{Coproducts}
It is known \cite{MR3761996} that the shifted Yangian admits left and
right coproducts relating it to the ordinary Yangian:
\begin{equation}
	\begin{split}
		\tr_L\colon Y^{\mu}(\g) &\to Y(\g)\otimes Y^{\mu}(\g), \\
		\tr_R\colon Y^{\mu}(\g) & \to Y^{\mu}(\g) \otimes Y(\g).
	\end{split}
\end{equation}
The left and right coproducts commute with each other and are co-associative, making the category of $Y^{\mu}(\g)$-modules into a bimodule category over the category of $Y(\g)$-modules.  

From the field theory perspective, the category of $Y^{\mu}(\g)$
modules is the category of line defects\footnote{At various points we
  need to distinguish between analytically-continued line defects,
  where the one-dimensional system has an algebra of operators but not
  a Hilbert space; and full line defects, where there is a Hilbert
  space too. Here for simplicity we are discussing full line defects.}
in the bulk with a parallel 't Hooft line of charge $\mu$ at $\infty$.

The coproduct tells us that we can fuse an ordinary line defect coming
from the left or the right with a line defect parallel to the 't Hooft
line at $\infty$.

From our definition of $Y^{\mu}(\g)$ and the presentation of $Y(\g)$
presented in \cite{Costello:2018gyb}, the left and right coproducts
are easy to define.  We present $Y(\g)$ as being generated by the
coefficients of an L-operator $L^0(z)$ which near $\infty$ goes like
$1 + O(1/z)$.  These are subject to the RLL relation, as well as the
extra relations coming from junctions between line defects.

Similarly, $Y^{\mu}(\g)$ is presented as above as generated by the
coefficients of $L^{\mu}(z)$ subject to the RLL relation as well as
the relations arising from junctions of Wilson lines.

To give a coproduct
\begin{equation} 
	\tr_R\colon Y^{\mu}(\g) \to Y^{\mu}(\g) \otimes Y(\g)  
\end{equation}
we need to give an L-operator
$\tr_R^{\mu}(z) \in Y^{\mu}(g)\otimes Y(\g)$ satisfying the relations
defining $Y^{\mu}(\g)$.  Since $Y(\g)$ and $Y^{\mu}(\g)$ are defined
in terms of the coefficients of the L-operators $L^0(z)$, $L^{\mu}(z)$
we must write $\tr_R L^{\mu}(z)$ in terms of $L^0(z)$ and
$L^{\mu}(z)$.

We calculate $\tr_R L^{\mu}(z)$ by considering a vertical Wilson line passing two horizontal line defects representing representations of $Y(\g)$ and $Y^{\mu}(\g)$, as in Figure \ref{figure_coproduct}.

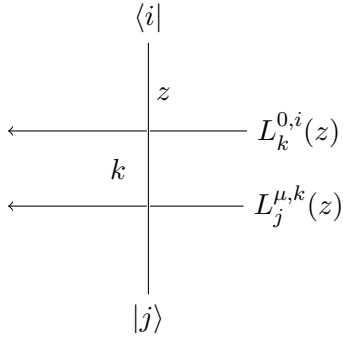
\begin{figure}[htbp]
\begin{center}
\begin{tikzpicture}
\node(In) at (0,4) {$\langle i|$};
\node(Out) at (0,0) {$|j \rangle$};
	\node(mid) at (-0.4,2){$k$};
\node(HLeft) at (-2,2.5) {};
	\node(HRight) at (2,2.5){$L^{0,i}_k(z)$}; 
\node(HLeft2) at (-2,1.5) {};
	\node(HRight2) at (2,1.5) {$L^{\mu,k}_j(z)$};
\draw [<-] (HLeft) to (HRight);
 \draw [<-](HLeft2) to (HRight2); 
\draw[very thick, color=white] (In) to (Out);
\draw (In) to (Out);  
\node at (0.2,3) {$z$};
\end{tikzpicture}
\caption{\small{The configuration of line defects associated to the coproduct. Here $k$ is the intermediate state.}}
\label{figure_coproduct}
\end{center}
\end{figure}

From this diagram, it is clear that the coproduct must be given by
\begin{equation} 
	\tr_R  L^{\mu,i}_j(z)  = L^{\mu,i}_k(z) L^{0,k}_j (z) ,
\end{equation}
where we treat both sides as series in $1/z$ and identify coefficients. 

In order to show that this defines a homomorphism from $Y^{\mu}(\g)$ to $Y^{\mu}(\g) \otimes Y(\g)$, we need to show that $\tr_R L^{\mu}(z)$ as defined above satisfies the relations defining $Y^{\mu}(\g)$. That is, $\tr_R L^{\mu}(z)$ must satisfy the RLL relation, the boundary behaviour at $z = \infty$, and the extra relations coming from vertices between Wilson lines, assuming that $L^0(z)$ and $L^{\mu}(z)$ do. 

This is not hard to check from the field theory picture. For example, in Figure \ref{figure_coproduct_relation} we present the diagram indicating why the coproduct $\tr_R L^{\mu}(z)$ respects the relations coming from vertices between Wilson lines, assuming that $L^0(z)$ and $L^{\mu}(z)$ do. 

\begin{figure}
\begin{center}
\begin{tikzpicture}
\begin{scope}[scale=0.5, every node/.style={transform shape}]
\node[draw, circle] (central) at (0:0) {};
\node(N1) at (-2.5,-3.5){};
\node(N2) at (-1.5,-3.5){}; 
 \node(N3) at (-0.5,-3.5){}; 
\node(N4) at (0.5,-3.5){};
\node(N5) at (1.5,-3.5){};
\node(N6) at (2.5,-3.5){};
\draw  (N1) to (-2.5,-2) to [out=90,in=120] (central);
\draw  (N2) to (-1.5,-2) to [out=90,in=180] (central);
\draw  (N3) to (-0.5,-2) to [out=90,in=240] (central);
\draw  (N4) to (0.5,-2) to [out=90,in=300] (central);
\draw  (N5) to (1.5,-2) to [out=90,in=0] (central);
\draw  (N6) to (2.5,-2) to [out=90,in=60] (central);
\draw (-3.5, -2.5) to (3.5,-2.5);
\draw (-3.5, -2) to (3.5,-2);
\node at (4.5,-1) {$=$};
\end{scope}

\begin{scope}[shift={(5,0)},scale=0.5,  every node/.style={transform shape}]
\node[draw, circle] (central) at (0:0) {};
\node(N1) at (-2.5,-3.5){};
\node(N2) at (-1.5,-3.5){}; 
 \node(N3) at (-0.5,-3.5){}; 
\node(N4) at (0.5,-3.5){};
\node(N5) at (1.5,-3.5){};
\node(N6) at (2.5,-3.5){};
\draw  (N1) to (-2.5,-2) to [out=90,in=120] (central);
\draw  (N2) to (-1.5,-2) to [out=90,in=180] (central);
\draw  (N3) to (-0.5,-2) to [out=90,in=240] (central);
\draw  (N4) to (0.5,-2) to [out=90,in=300] (central);
\draw  (N5) to (1.5,-2) to [out=90,in=0] (central);
\draw  (N6) to (2.5,-2) to [out=90,in=60] (central);
\draw (-3.5, 2) to (3.5,2);
	\draw (-3.5, -2.5) to (3.5,-2.5);
\node at (4.5,-1) {$=$};
\end{scope}

\begin{scope}[shift={(10,0)},scale=0.5,  every node/.style={transform shape}]
	\node[draw, circle] (central) at (0:0) {};
\node(N1) at (-2.5,-3.5){};
\node(N2) at (-1.5,-3.5){}; 
 \node(N3) at (-0.5,-3.5){}; 
\node(N4) at (0.5,-3.5){};
\node(N5) at (1.5,-3.5){};
\node(N6) at (2.5,-3.5){};
\draw  (N1) to (-2.5,-2) to [out=90,in=120] (central);
\draw  (N2) to (-1.5,-2) to [out=90,in=180] (central);
\draw  (N3) to (-0.5,-2) to [out=90,in=240] (central);
\draw  (N4) to (0.5,-2) to [out=90,in=300] (central);
\draw  (N5) to (1.5,-2) to [out=90,in=0] (central);
\draw  (N6) to (2.5,-2) to [out=90,in=60] (central);
\draw (-3.5, 2) to (3.5,2);
	\draw (-3.5, 1) to (3.5,1);
\end{scope}

\end{tikzpicture}
\end{center}
\caption{\small{Here, we see that the relations in the algebra coming from vertices joining Wilson lines are automatically compatible with the coproduct, by topological invariance of the system.  Topological invariance allows us to move the position of the horizontal 't Hooft line  line defects past the collection of Wilson lines attached to the vertex without effecting the result\label{figure_coproduct_relation} }}
\end{figure}
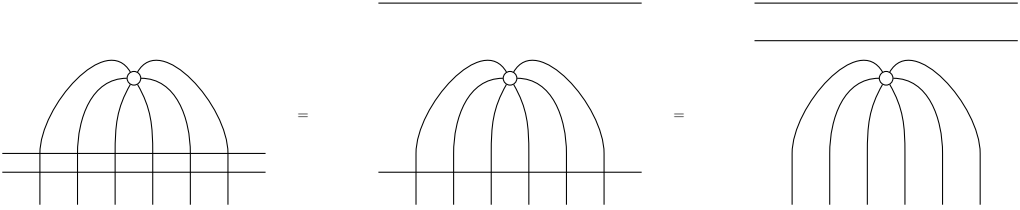

\subsection{More general coproducts}
We can, of course, study more general coproducts, by trying a formula
\begin{equation} 
	\tr L^{\mu+\mu'}(z) \overset{?}{=} L^{\mu}(z) L^{\mu'}(z). 
\end{equation}
It is not at all obvious that something like this will work, because
the expression on the right hand side may not have the correct
behaviour at $z = \infty$. However, we can check that if $\mu$, $\mu'$
both have the property that $\ip{\mu,\alpha} \ge 0$ ,
$\ip{\mu',\alpha } \ge 0$ for all positive roots $\alpha$, then
$\tr L^{\mu + \mu'}(z)$ as defined by this equation has the behaviour
at $z = \infty$ required of an L-operator of charge $\mu + \mu'$.

To see this, we note that we can write
\begin{equation} 
	L^{\mu}(z) L^{\mu'}(z) = A(z) z^{\mu} B(z) z^{\mu'} C(z) ,
\end{equation}
where $A$, $B$, $C$ are matrices acting on the representation where
our Wilson line lives which look like $1 +O(1/z)$. We can write
\begin{equation} 
  B(z) = e^{\sum_{\alpha< 0} b_\alpha(z)} e^{ b_0(z)}  e^{ \sum_{\alpha> 0} b_\alpha(z)}  ,
\end{equation}
where $b_{\alpha}(z)$ is in the $\alpha$ root space, and $b_0(z)$ is
in the Cartan.

Thus, 
\begin{equation} 
  L^{\mu}(z) L^{\mu'}(z) = A(z) e^{\sum_{\alpha < 0} z^{\ip{\mu,\alpha}} b_\alpha(z) } z^{\mu} z^{\mu'} e^{b_0(z)} e^{-\sum_{\alpha > 0} z^{\ip{\mu',\alpha} } b_\alpha(z) }  C(z) .
\end{equation}
We have $\ip{\mu,\alpha} \le 0$ if $\alpha < 0$ and
$\ip{\mu',\alpha} \ge 0$ if $\alpha > 0$.  Therefore
$A(z) e^{\sum_{\alpha < 0} z^{\ip{\mu,\alpha}} b_\alpha(z) }$ and
$e^{b_0(z)} e^{-\sum_{\alpha > 0} z^{\ip{\mu',\alpha} } b_\alpha(z) }
C(z) $ are both of the form $1 + O(1/z)$, so that
$L^{\mu}(z) L^{\mu'}(z)$ is of the form required to define an
L-operator of charge $\mu + \mu'$.

As in the case when $\mu' = 0$, it is easy to see that
$\tr L^{\mu + \mu'}(z)$ satisfies the RLL relation and the relations
coming from vertices between Wilson lines.

One can ask if this coproduct matches the ``standard'' one on
antidominant shifted Yangians \cite{Brundan:2004ca, MR3761996,
  MR3248988}.  For type A, this was shown in \cite{Frassek:2020lky},
but it appears to be unknown in general (although surely true).

\section{'t Hooft operators and Coulomb branches of three-dimensional
  $\mathcal{N}=4$ gauge theories}
\label{sec:coulomb}

In section \ref{sec:phasespace}, we have discussed the classical phase
space for four-dimensional Chern-Simons theory in the presence of an
't Hooft operator of charge $\mu$ at $0$ and $\eta$ at infinity. We
described the phase space as the space of element in $G((z))$ which
satisfy certain properties at $0$ and $\infty$.

The quantum version of such a description would be to describe the
operator algebra $A_\mu^\eta$ for the quantum mechanical system
arising from the quantization of the phase space, as well as an
algebra morphism $Y^\eta \to A_\mu^\eta$ describing the effective
coupling between the system and the internal degrees of freedom of a
generic transverse Wilson line.

In this section, we will propose an answer to this question by a
somewhat circuitous route. It remains to show that this answer
actually coincides with the result of a perturbative calculation in
the presence of the 't Hooft line, but we show that the answer passes
important consistency checks.

The proposal is simple: the phase space coincides with the Coulomb
branch of vacua of certain three-dimensional ${\cal N}=4$ ADE quiver gauge theories,
which admit a natural quantization via $\Omega$-deformation. This
deformation gives a ``quantum Coulomb branch algebra'' which is known
to admit an algebra morphism from $Y^\eta$, which is nicely compatible
with coproducts in a way we will review \cite{Braverman:2016pwk}. It
is our candidate for $A_\mu^\eta$.

It is worth pointing out that the existence of the algebra morphism
from $Y^\eta$ and the compatibility with the coproducts are {\it not}
obvious from the Coulomb branch description of the algebra and appear
rather miraculous. A proper justification of our proposal would thus
explain the appearance of these surprising extra structures.

Recall that four-dimensional Chern-Simons theory can be embedded into
a physical gauge theory in several different ways, related by string
dualities \cite{Costello:2018txb, Ashwinkumar:2018tmm,
  Nekrasov:String-Math-2017}.  For our purposes, it is most useful to
focus on the treatment of \cite{Costello:2018txb}, where the theory is
identified with an $\Omega$-deformation of six-dimensional
${\cal N}=(1,1)$ super Yang-Mills theory with gauge algebra
$\mathfrak{g}$. In that context, 't Hooft lines in the
four-dimensional Chern-Simons theory are naturally the image of
three-dimensional BPS 't Hooft defects in the physical theory. Such
defects preserve a three-dimensional ${\cal N}=4$ supersymmetry
subalgebra, and the $\Omega$-deformation of the ambient theory is
accompanied by a $\Omega$-deformation of the defect theory.

The six-dimensional ${\cal N}=(1,1)$ super Yang-Mills theory is IR
free and requires a UV completion. An 't Hooft defect has a somewhat
hybrid status: although it is defined using the IR free degrees of
freedom, the singularities associated to monopole bubbling may signal
the need for an extra UV input. The situation is somewhat analogous to
the treatment of instanton particles in five-dimensional maximally
supersymmetric super Yang-Mills theory: although much of the phase
space consists of semiclassical solitons described by instanton
solutions of the IR free gauge theory, a proper treatment of the
singularities associated to zero-size instantons require some UV
input.

The analogy becomes even sharper if we add a couple of directions to
space-time, and discuss instanton membranes in seven-dimensional ADE
super Yang-Mills theory. If the gauge theory is realized at an ADE
singularity in M-theory, then small instantons are M2-branes which can
leave the singularity along the ADE direction. This UV process is not
visible in the naive IR gauge theory description, but it is fully
captured by adding the information of the low-energy supersymmetric
quantum field theory which lives on zero-size instantons: a
three-dimensional ${\cal N}=4$ affine ADE quiver gauge theory with no
flavours. This theory has a Coulomb branch which reproduces the smooth
instanton moduli space, as well as a Higgs branch which describes the
motion of M2-branes away from the singularity.

In the case of 't Hooft defects in six-dimensional ${\cal N}=(1,1)$ super
Yang-Mills theory, we can ask for a low energy description of an 't
Hooft defect of charge $\mu$ ``covered'' by bubbling gauge
configurations leading to a charge $\eta$ at infinity. We want to
claim that this will be a three-dimensional ${\cal N}=4$ ADE quiver
gauge theory whose Coulomb branch reproduces the space of
supersymmetric gauge configurations with these charges, i.e.\ the
Bogomolny moduli space which coincides with the four-dimensional
Chern-Simons phase space discussed before. The Higgs branch of the
theory, which is only present when the Coulomb branch has
singularities, will describe dynamical processes which are invisible
in the six-dimensional super Yang-Mills description.

The ADE quiver gauge theory is characterized by two collections of
non-negative integers: the ranks $N_i$ of the $U(N_i)$ gauge groups at
each node and the numbers $M_i$ of fundamental flavours attached to
each node.  The $M_i$ encode the charge $\mu = \sum_i M_i w_i$ of the
't Hooft operator.  The Weyl orbit of the charge at infinity $\eta$ is
encoded by the {\it imbalance} at each node
\begin{equation}
  \delta_i = \sum_j C_{ij} N_j - M_i,
\end{equation}
where $C_{ij}$ is the Cartan matrix of the ADE diagram. The correct
phase space is reproduced only if the imbalance at each node is
non-negative.\footnote{In order for the quantized Coulomb branch to
  admit a trace and/or interesting weight modules, the imbalance
  should not be too positive. We will come back to this later.}

We should also mention that the theory has a collection of mass
deformation parameters, one for each individual flavour. Turning on
the mass parameters corresponds to fragmenting the 't Hooft defect
into a collection of defects of smaller magnetic charge.  By sending a
mass parameter at infinity one can send an 't Hooft line to infinity
as well.

When we turn on the $\Omega$-deformation, the masses are redefined by
amounts proportional to $\hbar$ and the notion of ``mass equal to
$0$'' is ambiguous. In a sense, once we quantize the system the 't
Hooft operators are always fragmented into pieces of minimal charge.

We will not attempt to give a full derivation for this effective
worldvolume theory. There are several equivalent UV completions of six-dimensional
super Yang-Mills theory in string theory. 't Hooft operators for some
basic charges can be engineered with the help of extra branes, but the
details are cumbersome. Obtaining more general charges may require
further manipulations, possibly done after restricting the system to
an intermediate UV completion such as ${\cal N}=(1,1)$ little string
theory \cite{Nekrasov:String-Math-2017}.

Notice that this type of analysis is subtle even for Wilson
lines. Naively, supersymmetric Wilson lines exist in six-dimensional
super Yang-Mills theory for any representation, but we know that only
those that can be lifted to Yangian representations are available in
four-dimensional Chern-Simons theory. This must mean that the ``bad''
lines, when placed in the $\Omega$-deformation, do not admit
counterterms which are required to preserve supersymmetry quantum
mechanically.  In any case, they must be absent in the UV
completion. This is challenging to verify directly.

%\DG{Which one? Both? Is the question sensible?}
%\KC{The bad lines are supersymmetric on flat space but not in the omega background} 

From this point of view, the fact that the $\Omega$-deformation of
``antidominant'' three-dimensional ${\cal N}=4$ ADE quiver gauge theories always
yields a gauge-invariant line defect in four-dimensional Chern-Simons
theory is rather striking.

\subsection{ADE Quantum Coulomb branches}
An $\Omega$-deformation reduces a three-dimensional ${\cal N}=4$ gauge
theory to a one-dimensional system governed by an operator algebra
denoted as the quantum Coulomb branch algebra. Mathematically, the
quantum Coulomb branch algebra is given by the BFN construction
\cite{Braverman:2016wma}.  The algebra is a quantization of the
algebra of holomorphic functions on the Coulomb branch of the
three-dimensional ${\cal N}=4$ gauge theory.  Besides the
quantization/$\Omega$-deformation parameter $\hbar$, the algebra
depends on certain equivariant parameters (aka masses), which become
the position of 't Hooft lines in the $z$-plane in the
four-dimensional Chern-Simons picture.

The quantum Coulomb branch algebra has a complicated mathematical
definition, which encodes the properties of generic BPS monopole
operators in the three-dimensional gauge theory. Some elements in the algebra have a
particularly simple definition:
\begin{itemize}
\item Some infinite collection of $H^{(n)}_i$ generators built from the vectormultiplet scalars, with no monopole charge.
\item Some infinite collection of $E^{(n)}_i$ generators of minimal positive monopole charge at the $i$-th node.
\item Some infinite collection of $F^{(n)}_i$ generators of minimal negative monopole charge at the $i$-th node.
\end{itemize}
For the ADE quivers relevant to the description of the 't Hooft
defects, these generators are known to satisfy the commutation
relations for the antidominant shifted Yangian $Y(\mathfrak{g})^\eta$
\cite{Braverman:2016pwk} and thus give an algebra morphism from
$Y(\mathfrak{g})^\eta$ to the quantum Coulomb branch algebra.

This morphism is actually surjective, so that the quantum Coulomb
branch algebra is a truncation of the Yangian.  The specific form of
the truncation depends on the mass parameters $m_i$.  Alternatively,
one can present the Yangian as a quantization of the
Beilinson-Drinfeld affine Grassmannian associated to points $m_i$ in
$\mathbb{C}$, and truncate that in a more canonical way.

The corresponding classical statement is that there are holomorphic
functions $H^{(n)}_i$, $E^{(n)}_i$, $F^{(n)}_i$, together with other
functions defined by their Poisson brackets, which can be assembled
together into a formal generating series $g(z)$ which satisfies the
classical limit of the shifted Yangian commutation relations and lies
in a certain submanifold determined by the $N_i$ and $m_i$ parameters.

This is precisely the phase space discussed in section
\ref{section_phase}. In the language we used in that section, $g(z)$
represents the monodromy of a Wilson line passing elementary 't Hooft
operators at positions $z_i = m_i$ and charges $\rho_i$ determined by
the node of the quiver the flavour is attached at.

Quantum mechanically, the analogue of the classical holonomy $g(z)$
evaluated in some representation $r$ is the collection of R-matrices
$g_{R}(z)$ between a generic representation of the Yangian and a
finite-dimensional representation $R$ with spectral parameter $z$.

\subsection{Coproduct}

The coproduct for the Yangian is a quantization of the simple
classical relation
\begin{equation}
  \tr g(z)=g_1(z)\otimes g_2(z) ,
\end{equation}
where we are using the group composition law. As we have seen in
section \ref{section_shifted}, this expression defines the coproduct
for the (shifted or unshifted) Yangian at the quantum level as well,
as long as we treat $g(z)$ as a matrix acting in an appropriate
representation.

In the context of Bogomolny equations, we can consider a situation
where the collection of Dirac singularities is split into two
sub-collections whose charge still lies in the root lattice. If we
separate the two collections by a large amount in the $x$-direction,
we can produce a BPS solution for the composite system by combining
BPS solutions for the subcollections.  This leads precisely to the
above coproduct.

This assertion is known to hold quantum mechanically for the quantized
Coulomb branches: there is a map from the quantized Coulomb branch for
ranks $N_i + N_i'$ and flavours $M_i + M'_i$ to the tensor product of
the quantized Coulomb branches for ranks $N_i$ and $N_i'$ and flavours
$M_i$ and $M'_i$. This map is compatible with the maps from the
Yangian and the Yangian coproduct~\cite{MR3761996}. With some abuse of
notation, we will denote this map between quantized Coulomb branches
as a coproduct as well.

Because the Yangian coproduct arises in four-dimensional Chern-Simons
perturbation theory precisely from the topological fusion of line
defects, this assertion is an implicit test of the duality between the
order and disorder definitions of the defects: it is compatible with
the topological fusion of line defects.

\subsection{$A_1$ examples}

In the case of $\mathfrak{sl}_2$, the quiver has one node and we have
three-dimensional SQCD, with $U(N)$ gauge group and $M$ flavours.
This corresponds to $\mu = M \sigma_3$ and $\eta = (2N-M) \sigma_3$.

\subsubsection{A single 't Hooft operator}

Placing a minimal 't Hooft charge at $\infty$, we have $N=M=1$: SQED
with one flavour. This theory has no Higgs branch, and a Coulomb
branch isomorphic to $\mathbb{C}^2$. The quantum Coulomb branch
algebra is a Weyl algebra. The map from the shifted Yangian is given
by the familiar oscillator representation
\begin{equation} 
	L(z) = \begin{pmatrix}
		z + bc & b \\
		c & 1 
		\end{pmatrix} 	
\end{equation}
as expected.

\subsubsection{Two 't Hooft operators, zero charge at infinity}
Next, we can look at $N=1$, $M=2$: SQED with two flavours. Both Higgs and Coulomb branches of this theory 
take the form of $A_1$ singularities. The quantum Coulomb branch algebra is a central quotient of $U(\mathfrak{sl}_2)$. 
The map from the Yangian is given by the familiar representation 
\begin{equation} 
	\til{L}(z) =  \begin{pmatrix}
		z + \hbar \tfrac{1}{2} \rho(h) &\hbar  \rho(f) \\
		\hbar \rho(e) & z - \hbar \tfrac{1}{2} \rho(h)
	\end{pmatrix}
\end{equation}
appropriately normalized.

 We thus recognize that a ``Verma module'' Wilson line is the same as two 't Hooft lines, separated by a distance $\hbar (2j+1)$.
 This statement is clearly close to the QQ relation, but it is subtly different: the charges at infinity have already been cancelled 
 against each other. 

 \subsubsection{Two 't Hooft operators, non-zero charge at infinity}
Finally, we can look at the case $N=2$, $M=2$. The map from the shifted Yangian is given by the representation 
\begin{equation} 
	L(z) = \begin{pmatrix}
		z^2  + a_1 z + a_0 & b_1 z + b_0\\
		c_1 z + c_0 & d_0
		\end{pmatrix} 	.
\end{equation}
Classically, we want to impose $\det L(z) = z^2$, or $\det L(z) = z(z-m)$ if we turn on a mass deformation. 
Quantum-mechanically, we impose quantum determinant relations. 

This is the first case where we can write a reasonable coproduct:
\begin{equation} 
	L(z) = \begin{pmatrix}
		z + bc & b \\
		c & 1 
		\end{pmatrix} 	
\begin{pmatrix}
		z + b'c' & b' \\
		c' & 1 
		\end{pmatrix} 	.
\end{equation}

\subsection{A choice of traces}
In order to define a closed line defect, we need to equip the
auxiliary one-dimensional system with a trace. This can be done in
multiple ways, as the quantum Coulomb branch algebra has a non-trivial
linear space of traces. More precisely it has a non-trivial space of
traces twisted by elements of the Cartan of the global $\mathfrak{g}$
symmetry, which can be defined simply as (analytic continuation of)
traces on highest weight modules for the algebra. Only certain linear
combinations of the twisted traces remain finite as the twisting is
turned off.

An important observation is that the Coulomb branch may not admit any
trace, twisted or not. By the {\it quantum Hikita conjecture}
\cite{Kamnitzer}, the space of twisted traces is isomorphic to the
quantum cohomology ring of the Higgs branch. A balanced ADE quiver has
a rich Higgs branch and a correspondingly rich space of traces.  As we
make the nodes unbalanced, the expected dimension of the Higgs branch
$\sum_i N_i(M_i - \delta_i)$ diminishes. In order to build transfer
matrices, we thus want our quivers not to be too unbalanced.

As traces are multiplicative under the coproduct, we can produce many
traces starting from simple examples, simply by grouping the
elementary 't Hooft operators in different sub-collections and
splitting the charge at infinity accordingly, as long as each
individual collection admits a trace.

Given a trace $\mathrm{Tr}_a$ on the quantized Coulomb branch algebra,
we can define generalized transfer matrices
\begin{equation}
\mathrm{L}^{\mu, \eta}_{a} = \mathrm{Tr}_{a} g_{R_1}[z_1] \cdots g_{R_L}[z_L] 
\end{equation}
acting on a spin chain with sites in representations $R_i$ and
impurity parameters $z_i$. These transfer matrices will commute with
the transfer matrices built from finite-dimensional representations of
the Yangian, essentially by construction. They depend on the charge
$\mu$ at $0$, the charge $\eta$ at infinity and the choice of trace
$a$.

When $\mu$ is minuscule and $\eta$ is in the Weyl orbit of $\mu$, one
can verify that the quantized Coulomb branch algebras are Weyl
algebras and $g_R$ are the oscillator L-operators. The unique twisted
trace on the Weyl algebra gives generalized transfer matrices which
coincide with the Q-operators.

\subsection{'t Hooft-Wilson lines and Coulomb branches}

An 't Hooft-Wilson line lifts in six dimensions to a Wilson line
sitting on an 't Hooft defect. In the enhanced low energy description,
that would be an half-BPS line defect in the three-dimensional
${\cal N}=4$ theory, in the sense studied by
\cite{Assel:2015oxa,Dimofte:2019zzj}. These line defects give rise to
tri-holomorphic sheaves on the Coulomb branch. In the presence of
$\Omega$-deformation, they support algebras of local operators which
are a sort of matrix generalization of the quantized Coulomb branch
algebra.

Line defects which correspond to well-defined 't Hooft-Wilson lines
must have the property that the corresponding algebras still admit
algebra morphisms from the shifted Yangians, so that they can
represent line defects in four-dimensional Chern-Simons theory.  We
thus predict the existence of a sub-category of line defects with this
property.

Physically, there is a general strategy to produce interesting line
defects: the ``vortex construction''. See \cite{Gaiotto:2012xa} for a
four-dimensional analogue.  The basic idea is to consider a
three-dimensional quiver theory with the same shape and $\eta$, but
with some extra flavours, whose masses are tuned in such a way to
allow for a Higgs branch to open up. If the setting is appropriate,
the Higgs branch vev will trigger a flow to the original
three-dimensional theory. A position-dependent Higgs branch vev will
trigger a ``vortex'' flow which ends producing an extra line defect in
the desired three-dimensional theory.

These operations survive the $\Omega$-deformation. The main difference
is that the tuning of the masses is modified by a multiple of $\hbar$,
which governs which vortex flow we follow. At the level of quantum
Coulomb branches, the adjustment of the mass parameters allows for a
further truncation of the algebra.

There is a simple interpretation of these manipulations: the extra
flavours represent a collection of extra 't Hooft lines and the
adjusted masses tell us how to tune the spectral parameter to allow
the 't Hooft lines to fuse into the original line and be truncated to
a more general 't Hooft-Wilson line. We leave a careful treatment of
this construction to future work.

\section{Building 't Hooft lines using Gukov-Witten type surface operators}\label{sec:surface_defects}
Above we described the 't Hooft line in four-dimensional Chern-Simons theory in the standard way, by specifying a singularity in the gauge field.  From this, we derived the phase space of the theory in the presence of the 't Hooft operator.  While convenient for many purposes, this description has some  disadvantages.  

For one thing, it is rather difficult to perform Feynman diagram computations in the presence of an 't Hooft line, described as a singularity in the gauge field.  In addition, it is rather difficult to use this description to understand the effect of an 't Hooft line on other defects in the theory, such as the surface defects that lead to integrable field theories \cite{Costello:2019tri}. 

A full microscopic understanding of 't Hooft lines will allow us to
build a Q-operator in all of the integrable field theories constructed
using the method of \cite{Costello:2019tri}.  This would be a line
defect in each such theory which commutes with the T-operator and
satisfies the TQ and QQ relations, in the normalized form we have been
using.

In this section we will present a more complete description, where the
't Hooft line lives on the boundary of an explicit Lagrangian surface
defect. (In contrast to the surface defects considered in
\cite{Costello:2019tri}, these surface defects are topological.)  This
description has many advantages, and is quite easy to compute with.
We did not introduce it earlier as it is a description that is rather
unfamiliar to most physicists, and is quite special to the situation
we are considering.  This description also allows us to prove the QQ
relation, expressing the fusion of certain 't Hooft lines in terms of
Wilson lines.

The surface defect we introduce is reminiscent of the Gukov-Witten
defect \cite{Gukov:2006jk} in $N=2$ supersymmetric gauge
theories. This defect has the feature that it can be removed by a
singular gauge transformation.  It thus can be thought of as a Dirac
string.  By making this surface defect end, we will find the 't Hooft
line.

 All our computations will be local, and apply equally well to the
 rational, trigonometric, and elliptic setups, and indeed to the
 higher-genus integrable field theories considered in
 \cite{Costello:2019tri}.  Because of this, we will work on
 $\R^2 \times \C$, and consider a surface operator at $0 \in \C$.

\subsection{Defining the surface defect}
Let $X$ be a complex manifold with a $G$-action.  We consider the
following $\sigma$-model with target $X$. The fields of the theory are
a map $\sigma\colon \R^2 \to X$, and a one-form
\begin{equation} 
	\eta \in \Omega^1(\R^2, \sigma^\ast T^\ast X). 
\end{equation}
Here and below, all tensors on $X$ are defining in the holomorphic sense, only using $(1,0)$ vectors and covectors.

The Lagrangian is
\begin{equation} 
	\int_{\R^2} \eta \d \sigma. 
\end{equation}
The field $\eta$ has a gauge transformation generated by
\begin{equation} 
	\chi \in \Omega^0(\R^2,\sigma^\ast T^\ast X). 
\end{equation}
The gauge transformation sends $\eta \mapsto \eta + \d \chi$.

This theory arises in two ways:
\begin{enumerate} 
	\item It is the $B$-twist of the $\sigma$-model with $(2,2)$ supersymmetry.  
	\item It is the analytically-continued version of the Poisson $\sigma$-model with target $X$ \cite{Cattaneo:1999fm}, with zero Poisson tensor.   That is, if we wrote the same formulae for the Lagrangian but where $X$ was a real manifold, we would find the Poisson $\sigma$-model.  Here, we are simply complexifying the fields. 
\end{enumerate}
Since this is an analytically-continued theory, we need to use an
appropriate contour. To describe the algebra of operators of the
theory, which is our main concern, the contour doesn't matter, so we
will not spend much time on this point.

The natural contour from the point of view of the Poisson
$\sigma$-model is to choose a real slice of $X$.  This is also the
natural thing to do in our context: we choose a contour for
four-dimensional Chern-Simons theory associated to a real form of $G$,
and we choose a real slice of $X$ which is acted on by the real form
of $G$.

We will couple this surface defect to four-dimensional Chern-Simons theory. The surface defect is placed at the plane $z = 0$.  The Lagrangian for the coupled theory is simply
\begin{equation} 
	\int_{\R^2 \times  0} \eta \d_A \sigma. 
\end{equation}
We can write this more explicitly. Let $a$ be a Lie algebra index and $i$ an index for local coordinates on $X$. Let $V^a_i$ be the holomorphic vector fields on $X$ giving the action of $\g$.  Then the Lagrangian is
\begin{equation} 
  \int_{\R^2 \times 0}
  \left(\eta \d \sigma + \eta^i (\sigma^\ast V^a_i) A_a\right). 
\end{equation}

The gauge transformation for $\eta$ is now given by the covariant derivative:
\begin{equation} 
	\delta \eta = \eta + \d_A \chi. 
\end{equation}
The bulk gauge transformations act on the fields $\eta,\sigma$ in the obvious way. The Lagrangian is invariant under bulk gauge transformations, but \emph{not} under the gauge transformation generated by $\chi$. By integration by parts, the variation under this gauge transformation is given by
\begin{equation} 
	- \int_{\R^2 \times 0} \chi^i F(A)_a \sigma^\ast (V^a_i). \label{gauge_failure} 
\end{equation}
To correct for this, we will need to vary the four-dimensional gauge-field by the gauge transformation $\chi$. The required variation is
\begin{equation} 
	\delta A_a = \delta_{z = 0} \chi^i \sigma^\ast (V^b_i) g_{ba} ,
\end{equation}
where $g_{ab}$ is the inverse of the Killing form on $\g$.  

The variation of the Chern-Simons action is
\begin{equation} 
  \delta \ChS(A) = \delta_{z = 0} \chi^i \sigma^\ast (V^a_i) F(A)_a ,
\end{equation}
which cancels the variation \eqref{gauge_failure}.  Note that only the $\zbar$ component of the four-dimensional gauge field acquires this extra variation.

\subsection{The equations of motion in the presence of the surface defect}
Let us assume that $G$ acts transitively on $X$, and that further, the stabilizer of a point in $X$ is a parabolic subgroup $P \subset G$.  We will describe the solutions to the equations of motion in this context.

By a gauge symmetry, we can assume that $\sigma$ is constant with value $x$.  The tangent space to $X$ at $x$ is the quotient $\mf{g} / \mf{p}$, where $\mf{p}$ is the Lie algebra of $P$. We let $\mf{n}^- \subset \mf{p}$ be the nilpotent subalgebra. Choose a nilpotent subalgebra $\mf{n}^+$ complementary to $\mf{p}$, and a Levi factor $\mf{l} \subset \mf{p}$. Thus, we have a triangular decomposition
\begin{equation} 
\mf{g} = \mf{n}^- \oplus \mf{l} \oplus \mf{n}^+. 
 \end{equation}
We will choose a basis of $\g$ with indices $t_i^{-}$ for elements of $\mf{n}^-$, $t_i^+$ for elements of $\mf{n}^+$, and $t_{l}^0$ for elements of the Levi factor $\mf{l}$. The Killing form is $\ip{t_i^-, t_j^+}
= \delta_{ij}$, $\ip{t_l^0, t_m^0} = \delta_{lm}$.  We will denote the
corresponding components of the gauge field by $A^-$, $A^+$, $A^0$. 

There is a canonical isomorphism
\begin{equation} 
	T_x X = \mf{n}^+. 
\end{equation}
Hence, we can view the field $\eta$ as a one-form valued in $\mf{n}^-$. 

In this gauge, the Lagrangian takes the form
\begin{equation} 
  \int_{\R^2 \times 0} \eta_i  (x) A^+_i  + \int_{\R^2 \times \C} \d z \, \ChS(A). 
\end{equation}
Varying $\eta$ tells us that, at $z = 0$, $A^+$ vanishes so that $A$
is the parabolic subalgebra $\mf{p} \subset \mf{g}$.

Next, let us vary $A$.  We find
\begin{equation} 
\eta_i \delta_{z = 0} + \d z F(A)^-_i = 0. 
 \end{equation}
This tells us that $A^-$ can have poles at $z = 0$. 

To understand this better, we work in an axial gauge in which $A_{\zbar = 0}$. In the absence of the surface defect, the equations of motion in this gauge tell us that $\dbar_z A_x = 0$ and $\dbar_z A_y = 0$. 

In the presence of the surface defect, these equations are modified to
saying that
\begin{align} 
  \dbar_z A_x^- + \eta_x &= 0 , \\ 
  \dbar_z A_y^- + \eta_y &= 0 .
\end{align}
This tells us that $A^-$ can have a first-order pole at $z = 0$, whose
residue is given by $\eta$.

The result of this analysis is that \emph{all} of the fields on the surface defect have been absorbed into a modification of the gauge field: $A^+$ has a zero at $z = 0$, and $A^-$ has a pole.  

In a similar way, the bulk gauge transformation $\c^+$ must have a zero at $z = 0$. The gauge transformation $\c^-$ can have a pole, whose residue is given by the defect gauge transformation $\chi$. 

For a concrete example, consider the case when $G = PSL_2(\C)$ and $X = \CP^1$. In this case, the surface defect is equivalent to allowing the component $A^1_2$ of the gauge field has a pole at $z = 0$, and requiring that $A^2_1$ has a corresponding zero. 

\subsection{Removing the surface defects corresponding to minuscule coweights}

Now let us analyze when we can remove the surface defect. As above,
our surface defect is built using a parabolic subgroup $P \subset G$,
and we continue to use the notation $\mf{n}^{\pm}$, $\mf{l}$, where
$\mf{p} = \mf{n}^- \oplus \mf{l}$.

The condition we need in order to be able to remove the surface defect
is that there is a cocharacter $\rho\colon \C^\times \to G$ such that
the $\mf{l}$ is the zero eigenspace of $\rho$ in the adjoint
representation, and $\mf{n}^{\pm}$ are the $\pm 1$ eigenspaces.  This
condition means that $\rho$ is a minuscule coweight.

In this situation, we can perform a gauge transformation by the
singular gauge field $z^{\rho}$ (where as before $z$ is the coordinate
in the holomorphic direction of the four-dimensional space-time).
Conjugating the gauge field by $\rho(z)$ will give the field $A^+$ a
first-order zero at $ z= 0$, and $A^-$ a first-order pole. This is
because $A^{\pm}$ are the $\pm 1$ eigenspaces of the cocharacter
$\rho$.

In the example where $G = PSL_2(\C)$ and $X = \CP^1$, we have
$z^{\rho} = \op{Diag}(z,1)$.  Conjugating by $\rho(z)$ gives $A^1_2$ a
pole at $z = 0$ and $A^2_1$ a zero, which is equivalent to introducing
the surface defect.

\subsection{The surface defect as a monodromy defect}

Suppose we have a simply-connected gauge group $G$, and we have a
surface defect as above associated to a minuscule coweight $\rho$ of
the adjoint form $G_{ad}$ of $G$.  The Weyl group orbits of minuscule
coweights are in bijection with the non-identity elements of the
center of $G$.  We will show that in this situation, the surface
defect associated to a minuscule coweight is equivalent to a
``monodromy defect'', where the gauge field has monodromy around the
location of the surface given the corresponding element of the center
of $G$.

Strictly speaking, the phrase monodromy defect is not quite correct
when the surface wraps the topological plane, because we only have a
partial gauge field in the holomorphic plane. In this case, the
``monodromy'' means that the bundle is a $G_{ad}$-bundle with lifts to
a $G$-bundle away from the location of the surface defect but there is
an obstruction to such a lift near the defect.

Let us explain these obstructions.  Let $\mc{G}$ be the sheaf on $\C$
of smooth maps to $G$. For any smooth manifold, $H^1(M, \mc{G})$ is
the set of isomorphism classes of smooth $G$-bundles on $M$.  There is
a short exact sequence
\begin{equation} 
	1 \to Z(G) \to \mc{G} \to \mc{G}_{ad} \to 1 ,
\end{equation}
where $Z(G)$ is the constant sheaf with value the center of $G$. This
leads to an exact sequence of cohomology groups
\begin{equation} 
	H^1(U,\mc{G}) \to H^1(U,\mc{G}_{ad}) \to H^2(U,Z(G)) 
\end{equation}
for any open $U$ in $\C$.  Thus, any smooth $G_{ad}$ bundle has a
characteristic class living in $H^2$ with coefficients in $Z(G)$, and
this characteristic class is the obstruction to lifting to a
$G$-bundle.

In the case $G = SL_n$, this characteristic class is the first Chern
class of a $PSL_n$ bundle, taken modulo $n$.

The surface defect we introduce, when we take the gauge group to be
the simply connected form $G$, is equivalent to working with the
adjoint form $G_{ad}$ but where we specify the value of the
characteristic class in $H^2_c(U, Z(G))$ (where $U$ is a neighbourhood
of the location of the surface defect and the subscript $c$ indicates
compact support). Suppose the surface defect is at $z = 0$, and is
associated to a Weyl orbit of minuscule coweights corresponding to an
element $\rho \in Z(G)$. Then we will show that the characteristic
class is
\begin{equation} 
	\delta_{z = 0} \rho \in H^2_c(U, Z(G)) . 
\end{equation}

To see this, we recall that we can remove the surface defect using the
gauge transformation $\rho(z) \in G_{ad}$.  This gives us a \v{C}ech
description of the principle $G_{ad}$ bundle in the presence of the
surface defect, as follows. We let $U$ be a neighbourhood of $z =
0$. The transition function gluing a bundle near $z = 0$ to a bundle
away from $z = 0$ is given by a smooth map
\begin{equation} 
	U \setminus \{z =0\} \to G_{ad}.  
\end{equation}
The $G_{ad}$ bundle sourced by the surface defect has transition
function given by $\rho(z)$.  This will define a $G$-bundle if and
only $\rho(z)$ lifts to a continuous map to $G$.  The obstruction to
doing so is the homotopy class of the map
$\rho(z)\colon S^1 \to G_{ad}$ in $\pi_1(G_{ad}) = Z(G)$.

This tells us that the obstruction to lifting the bundle sourced by
the surface defect is the class $\delta_{z = 0} \rho \in H^2(U,Z(G))$.

This description of the surface defect makes it clear that it is
entirely topological, and can be placed on any two-cycle $\Sigma$ on
the four-dimensional spacetime of four-dimensional Chern-Simons
theory.  In the presence of the surface defect, we take our gauge
group to be the adjoint form $G_{ad}$ but only include topological
types of bundles with characteristic class given by
$\delta_{\Sigma} \rho \in H^2(X,Z(G))$.

\subsection{Monodromy defects for $SL_n(\C)$}
In the case $G = SL_n(\C)$ and we use the surface defect associated to
$\CP^{N-1}$, there is a convenient way to rephrase this. In
perturbation theory, our surface defect has the effect of moving us
from the trivial bundle $\Oo^n$ on the holomorphic plane to the bundle
$\Oo(1) \oplus \Oo^{n-1}$.  This bundle has determinant $\Oo(1)$, and
so is not an $SL_n(\C)$ bundle.  More generally, if we start by
perturbing around some bundle $E$ with trivial determinant, and
introduce the surface defect, we will find the configuration is
equivalent to working with a bundle $E'$ with $\op{det} E' = \Oo(1)$.

The introduction of the surface defect means that instead of
considering bundles with trivial determinant, we are working with
bundles with determinant $\Oo(1)$.  Such bundles are principle bundles
for a twisted form of $SL_n(\C)$, which forms a non-trivial bundle of
groups on the holomorphic plane $\C$. This is the group of
automorphisms of the bundle $\Oo(1) \oplus \Oo^{n-1}$ which act as the
identity on the determinant bundle $\Oo(1)$.  Introducing the surface
defect means we are working with this twisted form of $SL_n(\C)$.

Asking that a bundle has determinant $\Oo(1)$ is the same as asking
that its first Chern class is $\delta_{z = 0}$. The obstruction to a
$PSL_n(\C)$ bundle lifting to an $SL_n(\C)$ bundle is the first Chern
class, modulo $n$.  Working with bundles with fixed determinant
$\Oo(1)$ is then equivalent to working $PSL_n(\C)$ bundles whose
modulo $n$ Chern class\footnote{A $PSL_n(\C)$ bundle is a complex
  vector bundle taken up to tensoring with a line bundle. Tensoring
  with a line bundle adds a multiple of $n$ to the first Chern class,
  so the first Chern class of a $PSL_n(\C)$ bundle is well-defined
  modulo $n$.} is $\delta_{z = 0}$.

\subsection{Surface defects and the dynamical Yang-Baxter equation}

In \cite{Costello:2017dso}, it was shown that if we consider
four-dimensional Chern-Simons theory with holomorphic curve $\Sigma$,
and we choose boundary conditions so that the holomorphic bundles we
consider on $\Sigma$ have no moduli, then the theory leads to
solutions of the Yang-Baxter equation with no dynamical parameter. In
\cite{Costello:2018gyb}, we discussed the case when the holomorphic
bundles on $\Sigma$ have moduli, in which case the solution to the
Yang-Baxter equation has a dynamical parameters given by the moduli.

Let us consider the case that $\Sigma$ is an elliptic curve $E$, with
gauge group $SL_n(\C)$.  The moduli of semistable holomorphic
$SL_n(\C)$ bundles on the elliptic curve is isomorphic to
$\CP^{n-1}$. Thus, we find solutions to the Yang-Baxter equation with
dynamical parameter.

Introducing surface defects changes the moduli.  In some cases, a
surface defect adds moduli; in other cases, moduli are removed.

For example, let us introduce a single $\CP^{n-1}$ surface defect in
the case when we work on an elliptic curve with gauge group $E$.  In
this case, we are no longer considering the moduli of $SL_n(\C)$
bundles on $E$, but the moduli of $PSL_n(\C)$ bundles whose first
Chern class modulo $n$ is $1$.  This moduli space is a point
\cite{MR1212625} (at least when we focus on stable bundles, as is
physically reasonable).  The unique stable bundle is the rigid
$PSL_n(\C)$ bundle studied in \cite{Costello:2018gyb}.  In this
example, we find that introducing the surface defect has moved us from
a situation with zero modes to one without zero modes.

In \cite{Costello:2018gyb}, it was shown that the path integral
defined in perturbation theory around this rigid $PSL_n(\C)$ bundle
yields Belavin's elliptic R-matrix, which does not have a dynamical
parameter.  This analysis was only perturbative, and neglected other
possible $PSL_n(\C)$ bundles.  Non-perturbatively, it would seem that
four-dimensional Chern-Simons theory for the group $PSL_n(\C)$ does
\emph{not} lead to Belavin's elliptic R-matrix, because of the
presence of bundles of other degree.

The correct, non-perturbative statement is that the elliptic R-matrix
without dynamical parameter arises from four-dimensional Chern-Simons
theory for the group $SL(n,\C)$ with a single surface defect.

\subsection{The dynamical Yang-Baxter equation in the rational case}

Now let us consider four-dimensional Chern-Simons theory on
$\R^2 \times \C$, with simply-connected gauge group $G$. The boundary
condition is that the gauge field tends to zero at $z = \infty$, so
that the bundle extends to a bundle on $\R^2 \times \CP^1$. This
bundle is trivialized at $\R^2 \times \infty$.

Introducing a surface defect corresponding to an element of the center of $G$ means that we use $G_{ad}$ bundles on $\R^2 \times \CP^1$, which have non-trivial topology on $\CP^1$.We expect that having a surface defect built from the $\sigma$-model on $G / P$, for a parabolic $P \in G$ corresponding to a minuscule coweight, we should find dynamical parameters living in $G / P$.    After all, the only extra degrees of freedom we have introduced live in $G / P$.  

We can also see this from the geometry of the moduli space of
$G$-bundles on $\CP^1$.  Connected components of the moduli stack of
holomorphic $G_{ad}$ bundles on $\CP^1$ are labelled by the center of
$G$, or equivalently, by Weyl orbits of minuscule coweights.  If we
ask that such bundles are trivialized at $\infty$, then each connected
component has a $G$-action.  In each connected component, there is an
open $G$-orbit which is of the form $G/P$, with parabolic $P$
corresponding to the minuscule coweight labelling the component.

We expect this setup to lead to rational R-matrices with a dynamical
parameter living in $G/P$.

\section{Making the surface defect end, and 't Hooft lines}

We would like to have an interface between gauge-theory configurations
with this surface defect at $z = 0$, and configurations without the
surface defect.  This turns out to be very simple: all we have to do
is impose a boundary condition for our surface defect.

We coordinatize the topological plane by $x,y$ and consider the
surface defect on the region where $y \ge 0$.  The most natural
boundary condition for the Poisson $\sigma$-model with target $X$ is
to set $\eta = 0$ on the line $y = 0$.

More generally, we can introduce a boundary condition where we set
$\eta = 0$ and at the boundary place a holomorphic vector bundle on
$X$. In order to be able to couple to the gauge field, we will need
this holomorphic vector bundle to be $G$-equivariant.

In the case that the surface defect is of the form $G / P$ for a
parabolic associated to a minuscule coweight, we have seen that we can
remove the surface defect by a gauge transformation $\rho(z)$.  When
the surface defect has boundary at $y = 0$ we should use a gauge
transformation which is $\rho(z)$ for $y < -\eps$ and the identity for
$y > 0$. When we do this, we are left with a line defect.

This line defect is an 't Hooft line (or an 't Hooft-Wilson line if we
include a non-trivial vector bundle on $X$ into the boundary
conditions). To see this, we note that classically, the field sourced
by this line defect is obtained from the trivial field configuration
by applying a gauge transformation which is $\rho(z)$ for $y \ll 0$
and the identity for $y \gg 0$.  The effect of such a gauge field on a
Wilson line in the $y$-direction is $\rho(z)$, which is precisely what
we expect from an 't Hooft line.

There is a subtlety here, in that the precise matrix $\rho(z)$ that
arises depends on which point on the manifold $X = G/P$ we use. (The
conjugacy class of the matrix does not).  To specify the field sourced
by the 't Hooft line defect, we need to say what the line defect does
at $x = \pm \infty$.  In this analysis we are giving our surface
defect Neumann boundary conditions $\sigma = \sigma_0$ for a point
$\sigma_0 \in G/P$ at $x = \pm \infty$.  When we use a gauge
transformation to write the surface defect as a line defect, the
choice of boundary condition for the surface defect specifies states
at the end of the 't Hooft line.

\subsection{Equivalence between the surface defect and affine Grassmannian definitions of the 't Hooft line}

The definition of the 't Hooft line using surface defects is
equivalent to that using the affine Grassmannian.  To see this, we
need to describe the moduli space of solutions to the equations of
motion of the theory when we have a surface defect that ends.

We will take $G = G_{ad}$ to be of the adjoint form, so that there are
no Dirac strings.  We will place the surface defect given by a $G/P$
$\sigma$-model at $z = 0$, $y \le 0$. The boundary condition at
$y = 0$ is $\eta = 0$.

We will consider the solutions to the equations of motion in a
neighbourhood of $z = 0$. When $y > 0$, we can trivialize the bundle.
The group $G[[z]]$ acts on the space of solutions by change of the
trivialization on the region where $y > 0$.

Away from $y = 0$, $z = 0$, the moduli space of solutions is the same
as in the absence of the surface defect, because the surface defect
can be removed by a gauge transformation. Thus, away from $y = 0$,
$z = 0$, the moduli space of solutions is described by the affine
Grassmannian $G((z)) / G[[z]]$.

Including the locus when $y = 0$, $z = 0$, we find that the moduli
space of solutions to the equations of motion is given by some
$G[[z]]$-orbit in the affine Grassmannian.  This orbit contains the
point in affine Grassmannian given by the minuscule coweight defining
the surface defect.  This proves the equivalence between the
surface-defect picture and the affine Grassmannian picture.

\subsection{'t Hooft lines as interfaces}

We have defined a fractional 't Hooft line for the adjoint form of a
group $G$ as the line operator living at the end of the surface
defect. Tautologically, the 't Hooft line is an interface between the
theory in the presence of the surface defect and the theory without a
surface defect.

One can ask why an 't Hooft line would be expected to behave in this
way.  Let us consider the rational case, so the holomorphic plane is
$\C$, and take the gauge group to be $SL_2(\C)$.  A fractional 't
Hooft line at $z = 0$, $y = 0$ gives rise to a Hecke transformation
which will turn the trivial bundle at $y < 0$ into a bundle isomorphic
to $\Oo(1)\oplus \Oo$ on $y > 0$.  A gauge field for the bundle
$\Oo(1)\oplus \Oo$ has non-trivial zero modes, leading to a dynamical
parameter.

From this we see that an 't Hooft operator provides an interface for
four-dimensional Chern-Simons theory on $\R^2 \times \C$ where there
is no dynamical parameter, to a setting where there is a dynamical
parameter.

There is an important difference between $PSL_2(\C)$ and $SL_2(\C)$,
or more generally between the adjoint and simply connected form of a
group. For $PSL_2(\C)$, the non-perturbative path integral for
four-dimensional Chern-Simons theory should sum over the different
topological types of gauge field on $\C$ (trivialized at
$\infty$). For $PSL_2(\C)$, there are two components, and the 't Hooft
line associated to the minuscule coweight moves us from one component
to the other.

For $SL_2(\C)$, the two connected components of the moduli of
$PSL_2(\C)$ bundles have a different interpretation.  The component
with trivial $w_2$ gives us $SL_2(\C)$ bundles, and this appears in
the theory without a surface defect. The component with non-trivial
$w_2$ occurs when we have a surface defect.

The fact that minuscule 't Hooft lines always move us from a setting
with no dynamical parameter to one with a dynamical parameter will
cause us difficulties when we define Baxter's Q-operator in terms of
't Hooft lines.  We will find that to define Baxter's Q-operator we
will need to have a line defect at $z = \infty$ as well as an 't Hooft
line at $z = 0$.

\section{'t Hooft lines at $\infty$ and Q-operators}
\label{sec:thooft_infinity}

In this section we will re-derive the phase space of minuscule 't
Hooft lines from the surface defect picture.  As we have seen, the
Q-operator arises when we have an 't Hooft line in the bulk but also
modify the boundary condition at $z = \infty$ along a line parallel to
the 't Hooft line.  Here, we will describe this modification of the
boundary condition at $z = \infty$ as introducing a surface defect at
$\infty$.

Thus, fix a minuscule coweight $\rho$. Let us decompose $\g$ as
$\g = \g_{-1} \oplus \g_{0} \oplus \g_{1}$ according to the
eigenvalues of $\rho$. We define a parabolic $P$ by saying that its
Lie algebra is $\g_{0} \oplus \g_{1}$.  We define unipotent subgroups
of $G$ by $G^{\pm 1} = \exp(\g_{\pm 1})$.  We let the corresponding
components of $A$ be $A^{-1}$, $A^0$, $A^1$.

The defect at $z = \infty$ is described by coupling the gauge theory
to the analytically-continued Poisson $\sigma$-model on the vector
space $\g^{1}$.  The fields of this theory are
\begin{align} 
	\til{\sigma}\colon \R^2 &\to \g^{1} , \\
	\til{\eta} & \in \Omega^1(\R^2,  \g^{-1} ).
\end{align}
The Lagrangian is $\int \til{\eta} \d \til{\sigma}$. As before,
$\til{\eta}$ has a gauge symmetry,
$\til{\eta} \mapsto \til{\eta} + \d \til{\chi}$.

We would like to couple this to the gauge theory. The boundary
conditions for the gauge theory require our gauge field $A$ is
divisible by $1/z$ near $\infty$. To couple in a non-trivial way to
the $\sigma$-model on $\g^{1}$, we need to involve the derivative of
$A$ in $u = 1/z$ at $z = \infty$.

We couple the topological $\sigma$-model with target $\g^{1}$ to the
gauge field $\partial_u A^{1}$ acting by translation, so that the
Lagrangian is
\begin{equation} 
  \int (\til{\eta} \d \til{\sigma} + \til{\eta} \partial_u A^{1}). 
\end{equation}
(Recall $u = 1/z$.)  The other components of the gauge field $A$ do
not couple to the surface defect.

Only the transformations corresponding to the $\partial_u \c^{1}$ ghosts will act in a non-trivial way. These act by translation.

As before,  the gauge transformation of $\til{\eta}$ is now 
\begin{equation} 
	\til{\eta} \mapsto \til{\eta} + \d \til\chi + \partial_u A^{1}. 
\end{equation}
For the coupled system to be gauge invariant, we need to make $A^{-1}$
vary under the $\chi$ gauge transformation:
\begin{equation} 
	A^{-1} \mapsto A^{-1} + \delta_{u = \eps} \chi ,
\end{equation}
where $\eps$ is a small parameter.  With this term, the variation of
the term $\int A^{-1} \d A^{1} u^{-2} \d u$ under $\chi$ gives us the
term $\int_{u = 0} \partial_u \d A^{1} \chi$, which cancels the
variation of $\int_{u = 0} \til{\eta} \partial_u A^{1}$ under $\chi$.

A repeat of the argument we gave before shows that, if we solve the equations of motion in the presence of the defect at $z = \infty$, then $A^{1}$ is forced to have a second-order zero, whereas $A^{-1}$ is allowed to be regular at $u = 0$.  The value of $A^{1}$ at $u = 0$ is given by $\til{\eta}$, whereas $\til{\sigma}$ can be set to zero by a gauge transformation.

Before we introduced the defect, the boundary condition stated that $A$ vanishes at $u = 0$.  Conjugating by $z^{-\rho} = u^{\rho}$ introduces an extra factor of $u$ to $A^{1}$, so that it has a second order zero at $u = 0$; and introduces an extra factor of $u^{-1}$ to $A^{-1}$, so that it can be regular. Thus, applying the gauge transformation $z^{-\rho}$  moves us from the setting without the defect at $z = \infty$ to one with the defect.

Next, suppose we introduced a defect both at $0$ and at $\infty$.  At $z = 0$, we will introduce the topological $\sigma$-model on the partial flag manifold $G/P$, and work in perturbation theory around the field configuration given by the image of the identity in $G$.  The stabilizer of this point is $P$, so that, following our earlier analysis, the defect is equivalent to asking that $A^{-1}$ has a zero at $z = 0$ and $A^1$ has a first-order pole.  Applying $z^{-\rho}$ moves us from the trivial defect to this defect.

We conclude that the gauge transformation $z^{-\rho}$ removes the defects at $0,\infty$ simultaneously. 

We can make the defect at $z = \infty$ end at $y = 0$, by using the boundary condition where $\eta = 0$.   If we remove the defect at $y < 0$ by a gauge transformation, we are left with a line defect at $\infty$, as we considered in section \ref{sec:q}. 

If we make both the bulk and boundary surface defects end at $y = 0$, we are left with the configuration that we found earlier gives the Q-operator.  We will re-derive the phase space of the Q-operator from this perspective shortly.   

At a first pass, however, consider a Wilson line at a point $z$ along $x = 0$.  Let us consider the effect of the two surface defects, ending at $y = 0$, on this Wilson line.  To remove the surface defects on $y < 0$, we apply the gauge transformation $z^{-\rho}$.   As we have seen, this removes both surface defects at $0$ and $\infty$, leaving us with line defects at $y = 0$, $z = \infty$.  

To leading order in $\hbar$, the state in the Wilson line at $x = 0$ transforms by $z^{-\rho}$ when it passes the line defects at $y = 0$.  This is because we have applied the gauge transformation $z^{-\rho}$ at one side and not the other.  This is precisely the behaviour that characterizes an 't Hooft line.  The quantum corrections to this expression can in principle be calculated by the exchange of gluons between the surface defects and the Wilson line, giving an operator which is $z^{-\rho}$ times a series in $\hbar z^{-1}$.  In practice, calculating these quantum corrections is quite non-trivial.

\subsection{The phase space in the presence of surface defects}

Let us consider the situation above, with surface defects at $0$ and
$\infty$ both ending at $y = 0$.  In this section we will show that we
can recover the phase space of minuscule 't Hooft lines that we
discussed in section \ref{section_phase}.

Let us consider the effective two-dimensional theory obtained by
compactifying on the $z$-plane $\CP^1$, with defects at $0$ and
$\infty$.  Before the introduction of the defects, this is the trivial
theory.  The defects make the effective theory the topological
$\sigma$-model with target $G/P \times \mf{g}_1$.  This is simply the
product of the topological $\sigma$-models we have inserted at
$z = 0$, $z = \infty$.

The exchange of a single gluon between the defects at $z = 0$ and
$z = \infty$ couples the topological $\sigma$-model on $G/P$ to that
on $\mf{g}_1$. As before, we use the notation $\eta$, $\sigma$ for the
fields of the topological $\sigma$-model on $G/P$, and $\til{\eta}$,
$\til{\sigma}$ for the $\sigma$-model on $\mf{g}_1$. The defect at
$z = \infty$ is only coupled to the components $A^1$ of the gauge
field in $\mf{g}_1$.

We will compute the coupling between the two defects using the
technique of \cite{Costello:2019tri}.  Let us explain the general
method. Choose a basis $t_a$ of $\mf{g}$, and let $J_a$, $\til{J}_a$
be the currents which couple the defects at $0$ and $\infty$ to the
gauge field.  If $c^{ab}$ denotes the quadratic Casimir, then the
coupling between defects at $z,z'$ is
\begin{equation} 
	c^{ab} J_a \til{J}_b \frac{1}{z - z'}.  
\end{equation}

In our case, the defect at $\infty$ is only coupled to elements of
$\mf{g}_1$. Because of the form of the quadratic Casimir, the coupling
can only involve the components of the current in $\mf{g}_{-1}$ at
$z = 0$. We can take a basis $X_i$ of $\mf{g}_1$ and the dual basis
$Y^i$ of $\mf{g}_{-1}$.  Let $V^i$ denote the vector fields on $G/P$
giving the action of $\mf{g}_{-1}$.  Then,
\begin{equation}
		J^i = \ip{\eta,V^i}. 
\end{equation}
The derivative $-z^2 \partial_z$ of $A_1$ couples to $\til{\eta}_i$ at $z = \infty$. 

From this, we see that the exchange of a gluon couples the two surface
defects by
\begin{equation} 
  -\lim_{z \to \infty} \int_{\R^2} \ip{\eta,V^i} \til{\eta}_i
  z^2 \partial_z \frac{1}{z}
  =  \int_{\R^2} \til{\eta}_i \ip{\eta, V^i}.
\end{equation}

This coupling also modifies the gauge transformations, so that
$\sigma$, $\til{\sigma}$ also transform by the gauge transformations
$\chi$, $\til{\chi}$:
\begin{equation}
	\begin{split}
	\til{\sigma}_i &\mapsto \til{\sigma}_i - \ip{\chi, V_i}, \\
	\sigma & \mapsto \sigma + \til{\chi}^i V_i.
	\end{split}
\end{equation}

This describes the effective two-dimensional model obtained by
integrating out the four-dimensional gauge fields, at the classical
level.  This model is the analytically-continued Poisson
$\sigma$-model \cite{Cattaneo:1999fm} on $G/P \times \mf{g}_{1}$. This
manifold has a holomorphic Poisson tensor, coming from the map
$\mf{g}_{-1} \to \op{Vect}(G/P)$ and the component of the quadratic
Casimir which lies in $\mf{g}_1 \otimes \mf{g}_{-1}$.

For example, if $G = SL_2(\C)$, then the holomorphic Poisson manifold
is $\CP^1 \times \C$, with coordinates $u$, $v$.  The Poisson tensor
is $u^2 \partial_u \partial_v$.

The boundary condition at $y = 0$ that gives rise to the 't Hooft
defects is that where $\eta= 0$, $\til{\eta} = 0$.

We can choose a reality condition for the analytically-continued
Poisson $\sigma$-model for which the fields only see the open subset
where the Poisson tensor is non-degenerate.

To describe this, note that there is an open orbit of $\exp(\g_{-1})$
in $G/P$ containing the image of the identity in $G/P$.  This open
orbit has no stabilizer, and is isomorphic to $\g_{-1}$.  If we assume
that $\sigma$ is in this orbit, then topological $\sigma$-model on
$G/P$ is replaced by that on $\g_{-1}$. In this orbit, the action is
simply
\begin{equation} 
  \int (\til{\eta}^i \eta_i + \eta_i \d \sigma^i
  + \til{\eta}^i \d \til{\sigma}_i)
\end{equation}
and the gauge transformations are $\delta \sigma^i = \til{\chi}^i$,
$\delta \til{\sigma}_i = \chi_i$, $\delta \eta_i = \d \chi_i$,
$\delta \til{\eta}^i = \d \til{\chi}^i$.

This is the Poisson $\sigma$-model with target the symplectic manifold $\mf{g}_{-1} \oplus \mf{g}_1$.  As is well-known, the bulk degrees of freedom of the Poisson $\sigma$-model with target a symplectic manifold are entirely massive, and the model localizes onto the boundary.

Concretely, we can see this as follows.    Under the field redefinition $\gamma_i = \eta_i + \d \til{\sigma}_i$, $\til{\gamma}^i = \til{\eta}^i - \d \sigma^i$, the action becomes simply
\begin{equation} 
  \int_{\R \times \R_{\ge 0}} \til{\gamma}^i \gamma_i
  - \int_{\R} \til{\sigma}^i \d \sigma_i .
\end{equation}
The second term is the Lagrangian for topological quantum mechanics
with values in the symplectic manifold $\mf{g}_1 \oplus \mf{g}_{-1}$,
and so describes the phase space of the 't Hooft line.

The bulk system is entirely massive.  Indeed, the fields $\gamma_i$,
$\til{\gamma}^i$ are invariant under the $\chi$, $\til{\chi}$ gauge
transformation, and $\sigma^i$, $\til{\sigma}_i$ transform by a
shift. We can choose a gauge where $\sigma_i(x,y)$ is independent of
$y$ (where the boundary is at $y = 0$).  In this gauge, the fields
$\sigma_i(x,y)$ are only boundary fields, and the fields $\gamma_i$,
$\til{\gamma}^i$ are massive with no kinetic term.

We have shown that the effective theory obtained by compactifying
four-dimensional Chern-Simons theory to $\R^2$, in the presence of two
surface defects which end, gives the trivial theory on $\R^2$ together
with a line defect describing topological quantum mechanics with
target $\mf{g}_1 \oplus \mf{g}_{-1}$.  This manifold is therefore the
phase space.  Thus, we have re-derived from this perspective the phase
space of the theory in the presence of an 't Hooft line at $0$ and at
$\infty$.

\section{Quantization of the 't Hooft line}
\label{sec:quantizing_tHooft}

Classical Wilson lines in four-dimensional Chern-Simons theory
sometimes have an anomaly to quantization. This was studied by an
explicit Feynman diagram computation in \cite{Costello:2018gyb}.  One
can ask if the same holds for 't Hooft lines.

If we describe the 't Hooft line as a singularity in the gauge field,
this question is very difficult to answer.  An advantage of the
surface-defect description is that we can answer this question
directly by using standard cohomological techniques.

For any topological line defect in a partially topological theory,
there is a differential-graded associative algebra $A$ of local
operators on the defect. (At the classical level, $A$ will be
commutative.) By descent, possible deformations of the defect are
given by elements of $A$ of ghost number $1$.  Thus, to compute the
possible first-order deformations of the defect, we need to compute
$H^1(A)$.  Similarly, anomalies to quantizing the system will be
described by $H^2(A)$.  This perspective was taken in
\cite{Costello:2018gyb} in the analysis of Wilson lines.

Here we will prove the following result.
\begin{theorem} 
  Let $A$ be the dg algebra of local operators on an 't Hooft line
  associated to a minuscule coweight of any simple group.  then
  $H^2(A) = 0$ and $H^1(A)$ is one dimensional. This implies that
  there are no anomalies to constructing the 't Hooft line at the
  quantum level, and that there is only one possible counterterm,
  which corresponds to moving the spectral parameter.
\end{theorem}
\begin{proof}
  Let us describe the algebra $A$ that appears in the analysis of a
  minuscule 't Hooft line.  (In the language of
  \cite{Costello:2016vjw}, this will be the factorization algebra of
  operators in the coupled system.)  The surface defect is the
  topological $\sigma$-model with target $G/P$.  On the boundary, we
  ask that $\eta = 0$, and the gauge symmetries of the system also
  vanish. Thus, the boundary operators are just built from the field
  $\sigma$. Further, the equations of motion show that
  $\d \sigma = 0$, so derivatives of $\sigma$ do not appear.
  Therefore the algebra of operators on the boundary of the defect is
  simply $\Oo(G/P)$, the algebra of polynomial functions on
  $G/P$. Since we are doing a perturbative analysis, we can work in a
  neighbourhood of a point in $G/P$, which is $\mf{g}_1$.  The algebra
  of operators is then the algebra of functions on $\mf{g}_1$, which
  is the symmetric algebra on $ \mf{g}_{1}^\vee$.

  This is before coupling to the bulk gauge theory.  The algebra of
  operators of the bulk system is generated by the ghost $\c$ and its
  $z$-derivatives. The generators $\partial_z^k \c$ are in ghost
  number $1$, anticommute with each other, and have usual BRST
  operator. In mathematical terms, this algebra is the Lie algebra
  cochain complex of $\mf{g}[[z]]$.

  When we couple the two systems, we have the operators $\sigma$ in
  $\mf{g}_{1}^\vee$ in ghost number $0$, and
  $\partial_z^k \c \in g^\vee$ in ghost number $1$.  Because the field
  in $\mf{g}_{1}^\vee$ transforms by a shift under the action of gauge
  symmetry, we find a term in the BRST operator whereby
\begin{equation} 
	Q \sigma^i = \c^i. 
\end{equation}
(Here $i$ runs over a basis of $\mf{g}_1^\vee$.) 

This makes it clear that the cohomology of the coupled system is
generated in ghost number $1$ by $\partial_z^k \c^a$, for $k > 0$ and
all values of $a$, and by $\c^a$ for $a$ corresponding to elements of
$\mf{g}_0 \oplus \mf{g}_{-1}$ In mathematical terms, the algebra of
operators is the Lie algebra cochains of the algebra
\begin{equation} 
	z \mf{g}[[z]] \oplus \mf{g}_0 \oplus \mf{g}_{-1}. 
\end{equation}
Algebras of this type are called (parabolic) Iwahori algebras in the mathematics literature.  It is the Lie algebra of the group of maps from the formal disc to $G$, which at the origin land in the parabolic subalgebra $P$ which exponentiates $\mf{g}_0 \oplus \mf{g}_{-1}$. 

Our task is to compute the cohomology of this Lie algebra in degrees $1$ and $2$.  The task is simplified by noting that all cohomology classes must live in the trivial representation of the subalgebra $\mf{g}_0$, which acts semi-simply on the Lie algebra cochain complex.  Let us use $\c^i$ for the ghosts in $\mf{g}_{-1}$ and $\c_i$ in $\mf{g}_1$.    We can decompose $\mf{g}_0$ into an Abelian algebra, spanned by the minuscule coweight $\mu$, and a semi-simple algebra $\mf{l}$.  We let $\c^0$ indicate the component of the ghost corresponding to the Abelian algebra, and $\c^{\alpha}$ the components corresponding to $\mf{l}$.   The spaces $\mf{g}_{\pm 1}$ are both irreducible representations of $\mf{g}_0$.

This discussion shows that $\mf{g}_0$ invariant cochains of ghost number $\le 2$ are of the form:
\begin{enumerate} 
	\item $\partial_z^k \c^i \partial_z^l \c_i$ for $l > 0$. Cochains of this nature are in ghost number $2$. 
	\item Cochains which only involve $\mf{g}_0[[z]] = \mf{l}[[z]] + \C \dot \mu [[z]]$.
\end{enumerate}
	Let us first consider the cochains in $\mf{l}[[z]]$.  These must be $\mf{l}$ invariant.   Since $\mf{l}$ is a sum of simple Lie algebras, there are no $\mf{l}$--invariant cochains in degree $1$, as there are no $\mf{l}$-invariant elements in $\mf{l}$.  There are, however, $\mf{l}$-invariant cochains in degree $2$, and we need to show that these are never closed.

	The results of \cite{MR2415401} compute the cohomology of $\mf{l}[[z]]$ and show that it is the same as the cohomology of $\mf{l}$.  This vanishes in degrees $1$ and $2$. Therefore, no $\mf{l}$-invariant cochain in degree $2$ can be closed (because there is no cohomology and also no exact cochains).

Next, we need to consider whether linear combinations of the expressions $\partial_z^k \c^i \partial_z^l \c^i$ are BRST closed.  The charge under rotating $z$ is a symmetry of the problem, so we can fix $k+l$.  

One of the terms in the BRST variation of $\c^i$ is $\c^0 \c^i$, and
one of the terms in the BRST variation of $\c_i$ is $-\c^0 \c_i$.  Let
us compute the coefficient of
$(\partial_z^k \c^0) \c^i \partial_z^l \c_i$ in the BRST variation of
$\partial_z^k \c^i \partial_z^l \c_i$, assuming $l > 0$.  We find
\begin{equation}
  \begin{split}
    Q \partial_z^k \c^i \partial_z^l \c^i
    &= \partial_z^k (\c^0 \c^i) \partial_z^l \c_i - \partial_z^k \c^i \partial_z^l (\c^0 \c_i) + \dotsb \\
    &=   (\partial_z^k  \c^0) \c^i \partial_z^l \c_i  + \dotsb
  \end{split}
\end{equation}
assuming that $l > 0$.   The only other expression whose BRST variation can produce a term like $(\partial_z^k \c^0) \c^i \partial_z^l \c_i$ is $\c^i \partial_z^{k+l} \c_i$.  

We conclude that in any linear combination
\begin{equation} 
	\sum A_{k,n-k} \partial_z^k \c^i \partial_z^{n-k} \c_i 
\end{equation}
	for scalars $A_{k,n-k}$  which is BRST invariant, the coefficient $A_{k,n-k}$ is determined by $A_{0,n}$, so that there is at most one BRST invariant term for each value of $n$.

However, there is also one BRST exact term for each value of $n > 0$, namely 
\begin{equation} 
	Q \partial_z^n \c^0 = \partial_z^n ( \sum \c^i \c_i). 
\end{equation}
This is not zero if $n > 1$.  This BRST exact term must therefore be the same as the unique BRST closed term. 

	This completes the proof that there is no cohomology in degree $2$.  Therefore, the 't Hooft line exists at the quantum level.

	Let us now finish by computing the degree $1$ cohomology. The only $\mf{g}_0$-invariant cochains in degree $1$ are $\partial_z^k \c^0$ for $k \ge 0$. We have seen that $\partial_z^k \c^0$ is not closed if $k > 0$.  However,  $\c^0$ is  BRST closed but not BRST exact. This implies that the defect has at most one deformation, given by the descendent of $\c^{0}$. This, in turn, is given by the internal $\int A^{0}$ of the component of the gauge field proportional to $\mu$.  

	As we have already seen, the Witten effect (section \ref{sec:Witten}) tells us that adding this expression to the boundary has the same effect as shifting the spectral parameter.  This tells us that, as desired, the only modification the minuscule 't Hooft line has is that we can shift its spectral parameter.  
\end{proof}

\section{The Koszul dual of the algebra of operators on the 't Hooft
  line is the dominant shifted Yangian}

Our analysis above showed that, classically, the algebra of operators
on the 't Hooft line is the Lie algebra cochains of
$\mf{p} \oplus z \g[[z]]$.  This is in contrast to the algebra of
operators of the bulk system on its own, which is the Lie algebra
cochains of $\g[[z]]$.

We would like to ask what the Koszul dual of the algebra of operators
on the 't Hooft line is.

We do not have space here to go into detail on the role of Koszul
duality in quantum field theory: see for instance
\cite{Costello:2020jbh}.  Denote by $A_{\mbf{H}}$ the dg algebra of
operators on the 't Hooft line, and the Koszul dual algebra by
$A_{\mbf{H}}^!$. Then, for any analytically-continued quantum
mechanical system with algebra of operators $B$, coupling $B$ to the
't Hooft line is the same as giving an algebra homomorphism
$A_{\mbf{H}}^! \to B$.

In \cite{Costello:2013zra} it was proven by an abstract argument that the Koszul dual of the algebra of operators of the bulk system is the Yangian, $Y(\g)$.  Classically, this is easy to see, as the Koszul dual of the Lie algebra cochains of $\g[[z]]$ is the universal enveloping algebra of $\g[[z]]$, which is the classical limit of the Yangian.  The argument at the quantum level presented in \cite{Costello:2013zra} relied on a uniqueness theorem of Drinfeld.  

This was made more explicit in \cite{Costello:2018gyb}, where we studied explicitly the Feynman diagrams which contribute to anomalies to coupling a line defect to four-dimensional Chern-Simons theory.  We found that these anomalies cancel exactly when the algebra $B$ of operators on the line defect has a homomorphism from the Yangian algebra $Y(\g)$.

Here we will prove the following result.
\begin{proposition}
	The Koszul dual of the algebra $A_{\mbf{H}}$ of local operators on a minuscule 't Hooft line is the dominant shifted Yangian $Y_{\mu}(\g)$, associated to the dominant coweight in the same Weyl orbit as the coweight defining the 't Hooft line. 
\end{proposition}
\begin{proof}	 
To see this, let us first look at the classical level. The classical Koszul dual is the universal enveloping algebra of $\mc{p} \oplus z \g[[z]]$.  Physically, this is reasonable: it says that to couple a line defect to the 't Hooft line, we need to be able to couple the components of the gauge field $A$ which are in $\mc{p}$, and all components of the derivatives $\partial_z^k A$ for $k > 0$.  

What happens at the quantum level?  In principle, one could try to reproduce the Feynman diagram analysis of \cite{Costello:2018gyb} in the presence of the 't Hooft line. This, however, could be quite challenging.   Instead, we can use the relationship between the bulk and boundary algebras.

Let us denote the algebra of bulk operators by $A_{Bulk}$. We give
this an ($A_\infty$) structure using the operator product in the
direction wrapped by the 't Hooft line.

By taking an operator in the bulk to one on the defect, we get a homomorphism
\begin{equation} 
	A_{Bulk} \to A_{\mbf{H}}. 
\end{equation}
This leads to a homomorphism of Koszul dual algebras in the other direction:
\begin{equation} 
	A^!_{\mbf{H}} \to A^!_{Bulk} = Y(\g). 
\end{equation}
This homomorphism has a simple explanation in gauge theory terms. If
we have an analytically-continued quantum mechanical system with
algebra of operators $B$, then coupling this system to the bulk gauge
theory is the same as giving a homomorphism $Y(\g) \to B$.  If we do
this, then we can bring the line defect to the 't Hooft line, giving
us a way of coupling the quantum mechanical system to the 't Hooft
line.  This gives us a homomorphism $A^!_{\mbf{H}} \to B$.  Since this
holds for any $B$, we can take $B = Y(\g)$ and so get a natural map
$A^!_{\mbf{H}} \to Y(\g)$.

This map is injective. To see this, it is enough to check it classically. At the classical level,  $A^!_{\mbf{H}}$ is the universal enveloping algebra of $\mf{p} \oplus z \g[[z]]$, and $Y(\g)$ becomes the universal enveloping algebra of $\g[[z]]$.  The homomorphism just comes from the embedding $\mf{p} \oplus z \g[[z]] \into \g[[z]]$, so it is injective.  

It is clear that the fact that $A^!_{\mbf{H}}$ is a subalgebra of $Y(\g)$ constrains it greatly.  As we will see, this suffices to show that the Koszul dual algebra $A^!_{\mbf{H}}$ is the \emph{dominant} shifted Yangian:
\begin{equation} 
	A^!_{\mbf{H}} = Y_{\mu}(\g). 
\end{equation}
Indeed, in the work of \cite{MR3248988}, the dominant shifted Yangian $Y_{\mu}(\g)$ is shown to be subalgebra of $Y(\g)$ with the same classical limit as that of $A^!_{\mbf{H}}$.

We need to show that this is enough to constrain the algebra uniquely.  This seems very plausible: after all, it is hard to lift a subalgebra of $U(\g[[z]])$ to one of the Yangian, and it seems unlikely that there is more than one way to do so.  

Let us explain the simple algebraic argument for this uniqueness. At
the classical level, the subalgebra $A^!_{\mbf{H}}$ contains all the
level one generators $\t^a[1]$ of the Yangian, and those level zero
generators $\t^i[0]$, $\t_\alpha[0]$ corresponding to
$\g_{1} \oplus \g_0$.  The algebra is generated by $\t^i[0]$,
$\t_\alpha[0]$ and $\t_i[1]$, because commutators with $\t^i[0]$
applied to $\t_j[1]$ give the rest of the level $1$ generators.
Commutators of level $1$ generators give level $2$ and higher
generators, so this description completely characterizes
$A^!_{\mbf{H}}$ as a subalgebra of $U(\g[[z]])$, at the classical
level.  This description also holds also for the dominant shifted
Yangian.

	We need to show that there is (up to a shift in the spectral parameter) only one way to lift this to a subalgebra of the quantum Yangian.  At the quantum level, these generators can be modified.  There is a symmetry which rotates $z$ and $\hbar$ simultaneously.  The level $0$ generators are uncharged under this symmetry, and the level $1$ generators have the same charge as $\hbar$.  This means that only level $1$ generators can be modified, and only by adding on $\hbar$ times a polynomial in the level $0$ generators. If we perform such a modification, then $A^!_{\mbf{H}}$ will be generated by the level $0$ generators in $\g_1 \oplus \g_0$, and some modification 
\begin{equation} 
	\t_i[1] + \hbar f_i 
\end{equation}
of the level $1$ generators in $\g_{-1}$. Here $f_i$ is some word in the level zero generators.  Since the algebra has $\mf{p} = \g_0 \oplus \g_1$ symmetry, $f_i$ must transform in $\mf{g}_{-1}$ viewed as a $\mf{g}_0$ representation.

If $f_i$ is simply $\t_i$, we can absorb it into a shift in the spectral parameter, which shifts the level $1$ generators by level $0$ generators.   Therefore  $f_i$ is an expression which is at least quadratic in the level $0$ generators.

Because the $f_i$ must generate a copy of $\mf{g}$ as a $\mf{p}$-representation, we must have 
\begin{equation} 
	[\t^j, [\t^k, [\t^l, f_i]]] = 0.  \label{eqn:nilpotent} 
\end{equation}
One of the terms in the operation of bracketing with $[t^j,-]$ is to
remove $\t_l$ and replace it with $\delta^j_l \t_0$, where $\t_0$
represents the generator corresponding to the coweight $\mu$.  This
implies that $f_i$ can contain at most two generators from $\g_{-1}$,
because if it contained more than two, then \eqref{eqn:nilpotent} can
not hold.

In the case that $f_i$ contains two generators from $\g_{-1}$, then it
can be written as
\begin{equation} 
  f_i = \t^j C^{kl}_{ij} \t_k \t_l ,
\end{equation}
where $C^{kl}_{ij}$ is in the universal enveloping algebra of $\g_0$.
Let us use equation \eqref{eqn:nilpotent} to show that this can not
arise.  First, we note that the only element in
$U(\g_0) \subset U(\g)$ which commutes with all elements of $\g_1$ is
the identity.  From this, it is easy to check that
\eqref{eqn:nilpotent} can only hold if the tensor $C^{kl}_{ij}$ is the
identity, in which case we see, by asking that $f_i$ transforms in the
representation $\g_{-1}$ of $\g_0$, that
\begin{equation} 
	f_i = \t^j \t_j \t_i.  
\end{equation}
An explicit calculation shows that this expression can not satisfy
\eqref{eqn:nilpotent}.

Similarly, if $f_i$ has one generator in $\g_{-1}$, then it is of the form
\begin{equation} 
	f_i = C^j_{i} \t_j ,
\end{equation}
where again $C^j_{i} \in U(\g_0)$.  If $C^{j}_i$ is not the identity,
then, since it does not commute with all elements in $\g_{1}$, we can
exclude this possibility using \eqref{eqn:nilpotent}. If $C^{j}_i$ is
the identity, then $f_i$ is a linear expression and so can be absorbed
by a shift in the spectral parameter.

This completes the proof that the Koszul dual of the algebra of operators on the 't Hooft line is the dominant shifted Yangian.
\end{proof}

\section{Q-operators in integrable field theories}
In this section we will analyze what happens if we study the surface defects of the type considered in \cite{Costello:2019tri} together with an 't Hooft line.  

The main result is a construction of a Q-operator in a class of integrable field theories.  This Q-operator is an interface instead of a line defect, and it satisfies satisfies the TQ relation, to leading order in $\hbar$.  It would be very interesting to pursue this analysis to all orders in $\hbar$, including quantum effects in the integrable field theory, but that is beyond the scope of this work.  

Holomorphic and antiholomorphic surface defects as in
\cite{Costello:2019tri} give rise, when we compactify to on the
$z$-plane, to a two-dimensional integrable field theory.  The
expectation value of the $x$ and $y$ components at $z \in \C$ is the
Lax matrix:
\begin{equation} 
	\mc{L}(z) = \ip{A(z)}. 
\end{equation}
We will focus on the class of defects called \emph{order} defects in
\cite{Costello:2019tri}, where we couple chiral or antichiral degrees
of freedom at various values of $z$. Typically, these auxiliary
systems are given by some system of fermions, or some $\beta$-$\gamma$
systems.  The details of the systems will not matter for our analysis,
however.

We will let $J_{(i)}$ denote the currents in the $i$'th chiral or
antichiral system, placed at $z_i$. The Lax matrix is then
\begin{equation} 
	\mc{L}^a(z) = \sum c^{ab} \frac{1}{z - z_i} J_{b,(i)}, 
\end{equation}
where $c^{ab}$ is the inverse of the quadratic Casimir. 

Each order defect has a Lagrangian $S_{(i)}$, typically that of a free
chiral or antichiral theory. Integrating out the gauge fields of
four-dimensional Chern-Simons theory couples these order defects
giving a Lagrangian of the form
\begin{equation} 
	\frac{1}{\hbar} \sum_{i} S_{(i)} + \frac{1}{\hbar} \sum_{i,j} \frac{1}{z_i - z_j} J_{(i)} \wedge J_{(j)}  ,
\end{equation}
where we treat the current in a chiral theory as a $(1,0)$ form and an
antichiral theory as a $(0,1)$-form, so that the
$J_{(i)} \wedge J_{(j)}$ term vanishes unless one defect is chiral and
the other antichiral.

The simplest models of this form have only two defects, one a system
of chiral fermions and one a system of antichiral fermions. In this
case, the model is the Gross-Neveu or Thirring model, or some similar
model with a $\psi \psi \br{\psi} \br{\psi}$ interaction.

\subsection{Coupling an integrable field theory to the topological $\sigma$-model}
View an 't Hooft line at $z$ as living at the boundary of the
topological surface defect wrapping the region $y \le 0$.  Then, the
analysis of \cite{Costello:2019tri} immediately implies that, in the
effective two-dimensional theory, the topological surface defect is
coupled to the integrable field theory using the Lax operator
$\mc{L}(z)$.  Concretely, if $V_a$ are the vector fields on $G/P$
generating the $\mf{g}$ action, then the coupling is by
\begin{equation} 
	\frac{1}{\hbar} \int_{x, y \ge 0} \ip{\eta, V_a} \wedge \mc{L}^a(z). 
\end{equation}
(The gauge transformations of the topological $\sigma$-model are also
modified, as we will see explicitly below in the case $G = SL_2$.)

This modification of the integrable field theory plays the role of the
Q-operator.  It is an interface between the original integrable field
theory and the theory coupled to the topological $\sigma$-model on
$G/P$.  The main goal of this section is to justify this, by deriving
the TQ relation to leading order in $\hbar$.

Take $G = SL_2$, and take the basis $e$, $f$, $h$ of $SL_2$ given
$[e,f] = h$, $[h,e] = 2 e$, $[h,f] = -2 f$.  Explicitly,
\begin{equation} 
  e = \begin{pmatrix} 0 & 1\\ 0 &0   \end{pmatrix},
  \qquad
  f = \begin{pmatrix} 0 & 0\\ 1 &0   \end{pmatrix},
  \qquad h = \begin{pmatrix} 1 & 0\\ 0 & -1   \end{pmatrix} .
\end{equation}
The flag variety $G/P$ in this case is $\CP^1$ and we will work in a
neighbourhood of the origin. The topological $\sigma$-model has two
fields $\eta$, $\sigma$ and the Lagrangian is
\begin{equation} 
	\frac{1}{\hbar}\int_{x,y \ge 0} \left(\eta \wedge \d \sigma   +  \mc{L}^e(z) \wedge \eta   - 2 \mc{L}^h(z)\wedge  \eta \sigma - \mc{L}^f(z) \wedge \eta \sigma^2 \right). \label{eqn_ift_lagrangian} 
\end{equation}
In our basis, we have
\begin{equation} 
  \mc{L}^e(z) = \sum_i \frac{1}{z - z_i} J_{f,(i)},
  \quad
  \mc{L}^h(z) = \tfrac{1}{2}\sum_i \frac{1}{z - z_i} J_{h,(i)},
  \quad
  \mc{L}^f(z) = \tfrac{1}{2}\sum_i \frac{1}{z - z_i} J_{e,(i)} .
\end{equation}

If we denote the parameter for gauge transformations of $\eta$ as
$\chi$, then the gauge transformations of $\eta$ are modified by
\begin{equation} 
	\delta \eta = \d \chi - 2 \chi \mc{L}^h(z) - 2 \chi \mc{L}^f(z) \sigma. 
\end{equation}

Further, the fields of the integrable field theory also transform under the gauge transformation.  Consider the system at $z_i$.  The fields of this system couple to the $1$-form $\eta$ by the current 
\begin{equation} 
	\frac{1}{z - z_i}  \left(   J_{f,(i)}    -  J_{h,(i)} \sigma - J_{e,(i)}  \sigma^2 \right). 
\end{equation}
Since $\eta$ is the gauge field associated to the gauge transformation $\chi$, then the fields in the system at $z_i$ transform under $\chi$ according to this current.

We note that the Lagrangian is gauge invariant.  To see this, we note
our conventions are such that
\begin{equation} 
  \d \mc{L}^e = -2 \mc{L}^h \wedge \mc{L}^e,
  \qquad
  \d \mc{L}^h = - \mc{L}^e \wedge \mc{L}^f,
  \qquad
  \d \mc{L}^f = 2 \mc{L}^h \wedge \mc{L}^f. 
\end{equation}
Then gauge variation has two terms: variation of the field $\eta$, and
variation of the fields in the chiral or antichiral theories we are
coupling to.  Each term vanishes independently.  The variation of
$\eta$ gives (after integrating by parts and using the fact that the
gauge parameter vanishes on the boundary)
\begin{multline}
    (- 2 \chi \mc{L}^h -2 \chi \sigma \mc{L}^f) \d \sigma  
    + \chi  \d \mc{L}^e - 2\chi \mc{L}^e (\mc{L}^h + \sigma \mc{L}^f ) \\ 
    - 2 \chi \d(\sigma \mc{L}^h) + 4 \chi \sigma^2 \mc{L}^h  \mc{L}^f 
    - \chi \d (\mc{L}^f \sigma^2) + 2 \chi \sigma^2 \mc{L}^f  \mc{L}^h.
\end{multline}
This expression simplifies to 
\begin{equation} 
  \chi \d \mc{L}^e - 2 \chi \mc{L}^e \wedge \mc{L}^h -2 \chi \sigma \d \mc{L}^h - 2 \chi \sigma \mc{L}^e \wedge \mc{L}^f - \chi \sigma^2 \d \mc{L}^f + 2 \chi \sigma^2 \mc{L}^h \wedge \mc{L}^f   ,
\end{equation}
which vanishes by the Lax equation.   

The variation of the field in the defect at $z_i$ gives
\begin{equation} 
  \frac{1}{(z-z_i)^2} \eta \chi
  [J_{f,(i)}  - \sigma J_{h,(i)} - \sigma^2 J_{f,(i)},
  J_{f,(i)}  - \sigma J_{h,(i)} - \sigma^2 J_{f,(i)} ]  ,
\end{equation}
which vanishes. 

\subsection{The interface is conserved}

So far, we have explained how to build a gauge-invariant coupling
between the integrable field theory and the new topological degrees of
freedom.  We will make the topological system end at $y = 0$, with
boundary condition $\eta = 0$ (we also require that the gauge
parameter $\chi$ vanishes on the boundary).  This gives an interface
between the original integrable field theory, and that with the
topological fields added.

We would like this interface to play the role of the Q-operator.  For
this to happen, we need to first verify that this interface is
conserved. This amounts to checking that we get the same system if,
instead of placing the topological degrees of freedom on the region
$x,y$ where $y \ge 0$, we place them on the region $x,y$ where
$y \ge \eps$.

We can work in an axial gauge for the gauge field of the Poisson
$\sigma$-model. This sets the $y$-component $\eta_y$ to zero. The
boundary condition is that $\eta_x = 0$.  If we modify the location of
the boundary, then to first order in $\eps$ the boundary condition
becomes $\eta_x + \eps \partial_y \eta_x = 0$.

If we modify the location of the boundary condition, we add on the term 
\begin{equation}
\int_x	\int_{y = 0}^{\eps} \left(\eta \wedge \d \sigma   +  \mc{L}^e(z) \wedge \eta   - 2 \mc{L}^h(z)\wedge  \eta \sigma - \mc{L}^f(z) \wedge \eta \sigma^2 \right).	
\end{equation}
To leading order in $\eps$, this is 
\begin{equation} 
	\eps \int_{y = 0} \iota_{\partial_y}  \left(\eta \wedge \d \sigma   +  \mc{L}^e(z) \wedge \eta   - 2 \mc{L}^h(z)\wedge  \eta \sigma - \mc{L}^f(z) \wedge \eta \sigma^2 \right)	,
\end{equation}
where $\iota_{\partial_y}$ indicates contraction with this vector
field.  This vanishes to leading order in $\eps$, because our gauge
choice is $\eta_y = 0$ so that we need only consider terms that only
depend on $\eta_x$; by $\eta_x$ is proportional to $\eps$ by the
boundary condition $\eta_x + \eps \partial_y \eta_x = 0$.

To finish the analysis, we need to check that the modification of the boundary condition to $\eta_x + \eps \partial_y \eta_x = 0$ has no effect. The equations of motion for $\eta_x$ are (using the gauge condition that $\eta_y = 0$)
\begin{equation} 
	\partial_y \eta_x = \eta_x (\mc{L}^e_y - 2 \sigma \mc{L}^h_y - \sigma^2 \mc{L}^f_y) . 
\end{equation}
Since this is divisible by $\eta_x$, we see that the condition $\eta_x + \eps \partial_y \eta_x = 0$ is equivalent to the condition that $\eta_x = 0$.

We have shown that the interface is conserved, meaning that we can
move its location freely.  In that sense, it behaves like the
T-operator which we will now discuss.

\subsection{The TQ relation}

The T-operator is the path ordered exponential of the Lax matrix:
\begin{equation} 
	\mbf{T}(z) = \op{PExp} \int_{y=0} \begin{pmatrix}
		\mc{L}^h(z) & \mc{L}^e(z) \\
		\mc{L}^f(z) & - \mc{L}^h(z) 
	\end{pmatrix}.
\end{equation} 
The T-operator is also conserved, because of the flatness of the Lax
connection. (Note that there is no factor of $\hbar$ here, unlike in
the Lagrangian of the topological surface defect or the integrable
field theory.)

Let us consider placing the T-operator at the same location as the
Q-operator, that is, at the boundary of the topological surface
defect.

Varying the field $\eta$ in \eqref{eqn_ift_lagrangian} tells us that
\begin{equation} 
	\mc{L}^e(z) =  \d \sigma + 2 \mc{L}^h(z) \sigma + \mc{L}^f(z) \sigma^2.  
\end{equation}
When the transfer matrix is placed in the same location as the Q-operator, we can $\mc{L}^e(z)$ in the transfer matrix by this expression. This gives us 
\begin{equation} 
	\mbf{T}(z)\mbf{H}(z) = \op{PExp} \int_{y=0} \begin{pmatrix}
		\mc{L}^h(z) &   \d \sigma + 2 \mc{L}^h(z) \sigma + \mc{L}^f(z) \sigma^2   \\
		\mc{L}^f(z) & - \mc{L}^h(z) 
	\end{pmatrix}.
\end{equation}
Here, $\mbf{H}(z)$ is the interface given by the boundary of the
surface defect with fields $\sigma, \eta$.

The T-operator is the parallel transport of the connection given by the matrix in the equation, and as such it is invariant under gauge transformations.   There is a gauge transformation  
\begin{equation}
  \begin{split}
    & \begin{pmatrix}
      \mc{L}^h(z) &   \d \sigma + 2 \mc{L}^h(z) \sigma + \mc{L}^f(z) \sigma^2   \\
      \mc{L}^f(z) & - \mc{L}^h(z)  
    \end{pmatrix} 
    \\&	= \d \begin{pmatrix} 1 &  \sigma \\ 0 & 1 \end{pmatrix}  +  \begin{pmatrix} 1 & - \sigma \\ 0 & 1 \end{pmatrix} \begin{pmatrix} \mc{L}^h(z) + \sigma \mc{L}^f(z)  &  0  \\ \mc{L}^f(z) & -\mc{L}^h(z) - \sigma \mc{L}^f(z)  \end{pmatrix}    \begin{pmatrix} 1 &  \sigma \\ 0 & 1 \end{pmatrix}   .
  \end{split}
\end{equation}
Using this gauge transformation, the T-operator (in the presence of
the 't Hooft line) becomes
\begin{equation} 
	\mbf{T}(z) \mbf{H}(z) = \mbf{H}(z)  \op{PExp} \begin{pmatrix} \mc{L}^h(z) + \sigma \mc{L}^f(z)  &  0  \\ \mc{L}^f(z) & -\mc{L}^h(z) - \sigma \mc{L}^f(z)  \end{pmatrix}   .
\end{equation}
(On the right hand side we include the factor $\mbf{H}(z)$ to
emphasize that the path-ordered exponential is taking place at the
boundary of the topological surface defect.)

If we make the $x$-direction periodic, and take the trace of the
transfer matrix, we can rewrite this as
\begin{equation} 
	\mbf{T}(z) \mbf{H}(z) = \mbf{H}(z)  \exp\left( \int   \mc{L}^h(z) + \sigma \mc{L}^f(z)\right) +  \mbf{H}(z) \exp\left(  \int  - \mc{L}^h(z) - \sigma \mc{L}^f(z)\right) . 
\end{equation}
On the right hand side, we have two copies of the 't Hooft line,
dressed by the operators $\pm( \mc{L}^h(z) + \sigma \mc{L}^f(z) )$.

This is very close to the TQ relation, with the normalization we are using.   The operators $\pm(\mc{L}^h(z) + \sigma \mc{L}^f(z))$ are problematic, however.  The final step is to show that we can trade these operators at the boundary for a variation in the parameter $z$. 

\subsection{Varying $z$}

Let us consider the effect of varying the position $z$ of the
topological surface defect.  We will see that this can be compensated
for by a field redefinition.

Varying $z$ modifies the coupling between the topological theory and
the integrable field theory by adding the term
\begin{equation} 
  \int \eta \left( \partial_z \mc{L}^e(z) - 2 \sigma \partial_z \mc{L}^h(z) - \sigma^2 \partial_z \mc{L}^f(z) \right).   
\end{equation}
(The gauge variation of the chiral and antichiral fields is modified
in a similar way.) We will show that this can be compensated for by a
field redefinition.

On the chiral or antichiral defect at $z_i$, we can perform a field
redefinition given by the global $\mf{sl}_2$ symmetry.  We will label
the field redefinition by the associated current $J_{e,(i)}$,
$J_{f,(i)}$, $J_{h,(i)}$.  We are interested in these field
redefinitions as first-order perturbations of the identity.

For instance, if we have a set of complex chiral fermions at $z_i$ in
the fundamental representation of $\mf{sl}_2$, then
$1 + \eps J^e_{(i)}$ sends $\psi_1 \to \psi_1 + \eps \psi_2$,
$\psi^2 \to \psi^2 - \eps\psi^1$.

We can also make these field redefinitions depend on the field
$\sigma$ of the topological surface defect. For example, the field
redefinition associated to $1 + \eps \sigma J^e_{(i)}$ sends
$\psi_1 \to \psi_1 + \eps \sigma \psi_2$,
$\psi^2 \mapsto \psi^2 - \eps \sigma \psi^1$.

Varying $z$ to $z + \eps$ can be compensated for by the field
redefinition
\begin{equation} 
  1 +  \eps \left( \sum \frac{1}{z - z_i} (\tfrac{1}{2} J_{h,(i)} + \sigma J_{e,(i)})    \right).
\end{equation}
Under this field redefinition, the components of the Lax matrix vary as:
\begin{equation}
	\begin{split}
		\delta \mc{L}^e(z) &= \delta \sum \frac{1}{z - z_i} \delta J_{f,(i)} = \sum \eps \frac{1}{(z - z_i)^2} \left( -J_{f,(i)} +  \sigma J_{h,(i)}   \right)  = \eps \partial_z \mc{L}^e(z) - 2 \eps \sigma \partial_z \mc{L}^h(z), \\
		\delta \mc{L}^h(z) &= \tfrac{1}{2} \delta \sum \frac{1}{z - z_i} \delta J_{h,(i)} =  - \sum \eps \frac{1}{(z - z_i)^2}  \sigma J_{e,(i)} =  \eps \sigma \partial_z \mc{L}^f(z),  \\
		\delta \mc{L}^f(z) &= \delta \sum \frac{1}{z - z_i} \delta J_{e,(i)} =   \sum \eps \frac{1}{(z - z_i)^2}  J_{e,(i)} = -\eps \partial_z \mc{L}^f(z)  .
	\end{split}
\end{equation}
Thus,  this field redefinition sends
\begin{equation} 
	\mc{L}^e(z) - 2\sigma \mc{L}^h(z) - \sigma^2 \mc{L}^f(z) \mapsto  \mc{L}^e(z+\eps) - 2\sigma \mc{L}^h(z+\eps ) - \sigma^2 \mc{L}^f(z+\eps)  . 
\end{equation}
That is, it has the same effect as varying $z$.  

This shows us that the topological surface defect, when coupled to the integrable field theory, does not change when we vary $z$.    

What happens when the surface defect has a boundary?  In that case, in
the bulk of the surface defect, we can compensate for the variation of
$z$ by the field redefinition described above.  However, this
introduces an extra term in on the boundary.

Let us see this concretely in the case when we perform the required
field redefinition to a system of chiral fermions $\psi_1$, $\psi_2$,
$\psi^1$, $\psi^2$.  The transformation corresponding to the element
$e \in \mf{sl}_2$ sends $\psi_1 \to \psi_2$, $\psi^2 \to -\psi^1$.
This, of course, preserves the Lagrangian
$\frac{1}{\hbar} \int \psi_i \dbar \psi^i$.  When the surface defect
has a boundary, we only need to perform the field redefinition on the
region $y \ge 0$.  This sends $\psi_1 \to \delta_{y \ge 0} \psi_2$,
$\psi^2 \to - \delta_{y \ge 0} \psi^1$.  This does not preserve the
Lagrangian: we pick up a term
\begin{equation} 
	-\frac{1}{\hbar} \int \psi_2 (\dbar \delta_{y \ge 0}) \psi^1 = \frac{1}{\hbar} \int_{y = 0} \psi_2 \psi^1 = \frac{1}{\hbar} \int_{y = 0} J_{e}.  
\end{equation}

In the same way, if we perform the field redefinition associated to
any current $J$ in the region $y \ge 0$ in the chiral or antichiral
theories that we place at $z_i$, we pick up a boundary term given by
$\hbar^{-1} \int_{y = 0} J$.

This tells us that if, in the bulk of the surface defect, we vary $z$
to $z + \hbar$, we can compensate for this variation by the field
redefinition
$\hbar \left( \sum \frac{1}{z - z_i} (\tfrac{1}{2} J_{h,(i)} + \sigma
  J_{e,(i)}) \right)$, but that this also picks up a boundary term of
the form
\begin{equation} 
  \int_{y = 0}   \left( \sum \frac{1}{z - z_i} (\tfrac{1}{2} J_{h,(i)} + \sigma J_{e,(i)})    \right)
  = \int_{y = 0} \left(\mc{L}^h(z) + \sigma \mc{L}^f(z)\right).
\end{equation}
Note that the factors of $\hbar$ have cancelled out, and we have shown that
\begin{equation}
  \mbf{H}(z + \hbar) = \mbf{H}(z) \exp \left(  \int_{y = 0} \mc{L}^h(z) + \sigma \mc{L}^f(z) \right). 
\end{equation}
(Our analysis here is valid to leading order in $\hbar$.)

Finally, this allows us to prove the TQ relation:
\begin{equation} 
\begin{split}
	\mbf{T}(z) \mbf{H}(z) &= \mbf{H}(z)  \exp\left( \int   \mc{L}^h(z) + \sigma \mc{L}^f(z)\right) +  \mbf{H}(z) \exp\left(  \int  - \mc{L}^h(z) - \sigma \mc{L}^f(z)\right) \\ 
	&= \mbf{H}(z + \hbar) + \mbf{H}(z - \hbar)
\end{split}
\end{equation}
to leading order in $\hbar$.  

We have derived this for integrable field theories obtained as order
defects in four-dimensional Chern-Simons theory.  The analysis is
valid, with some small variations, in the trigonometric and elliptic
cases also.  However, a more involved analysis will be necessary to
study the Q-operators when we have disorder defects.

\section{Composing 't Hooft lines, QQ relations and  and parabolic Verma modules} \label{sec:QQ}
The most fundamental relation among the T and Q operators is the QQ relation, which expresses the product of two Q operators in terms of the T-operator associated to a Verma module.   Here we will proof the 't Hooft line version of this relation.  If $\mu$ is a minuscule coweight, we let $\mbf{H}_{\mu}(z)$ be the corresponding 't Hooft line, defined as before in terms of a surface defect for $G/P$ with a boundary condition.

In this section we will relate the product of two 't Hooft lines with a certain analytically-continued quantum mechanics.  The quantum mechanics lives on $T^\ast G/P$ viewed as a holomorphic symplectic manifold.   In analytically-continued quantum mechanics,  we have access to the algebra of operators of the quantum mechanical system, but we have not specified which representation of this algebra will be the Hilbert space.   The algebra of operators is what mathematicians would call $\op{Diff}(G/P, K^{1/2})$, the algebra of holomorphic differential operators on $G/P$ twisted by $K^{1/2}$. The factor of $K_X^{1/2}$ is included to ensure that the system is symmetric under time-reversal.   We include a factor of $\hbar$ in our definition of $\op{Diff}(G/P, K^{1/2})$, viewing it as a deformation quantization of $T^\ast G/P$.  The factor of $\hbar$ can be removed by redefinition of the generators. 

The system admits a deformation, where the algebra of operators is $\op{Diff}(G/P, K^{c})$.  This deformation comes from a holomorphic symplectic deformation of the target manifold $T^\ast G/P$.  We will describe this algebra explicitly momentarily. 

We will let $\mbf{T}^{\mu}_j (z)$ be the analytically-continued line defect in four-dimensional Chern-Simons theory obtained by coupling analytically-continued quantum mechanics, with algebra of operators $\op{Diff}(G/P, K^{c})$, at the point $z$ in the spectral parameter plane. This algebra acts on a parabolic Verma module for $\mf{g}$ with highest weight $j = c \op{dim} \g_1$; it is more convenient to parametrize in terms of $j$ rather than $c$.  

Here we will show the following proposition.
\begin{proposition}
For any minuscule coweight $\mu$ of any simple group, There is an isomorphism of analytically-continued line defects
	\begin{equation} 
		\mbf{H}_{-\mu}(-a \hbar) \mbf{H}_{\mu}(a \hbar) \iso \mbf{T}^{\mu}_{\op{dim} \g_1\left(-\tfrac{1}{2} + \tfrac{2 a}{ \sh^\vee} \right)   }(0) ,
	\end{equation}
where $\sh^\vee$ is the dual Coxeter number.
\end{proposition}

\subsection{Partial flag varieties and oscillator algebras}

To make the proposition more concrete, let us describe in more detail
the algebra of twisted differential operators on $G/P$, and how it
relates to oscillator algebras.  As usual, write
$\g = \g_1 \oplus \g_0 \oplus \g_{-1}$ where
$\mf{p}= \g_0 \oplus \g_{1}$. Let $U\subset G/P$ be the open subset
which is the orbit of the base point under $\exp(\g_{-1})$.  If $q_i$,
$p^i$ are bases of $\g_{1}$, $\g_{-1}$, the algebra $\op{Diff}(U)$ is
generated by variables $q_i$, $p^i$ with canonical commutation
relations
\begin{equation} 
	[p^i, q_j] = \hbar \delta^i_j.	
\end{equation}
For every complex number $c$, there is an injective homomorphism
\begin{equation} 
	\op{Diff}(G/P, K^{c}) \to \op{Diff}(U) 
\end{equation}
defined, geometrically, by restricting a global differential operator
to $U$ and using a trivialization of the bundle $K$ on $U$.

A natural Fock representation of $\op{Diff}(U)$ on $\C[p^i]$, where
$q_i$ acts as $-\hbar \partial_{p^i}$, gives rise to a representation
of $\op{Diff}(G/P, K^{c})$.  This representation is a parabolic Verma
module.  It has a highest weight vector annihilated by $\mf{g}_{1}$ in
a rank one representation of $\mf{g}_0$.  It will be convenient to use
the highest weight of the parabolic Verma module, rather than the
twist parameter $c$, as our parameter.

To describe $\op{Diff}(G/P, K^{c})$ as a subalgebra of the oscillator
algebra, we will choose a basis of $\g$ compatible with the
decomposition $\g = \g_1 \oplus \g_0 \oplus \g_{-1}$.  Our basis
elements will be $\t_i$, $\t^i$ for $\g_{-1}$,$\g_1$, and $\t_\alpha$
for the orthogonal complement of $\mu$ in $\mf{g}_0$, and $\t_{0}$
corresponding to the coweight $\mu$.  The commutation relations are
\begin{equation} 
	\begin{split}
		[\t^i, \t_j] &= -\frac{1}{\ip{\mu,\mu}} \delta^i_j \t_0 + f^{i\alpha}_{j} \t_\alpha, \\
		[\t^0, \t^i] &= -\t^i, \\
		[\t^0, \t_i] &=  \t_i, \\
		[\t_\alpha, \t^i] &= c_{\alpha \beta} f^{i\beta}_{j} \t^j, \\
		[\t_\alpha, \t_i] &= - \c_{\alpha \beta} f^{j \beta}_{i} \t_j.  
	\end{split}
\end{equation}
Here $c_{\alpha \beta}$ is invariant pairing on $\mf{g}$.  

These act on symplectic manifold with coordinates $q_i$, $p^i$ by the Hamiltonians 
\begin{equation} 
	\begin{split}
		\t^i & \mapsto   p^i, \\
		\t_{0} & \mapsto   -q_i p^{i} + \hbar a, \\
		\t_\alpha & \mapsto - c_{\alpha \beta} f_{i}^{\beta j} p^i q_{j}, \\
		\t_i & \mapsto  -\frac{1}{\ip{\mu,\mu}} \hbar a  q_i -  C^{jk}_{il}  q_j q_k p^l 
	\end{split} \label{eqn:parabolic_coupling}
\end{equation}
for some tensor $C$ whose precise form doesn't matter. Here $a$ is the highest weight of the parabolic Verma module.

We can write the Lagrangian for the analytically-continued Wilson line as
	\begin{equation} 
		\frac{1}{\hbar} \int_{x} p^i \d q_i + \frac{1}{\hbar}   \int_{x} A_x^a H_a(p,q) ,
	\end{equation}
where $H_a$ are the Hamiltonians written above.

If we multiply all the  Lie algebra generators by $1/\hbar$, this gives a homomorphism from $\mf{g}$ to the oscillator algebra. The algebra of twisted differential operators on $G/P$ is the subalgebra of the oscillator algebra generated by $\mf{g}$, where the parameter $a$ encodes the twist.

The only parameter in this expression is the constant $a$. We can
transform the system by sending $p \to -p$ and reversing $x \to -x$,
which reverses the order of the line defect. This sends
$a \mapsto -a-\op{dim} \g_1$.  To see this, we note that sending
$x \to -x$ reverses the ordering of the operators, and so sends
\begin{equation} 
  -q_i p^i \mapsto p^i q_i = q_i p^i + \hbar(\op{dim} \mf{g}_1). 
\end{equation}
(Similarly for the Hamiltonian associated to $t_i$ which also depends
on $a$.)  We also pick up a sign from reversing the order of the
integral along $x$, leading to the transformation
\begin{equation} 
  -q_i p^i + \hbar a \mapsto  -q_i p^i - \hbar (  \op{dim} \mf{g}_1 +a ). 
\end{equation}

From this, we see that Wilson line is invariant under sending
$x \to -x$ has $a = -\tfrac{1}{2} \op{dim} \mf{g}_1 $.  When
describing the algebra as twisted differential operators on $G/P$,
this corresponds to the twist by $K^{1/2}$, which is self-dual.

\subsection{Colliding 't Hooft lines}
Now let us turn to the proof of the proposition.  As before, at $z=z_0$, we couple to the $\sigma$-model at $G/P$ wrapping $y \le 0$, giving the 't Hooft line $\mbf{H}_{\mu}(z_0)$. 

	We would like to compare this to the same surface wrapping $y \ge 0$. This gives the 't Hooft line of charge $-\mu$ at $z_0$. To see this, we recall that the presence of the surface defect is equivalent to allowing certain singularities in the gauge field. Having the surface defect at $y \le 0$ gives us a Hecke transformation removing the singularities, and having it at $y \ge 0$ creates them. 

Consider the configuration where we have the $G/P$ surface defect at $z_0$, in the range $-\eps \le y \le \eps$, with the same boundary conditions as before at $y = \pm \eps$. This configuration is given by the composition of line defects 
\begin{equation} 
	\mbf{H}_{-\mu}(z_0)  \mbf{H}_{\mu}(z_0). 
\end{equation}
By applying a gauge transformation, we can map this to a configuration
where the surface defects are at $y \le -\eps$, $y \ge \eps$.
Therefore, to understand the composition of 't Hooft lines, we have to
understand the simpler question of the effective quantum mechanical
system obtained from our topological $\sigma$-model on an interval of
small width.

In general, consider the Poisson $\sigma$-model with target $X$ on the
strip $[-\eps,\eps] \times \R$, with Neumann boundary conditions at
both ends, and with zero Poisson tensor.  As above, coordinates are
$y$, $x$ with $-\eps \le y \le \eps$. Neumann boundary condition means
that the field $\eta$, and its gauge transformation $\chi$, both
vanish at the boundary.

It is not difficult to check that the resulting model is equivalent to
topological quantum mechanics on $T^\ast X$.  Recall that the fields
are $\sigma\colon [-\eps,\eps] \times \R \to X$ and
$\eta \in \Omega^1([-\eps,\eps] \times \R, \sigma^\ast T^\ast X)$. We
note that varying $\eta_x$ tells us that $\partial_y \sigma = 0$, so
that $\sigma$ can be viewed as a map $\sigma\colon \R \to X$. If we
vary $\sigma$ by adding on a term that only depends on $y$, we find
that $\eta_x = 0$.  By a gauge transformation, we can set $\eta_y$ to
be $\eta^0 \delta_{y = 0}$, for some
$\eta^0 \in \Omega^0(\R, \sigma^\ast T^\ast X)$.  The resulting action
is $\int_{\R} \eta^0 \d \sigma$, as desired.

Applying this to our situation, we find that $\mbf{H}_{-\mu}(z_0) \mbf{H}_{\mu}(z_0)$ is equivalent to a Wilson line where we couple topological quantum mechanics on $T^\ast G/P$.

This analysis is valid at the classical level. To lift it to the quantum level, we need to know what possible deformations this Wilson line has.  In appendix \ref{appendix:parabolicverma}  we prove the following. 
\begin{theorem}
	The analytically-continued Wilson line given by topological quantum mechanics on $T^\ast G/P$ exists at the quantum level. Furthermore, this line defect has only two deformations, one given by shifting the position in the $z$-plane, and one by replace differential operators by differential operators twisted by $K^{c}$.	
\end{theorem}
This is proved by simply computing the cohomology groups which describe obstructions to quantizing the Wilson line and  deformations of the Wilson line. 

From this, we conclude that 
\begin{equation} 
	\mbf{H}_{-\mu}(0) \mbf{H}_{\mu}(0) = \mbf{T}^{\mu}_{j}(\lambda) 
\end{equation}
for some values of $\lambda,j$.   More generally, we have
\begin{equation} 
  \label{eqn_composite_param} 
  \mbf{H}_{-\mu}(-\hbar a) \mbf{H}_{\mu}(\hbar a)= \mbf{T}^{\mu}_{j}(\lambda),
\end{equation}
where $j$, $\lambda$ depend on $a$.

We can fix the parameter $\lambda$  by symmetries.  Recall that all the line defects are at $y = 0$. Consider the symmetry of the system which sends $y \to -y$ and $z \to -z$.   This sends 
\begin{equation} 
  \mbf{H}_{-\mu}(-\hbar a) \mbf{H}_{\mu}(\hbar a) \to \mbf{H}_{\mu}(-\hbar a) \mbf{H}_{-\mu}(\hbar a)
\end{equation}
and
\begin{equation} 
	\mbf{T}^{\mu}_{j}(\lambda) \to \mbf{T}^{\mu}_{j}(-\lambda). 
\end{equation}
We conclude that 
\begin{equation} 
		\mbf{H}_{\mu}(-\hbar a) \mbf{H}_{-\mu}(\hbar a) = \mbf{T}^{\mu}_{j}(-\lambda) . 
\end{equation}
In the case when $\mu$ and $-\mu$ are in the same Weyl orbit, $\mbf{H}_{\mu} = \mbf{H}_{-\mu}$ by a gauge transformation.  Therefore,
\begin{equation} 
	\mbf{H}_{-\mu}(-\hbar a) \mbf{H}_{\mu}(\hbar a) = \mbf{T}^{\mu}_{j}(-\lambda) = \mbf{T}^{\mu}_{j}(\lambda)  
\end{equation}
by \eqref{eqn_composite_param}.  This tells us that $\lambda = 0$.

In certain cases, $\mu$ and $-\mu$ are not in the same Weyl orbit.  This happens for the spinor coweights of $\mf{so}(4n)$ and the minuscule coweights of $\mf{sl}(n)$ with $n > 2$.  In this case, there is an outer automorphism which switches $\mu$ and $-\mu$.  Further, the analytically-continued line defects $\mbf{T}^{\mu}$ and $\mbf{T}^{-\mu}$ are actually isomorphic in these cases. In these cases, we can also conclude that $\lambda = 0$.

Next, we need to fix the value of the constant $j$ in
\eqref{eqn_composite_param}. Let us first take $a = 0$.  The surface
defect wrapping $-\eps \le y \le \eps$ at $z = 0$ is invariant under
the symmetry which sends $x \to -x$ and which sends
$\eta \mapsto -\eta$ where, as above, $\eta$ is the one-form field on
the surface defect.  This symmetry preserves the boundary conditions
at $y = \pm \eps$, and also extends to a symmetry of the coupled 2d-4d
system, where we also send $z \mapsto -z$.

When we send $\eps \to 0$ to turn the surface defect into a line
defect, the line defect is invariant under the symmetry sending
$x \to -x$ and $z \mapsto - z$. This switches the orientation of the
line, and so sends the algebra of operators to its opposite.  This
fixes $j =- 1/2 \op{dim} \g_1$, as this is the only value of $j$ for
which the line defect is isomorphic to its opposite.

We have proved, as desired, that
\begin{equation} 
	\mbf{H}_{-\mu}(0) \mbf{H}_{\mu}(0) = \mbf{T}_{-1/2 \op{dim} \g_1}(0).  
\end{equation}

If we include the shifts in the 't Hooft line, we must have
\begin{equation} 
	\mbf{H}_{-\mu}(-\hbar a) \mbf{H}_{\mu}(\hbar a) = \mbf{T}^{\mu}_{-\tfrac{1}{2}\op{dim} \g_{1} + F(a)} (0) 
\end{equation}
for some function of $a$.  The transformation $z \mapsto -z$,
$x \mapsto -x$ leads to the equality
\begin{equation} 
	\mbf{H}_{-\mu}(\hbar a) \mbf{H}_{\mu}(-\hbar a) = \mbf{T}^{\mu}_{-\tfrac{\op{dim} \g_1}{2} - F(a)} (0) 
\end{equation}
so that $F$ is an odd function of $a$.  

Further, the system with 't Hooft lines has a symmetry which scales $z$ and $\hbar$ simultaneously.  For this symmetry to extend to the line defect $\mbf{T}^{\mu}_{-\tfrac{1}{2}\op{\dim} \g_1 + F(a)} (0)$, the function $F(a)$ must be linear in $a$, leading to the conclusion. 

To fix the coefficient, we use the framing anomaly. The 't Hooft line
is invariant under the reflection $x \to -x$ which reverses the
ordering along the line.  As shown in \cite{Costello:2017dso}, for any
line defect, to construct the dual line, we first reflect and then
shift in the $z$ plane by $\sh^\vee/2$.  There is an interface linking
the trivial line with the fusion of the original line and its dual.

From this we see that
$\mbf{H}_{-\mu}(-\hbar \sh^\vee/4) \mbf{H}_{\mu}(\hbar \sh^\vee/4)$
has an interface with the trivial line.  Since
\begin{equation} 
	\mbf{H}(-\hbar \sh^\vee/4) \mbf{H}(\hbar \sh^\vee/4) = \mbf{T}^{\mu}_{-\tfrac{\op{dim} \g_1}{2} - F(\sh^\vee/4)} ,
\end{equation}
the right hand side must have an interface with the trivial line. This only happens when the highest weight is zero,  so  $F(\sh^\vee/4) = - \tfrac{\op{dim} \g_1}{2}$, so that in general
\begin{equation} 
	\mbf{H}_{-\mu}(-\hbar a) \mbf{H}_{\mu}(\hbar a) = \mbf{T}^{\mu}_{-\tfrac{\op{dim} \g_1}{2} + \tfrac{2 a \op{dim} \g_1} { \sh^\vee}   } (0). 
\end{equation}

In particular, for $\mf{g} = \mf{sl}_2$, we have $\sh^\vee = 2$ and $\op{dim} \g_1 = 1$, so that
\begin{equation} 
	\mbf{H}_{-\tfrac{1}{2}}(-\hbar a) \mbf{H}_{\tfrac{1}{2}}(\hbar a) = \mbf{T}^{+}_{a - \tfrac{1}{2} } (0) ,
\end{equation}
where on the right hand side we have the analytically-continued Verma module of spin $j = a - \tfrac{1}{2}$. 

Evidently, this is reminiscent of the standard QQ relation \cite{Bazhanov:2010ts}
\begin{equation} 
	 2 i \op{sin} (\hbar \phi/2 )	\mbf{T}^+_{j - \frac{1}{2}} (z,\phi) {=} \mbf{Q}_{+}(z + \hbar j, \phi) \mbf{Q}_{-}(z  - \hbar j,\phi)
\end{equation}
except for the dependence on the twist parameter $\phi$.  We will explain shortly how the dependence on the parameter arises from placing 't Hooft lines at $\infty$, and also explain how the normalizations work out. 

\subsection{The QQ relation for a minuscule coweight}
In this section we will derive the QQ relation for the 't Hooft line of an arbitrary minuscule coweight of an arbitrary group.  Recall that the Q-operator, appropriately normalized, comes from 't Hooft operator placed at $0$ and at $\infty$.  The 't Hooft operator at $\infty$ is given by a surface operator which is the $\sigma$-model whose target $\mf{g}_{1}$ and which couples only to the components of the gauge field in $\mf{g}_{1}$.  

By the analysis above, colliding the two 't Hooft lines at $\infty$ gives us an analytically-continued line defect whose algebra of operators is the oscillator algebra with generators $\til{p}_i$, $\til{q}^i$ with commutation relations 
\begin{equation} 
	[\til{p}_i, \til{q}^i] = \hbar. 
\end{equation}
This line defect is only coupled to the components $A^i$ of the gauge
field corresponding to $\mf{g}_{1}$, via
\begin{equation} 
  \frac{1}{\hbar} \int_{z = \infty} z^2 \partial_z A^i \til{p}_i .
\end{equation}
We can compute the L-operator from a probe Wilson line at $z$ crossing this line defect at $\infty$, by computing the exchange of a single gluon.  The exchange of a single gluon is
\begin{equation} 
	\t^i \til{p}_i    \lim_{z \to \infty} z^2 \partial_z \frac{1}{z} = -\t^i \til{p}_i
\end{equation}
(where $\t^i$ acts in the representation associated to the probe Wilson line).  Note that this expression is independent of $\hbar$, because the coupling with the line defect at $\infty$ has a factor of $1/\hbar$, and the propagator has a factor of $\hbar$.  

This tells us that, to leading order in $\hbar$, we have \begin{equation} 
	L(z) =   \exp \left( - \t^i \til{p}_i \right).
\end{equation}
Note that the L-operator is independent of $z$.  

For example, for $\mf{g} = \mf{sl}_2$ with a Wilson line in the fundamental representation, we find
\begin{equation} 
	L(z) = \begin{pmatrix}
		1 & - \til{p} \\
		0 & 1
	\end{pmatrix} \label{eqn:triangular_oscillator}
\end{equation}
to leading order in $\hbar$.  

We will denote this line defect at $\infty$ by
$\mbf{T}^{\mu}_{\infty}$, where $\mu$ is the minuscule coweight.  In
the case of $\mf{sl}_2$, we will write it as $\mbf{T}^{+}_{\infty}$,
$\mbf{T}^-_{\infty}$ where $+$, $-$ are the minuscule coweights
$\pm \half$.

In \cite{Bazhanov:2010ts}, they show that this L-operator appears in
the QQ relation.  They find the composition of two Q-operators gives
rise to two oscillator line defects. One is the Verma module for
$\mf{sl}_2$, and one the oscillator representation
\eqref{eqn:triangular_oscillator}. In symbols, before taking the
trace, their relation is
\begin{equation} 
  \label{eqn:QQ}
  \mbf{T}^{-}_{\infty}	\mbf{\til{T}}^+_{j - \frac{1}{2}} (z) {=} \mbf{Q}_{+}(z + \hbar j, \phi) \mbf{Q}_{-}(z  - \hbar j,\phi).
\end{equation}
Here, $\mbf{\til{T}}$ is the unnormalized T-operator.

This is exactly what we find in our analysis.  The Q-operator, up to
normalization, is given by an 't Hooft line at $z$ and an 't Hooft
line at $\infty$.  As before, we will abuse notation and let
$\mbf{H}_{\pm}(z)$ be the 't Hooft line of charge $\pm \tfrac{1}{2}$
at $z$ and $\mp \tfrac{1}{2}$ at $\infty$. Our relation from colliding
't Hooft lines is
\begin{equation} 
  \mbf{T}^{-}_{\infty} 	\mbf{T}^+_{j - \frac{1}{2}} (z) {=} \mbf{H}_{+}(z + \hbar j) \mbf{H}_{-}(z  - \hbar j),
\end{equation}
which is identical to \eqref{eqn:QQ}.  In our prescription, the extra
oscillator representation comes from the collision of 't Hooft lines
at $\infty$.

The only difference between what we found and what is found in
\cite{Bazhanov:2010ts} is in the normalization. Recall that the
Q-operators are normalized with respect to the 't Hooft line by
\begin{equation} 
	\mbf{Q}_{\pm}(z) = G(z)^L \mbf{H}_{\pm}(z)  ,
\end{equation}
where $L$ is the number of sites on the spin chain and
\begin{equation} 
	G(z) =\frac{1}{(2 \hbar)^{1/2}}    \frac{ \Gamma\left(\frac{1}{2 \hbar} (z+\tfrac{\hbar}{2}) \right)} 
	{\Gamma\left(\frac{1}{2 \hbar} (z+\tfrac{3\hbar}{2}) \right)},
\end{equation}
and the T-operators are normalized so that
\begin{equation} 
	\mbf{\til{T}}_j(z) = F(z,j)^L \mbf{T}_j(z),
\end{equation}
where 
\begin{equation} 
	F(z,j)  
	= \frac{1}{2  \hbar} \frac{ \Gamma\left(\frac{1}{2 \hbar} (z+\hbar(j+1) \right) \Gamma\left(\frac{1}{2 \hbar} (z -  \hbar j )  \right)} 
	{\Gamma\left(\frac{1}{2 \hbar} (z+\hbar(j+2) \right)  \Gamma\left(\frac{1}{2 \hbar} (z-\hbar j+ \hbar) \right)} 
\end{equation}
is the normalizing function we used earlier. 

We have
\begin{equation} 
	G(z+\hbar j) G(z - \hbar j) = F(z,j-\tfrac{1}{2}) 
\end{equation}
so that all the normalizing factors work out to match our relation
between 't Hooft lines with the QQ relation \eqref{eqn:QQ} derived in
\cite{Bazhanov:2010ts}.

\subsection{The QQ relation after taking the trace}

The line defect $\mbf{T}^-_{\infty}$ at $\infty$ couples to a Wilson
line only by an upper-triangular matrix.  Therefore, at first sight,
it seems that it disappears after taking the trace.

This is not quite correct, however, because the line defect has
algebra of operators the oscillator algebra. A trace for the
oscillator algebra only exists once we introduce a twist parameter.
The twist parameter is obtained by modifying the boundary condition so
that our gauge field does not vanish at $z = \infty$, but instead
takes constant value $\phi$ (which we can take to be in the Cartan).
The background field $\phi$ couples non-trivially to the quantum
mechanical system at $z = \infty$.  The coupling is quadratic. If we
take a basis $\phi_r$ of the Cartan, and a basis $\til{p}_i$ of
$\mf{g}_{1}$ with weights $w^r_i$, then the coupling is
\begin{equation} 
	-\frac{1}{\hbar}\sum \phi_r w^r_i \tfrac{1}{2} \left(  \til{p}_i \til{q}^i + \til{q}^i \til{p}_i \right) .
\end{equation}
This expression is time-reversal invariant.

This coupling gives a non-trivial Hamiltonian to the topological
quantum mechanics. We can absorb the factor of $\hbar$ into the
normalization of $\til{p}_i$. If we take the $x$-direction, which the
defect wraps, to be of period $1$, then we can compute the partition
function of the system by taking the trace on the Fock module
$\C[\til{p}_1,\dots \til{p}_k]$.

Since the Hilbert space is a tensor product of the spaces
$\C[\til{p}_i]$, we see the partition function is a product.  Each
factor in the tensor product gives us
\begin{equation} 
  e^{\tfrac{1}{2} \phi_r w^r_i} + e^{\tfrac{3}{2} \phi_r w_r^i} + \dots = \frac{1  }{e^{-\tfrac{1}{2} \phi_r w_{r_i}} - e^{\tfrac{1}{2} \phi_r w^r_i } } = -\frac{1}{2 \sinh (\tfrac{1}{2}\phi_r w^r_i ) } .
\end{equation}
Therefore the trace of the line defect at infinity is
\begin{equation} 
	\prod_{i} \frac{-1}{ 2 \sinh (\tfrac{1}{2}\phi_r w^r_i )} = \prod_{\alpha, \ip{\mu,\alpha} < 0}  \frac{-1}{2 \sinh (\tfrac{1}{2} \phi_{\alpha} ) }.
\end{equation}
On the right hand side, we have written the expression as a product
over all roots $\alpha$ with $\ip{\mu,\alpha} < 0$.

Including this factor from the collision of 't Hooft lines at
$\infty$, we find the general form of the QQ relation, valid for any
minuscule coweight, is
\begin{equation} 
	\prod_{\alpha, \ip{\mu,\alpha} < 0}  (-2) \sinh (\tfrac{1}{2} \phi_{\alpha} ) \mbf{T}^{\mu}_{-\tfrac{\op{dim} \g_1}{2} + a   } (z) =  	\mbf{H}_{-\mu}\left(z-\hbar  \frac{a \sh^\vee}{ \op{dim} \g_1 } \right) \mbf{H}_{\mu}\left(z+\hbar \frac{a \sh^\vee}{\op{dim} \g_1} \right) . 
\end{equation}

We can compare this with the relation derived in
\cite{Bazhanov:2010ts} in the case of $\mf{sl}_2$, which is
\begin{equation} 
	2 i \op{sin} ( \phi/2 )	\mbf{\til{T}}^+_{j - \frac{1}{2}} (z) {=} \mbf{Q}_{+}(z + \hbar j) \mbf{Q}_{-}(z  - \hbar j).
\end{equation}
The relations are the same after rotating the basis of the Cartan by
$-\i$ and including the normalizing factors relating the Q-operators
with the 't Hooft lines.

This completes our derivation of the QQ relation from the collision of
't Hooft lines.  It is worth pointing out that the last step, of
taking the trace, requires choosing a contour for the path integral of
the 't Hooft line.  This corresponds to replacing the
analytically-continued line defects by an actual line defect, by
choosing a module for the algebra of operators.

Although we do not give the detailed analysis here, we note that the
method applies equally well in the trigonometric case. In the
trigonometric case, the 't Hooft line at $\infty$ can be placed at
either $z = 0$ or $z = \infty$, leading to extra possibilities.

In the elliptic case, things are a little more complicated as there is
no way to place an 't Hooft line at $\infty$. A single 't Hooft line
must be viewed as an interface between integrable models with
dynamical parameter corresponding to different components of the
moduli of $G$-bundles on an elliptic curve.

\subsection*{Acknowledgements}

We would like to thank Edward Witten and Masahito Yamazaki for helpful
discussions.  This research was supported in part by a grant from the
Krembil Foundation. K.C. and D.G. are supported by the NSERC Discovery
Grant program and by the Perimeter Institute for Theoretical
Physics. J.Y. is supported by the National Key Research and
Development Program of China (No.\ 2020YFA0713000).  Research at
Perimeter Institute is supported in part by the Government of Canada
through the Department of Innovation, Science and Economic Development
Canada and by the Province of Ontario through the Ministry of Colleges
and Universities. Any opinions, findings, and conclusions or
recommendations expressed in this material are those of the authors
and do not necessarily reflect the views of the funding agencies.

\appendix

\section{The prefactor in the L-operator}
\label{sec:appendix_prefactor}

Consider an arbitrary Wilson line associated to a representation of
$\mf{sl}_2$ in which the quadratic Casimir
$\tfrac{1}{4} h^2 + \tfrac{1}{2}( ef + fe) $ acts by a constant $C$,
which in a highest weight representation with highest weight vector of
spin $j$, is $j(j+1)$. (Our conventions are such that $h$ acts by
$2j$.)  Let $\rho(e)$, $\rho(f)$, $\rho(h)$ be the matrices defining
this representation, where as usual $[e,f] = h$, $[h,e] = 2 e$,
$[h,f] = 2 f$.  Then the L-operator is
\begin{equation} 
	L(z) = F(z,j) \begin{pmatrix}
		z + \hbar \tfrac{1}{2} \rho(h) &\hbar  \rho(f) \\
		\hbar \rho(e) & z - \hbar \tfrac{1}{2} \rho(h)
	\end{pmatrix} .
\end{equation}
The quantum determinant relation states that 
\begin{equation} 
	L_{11}(z) L_{22}(z + \hbar) - L_{12}(z) L_{21}(z + \hbar) = 1. 
\end{equation}
(It can be tricky to track down the conventions which determine the
precise shift in $z$. However, there is only one possible shift for
which a term in the quantum determinant proportional to $h$ drops out,
and it is proportional to the identity operator.)

This becomes
\begin{equation} 
	(z + \hbar \tfrac{1}{2} \rho(h) )(z + \hbar -\hbar \tfrac{1}{2}\rho(h) ) -   \hbar^2 \rho(e) \rho(f)  = F(z,j)^{-1} F(z+\hbar,j)^{-1} .
\end{equation}
Expanding the left hand side we find
\begin{multline} 
	z^2 + \hbar z - \hbar^2 \tfrac{1}{4} \rho(h)^2 + \hbar^2 \tfrac{1}{2} \rho(h) - \hbar^2 \rho(e) \rho(f) \\
	=z^2 + \hbar z - \hbar^2 \tfrac{1}{4} \rho(h)^2 + \hbar^2 \tfrac{1}{2}  \rho(h) - \hbar^2 \tfrac{1}{2} \left( \rho(e) \rho(f) + \rho(f) \rho(e)  -  [\rho(f), \rho(e)] \right). 
\end{multline}
The terms $\rho(h)$ and $[\rho(f), \rho(e)]$ cancel, and $\tfrac{1}{4} \rho(h)^2 + \tfrac{1}{2}(\rho(e)\rho(f) + \rho(f) \rho(e))$ is $C= j (j+1)$.  Therefore the equation becomes
\begin{equation} 
	z^2 + \hbar z  -  \hbar^2 j(j+1) = F(z,j)^{-1} F(z+\hbar,j)^{-1}. 
\end{equation}
Note 
\begin{equation} 
	z^2 + \hbar z - \hbar^2 j(j+1) = (z + \hbar (j+1) ) (z - \hbar j). 
\end{equation}
A solution to this equation is given by 
\begin{equation} 
	F(z,j)  
	= \frac{1}{2  \hbar} \frac{ \Gamma\left(\frac{1}{2 \hbar} (z+\hbar(j+1) \right) \Gamma\left(\frac{1}{2 \hbar} (z -  \hbar j )  \right)} 
	{\Gamma\left(\frac{1}{2 \hbar} (z+\hbar(j+2) \right)  \Gamma\left(\frac{1}{2 \hbar} (z-\hbar j+ \hbar) \right)}.
\end{equation}

Note that factorize $F(z,j)$ is a product of two factors, one of which only depends on $z_+ = z+\hbar (j+1/2)$ the other on $z_- = z - \hbar (j+1/2)$:
\begin{equation} 
	F(z,j) = G(z_+) G(z_-)  ,
\end{equation}
where
\begin{equation} 
	G(z) =\frac{1}{(2 \hbar)^{1/2}}    \frac{ \Gamma\left(\frac{1}{2 \hbar} (z+\tfrac{\hbar}{2}) \right)} 
	{\Gamma\left(\frac{1}{2 \hbar} (z+\tfrac{3\hbar}{2}) \right)} .
\end{equation}
This factorization corresponds to the factorization of the T-operator
into a product of two 't Hooft lines.

\section{Quantization of analytically-continued Wilson lines}\label{appendix:parabolicverma}
Here we will prove the following.
\begin{theorem} 
	The analytically-continued Wilson line given by topological quantum mechanics on $T^\ast G/P$ exists at the quantum level. Furthermore, this line defect has only two deformations, one given by deforming  the topological quantum mechanics by a Wess-Zumino term associated to the class in $H^2(G/P) = \C$, and one given by the position in the $z$-plane. 
\end{theorem}

\begin{proof}
  Let us be a little more precise about the statement.  The Lagrangian
  for the quantum-mechanical system is
  \begin{equation}
    \frac{1}{\hbar} \int p \d q. 
  \end{equation}
  We will first fix some quantization of the quantum mechanical
  system, compatible with the $\C^\times$ action which scales $p$ and
  $\hbar$. Algebraically, this amounts to deforming the Poisson
  algebra of functions $\Oo(T^\ast G/P)$ into a non-commutative
  algebra equipped with a filtration whose associated graded is the
  original Poisson algebra. It is known \cite{MR2119140} that there is
  a bijection between such quantizations and classes in $H^2(G/P)$. In
  terms of the Lagrangian, this term is the Wess-Zumino term
  associated to a class in $H^2(G/P)$.

  Let us fix one such quantization, with quantum algebra $B$.  This is
  the algebra of local operators on the line defect. We now analyze
  whether we can couple this to four-dimensional Chern-Simons
  theory. Algebraically, this amounts to finding a homomorphism from
  the Yangian $Y(\g)$ to $B$, extending the known homomorphism
  $U \g \to B$.

  When we treat the gauge field classically, we can couple the line
  defect with algebra of operators $B$. The algebra of operators of
  the coupled system is $C^\ast(\g[[z]], B)$, the Lie algebra cochain
  complex of $\g[[z]]$ with coefficients in $B$.

  Anomalies to the line defect existing at the quantum level are given
  by $H^2$ with coefficients in the algebra of local operators of the
  couple along the line.  Deformations of the line defect are given by
  $H^1$.

  Consider any line defect with algebra of operators $B$.  In
  \cite{Costello:2017dso}, appendix C, it was proved that the
  cohomology group $H^2(\g[[z]], B)$ vanishes unless there is a
  $G$-invariant map $\wedge^2 \g \to B$ which does not factor through
  the adjoint representation.  Further, the group $H^1(\g[[z]], B)$
  has dimension the number of copies of the adjoint representation in
  $B$.  Each copy of the adjoint representation gives a deformation of
  the line defect where we couple $\partial_z A$ to that copy of the
  adjoint.

  If $B$ is any of the quantum algebras deforming functions on
  $T^\ast G/P$, then, as a representation of $G$, it is isomorphic to
  functions on $T^\ast G/P$. We need to check that:
  \begin{enumerate} 
  \item There is only one copy of the adjoint representation in $\Oo(T^\ast G/P)$. 
  \item Every map of $G$-representations $\wedge^2 \g \to \Oo(T^\ast G/P)$ factors through the adjoint representation.
  \end{enumerate}

  The case $G = SL_2(\C)$, $G/P = \CP^1$ is rather special. Here,
  since $\wedge^2 \mf{sl}_2(\C) = \mf{sl}_2(\C)$, the second condition
  is vacuous, and it is rather easy to see that there is only one copy
  of the adjoint representation in $\op{Diff}(\CP^1,K^{1/2})$. We will
  assume that $G \neq SL_2(\C)$ in what follows

  As a $G$-representation, the space $\Oo(T^\ast G/P)$ is isomorphic
  to the $P$-invariants in $\Oo(G) \times \Sym^\ast \mf{g}_1$, where
  as above we decompose
  $\mf{g} = \mf{g}_{-1} \oplus \mf{g}_0 \oplus \mf{g}_1$ and the Lie
  algebra of $P$ is $\mf{g}_0 \oplus \mf{g}_1$. Also $P$ acts on the
  right on $\Oo(G)$, and $G$ acts on the left on $\Oo(G)$ and
  $\Sym^\ast \mf{g}_1$ has only a $P$ action, not a $G$-action.

  By the Peter-Weyl theorem, $\Oo(G) = \oplus V_L \otimes V^\ast_R$
  where the direct sum is over all irreducible representations $V$ of
  $G$, and $V_L$ has the left action of $G$, $V^\ast_R$ is the dual
  representation with the right action of $G$.

  Therefore, functions on $T^\ast G/P$ can be written as
  \begin{equation} 
    \oplus_{V} V \otimes \left( V^\ast \otimes \Sym^\ast \g_1 \right)^P ,
  \end{equation}
  where we are summing over all irreducible representations of $G$.
  Each irreducible representation $V^\ast$ of $G$ is restricted to a
  $P$-representation, and each $\Sym^k \g_1$ is viewed as a $P$
  representation because $\g_1 \subset \g$ is a sub-$P$ representation
  of $\g_1$. (Recall that $P = \exp(\g_0 \oplus \g_1)$).

  The number of copies of the adjoint representation appearing in
  $\Oo(T^\ast G/P)$ is
  \begin{equation} 
    \op{dim}  ( \g \otimes \Sym^\ast \g_1)^P. 
  \end{equation}
  Note that
  \begin{equation} 
    \mf{g}_0 = \mf{l} \oplus \C \cdot \rho ,
  \end{equation}
  where $\mf{l}$ is a semi-simple Lie algebra with no Abelian factors,
  $\rho$ is the coweight defining $P$. The eigenvalues of $\rho$ on
  $\g_1,\g_0,\g_{-1}$ are $1,0,-1$. We can first pass to
  $\rho$-invariants of $\g \otimes \Sym^\ast \g_1$, which are
 \begin{equation} 
   \g_{-1} \otimes \Sym^1 \g_1 \oplus \mf{l}\otimes \Sym^0 \g_1 \oplus \C \cdot \rho\otimes \Sym^0 \g_1.
  \end{equation}
  Next, let us pass to $\mf{l}$-invariants. In every case, $\g_{-1}$
  is an irreducible representation of $\mf{l}$.  Then the space of
  $\mf{l}$-invariants is two dimensional, given by the invariants in
  $\g_{-1} \otimes \g_1$ and by the third factor $\C \cdot \rho$.
  But, $\rho$ is not invariant under $\g_1 \subset \mf{p}$, so that
  there is only at most a one-dimensional subspace of $P$-invariants,
  and hence at most one copy of the adjoint representation.  Since
  there is evidently at least one copy of the adjoint representation
  appearing in functions on $T^\ast G/P$, this proves there is exactly
  one copy, as desired.

  Next, we need to check that the space of $P$-invariants in
  $\wedge^2 \g \otimes \Sym^\ast \g_1$ is one-dimensional (spanned by
  the $P$-invariants in the adjoint representation).  This is required
  to show that there are no obstructions to quantizing the line
  defect, i.e.\ to show that the algebra of twisted differential
  operators on $G/P$ has a homomorphism from the Yangian.

  The $\rho$-invariants in $\wedge^2 \g \otimes \Sym^\ast \g_1$ are
  \begin{multline} 
    \wedge^2 \g_0 \otimes \Sym^0 \g_1 \oplus (\g_1 \otimes \g_{-1}) \otimes \Sym^0 \g_1 \\
    \oplus (\g_{-1} \otimes \g_0)\otimes \Sym^1 \g_1 \oplus (\wedge^2 \g_{-1}) \otimes \Sym^2 \g_1. 
  \end{multline}
  Let us now pass to $\mf{l}$-invariants. There are no $\mf{l}$-invariants in 
  \begin{equation} 
    \wedge^2 \g_0 = \wedge^2 (\mf{l} \oplus \C \cdot \rho) = \mf{l} \oplus \wedge^2 \mf{l}
  \end{equation}
  because $\mf{l}$ is semi-simple with no Abelian factors.  The possible $\mf{l}$ invariants live in
  \begin{align}	
    \g_{-1} \otimes \mf{l} \otimes \Sym^1 \g_1, \label{sym1}\\
    \g_{-1} \otimes \C \cdot \rho \otimes \Sym^1 \g_1, \label{sym1other}\\
    \g_{-1} \otimes \g_1 \otimes \Sym^0 \g_1, \label{sym0} \\
    \wedge^2 \g_{-1} \otimes \Sym^2 \g_1. \label{sym2}
  \end{align}	
  The first three lines manifestly contain at least one
  $\mf{l}$-invariant element.  The last line may or may not.

  Next, let us impose the constraint of invariance under $\g_1$,
  recalling that we are interested in invariants for
  $\mf{p} =\g_0 \oplus \g_1$.  Since $\g_1$ acts trivially on each
  symmetric power $\Sym^k \g_1$, we find that $\g_1$ invariance
  imposes independent constraints for each power $\Sym^k \g_1$ that
  appears.

  We will treat the $\Sym^0 \g_1$ factor \eqref{sym0} first. Using the
  fact that the component of the Lie bracket
  $\g_1 \otimes \g_{-1} \to \g_0$ landing in $\C \cdot \rho$ provides
  a non-degenerate pairing between $\g_1$ and $\g_{-1}$, we see that
  there are no $\g_1$ invariants in $\g_{-1} \otimes \g_1$.
	
  For the $\Sym^2 \g_1$ factor \eqref{sym2}, the same argument implies
  that there are no $\g_1$ invariants in $\wedge^2 \g_{-1}$.

  Only the $\Sym^1 \g_1$ factors \eqref{sym1}, \eqref{sym1other}
  remain.  Let us make the assumption that the adjoint representation
  of $\mf{l}$ appears only once in $\g_{-1} \otimes \g_1$. We will
  check this on a case by case basis shortly.

  Under this assumption, there are two $\mf{l}$ invariant tensors in
  $\g_{-1} \otimes \g_0 \otimes \Sym^1 \g_{1}$. One is in
  $\g_{-1} \otimes \C \cdot \rho \otimes \Sym^1 \g_1$ and one is in
  $\g_{-1} \otimes \mf{l} \otimes \Sym^1\g_1$. This first of these two
  tensor is not $\g_1$ invariant, so that we are left with at most one
  $\mf{p}$ invariant element in $\wedge^2 \g \otimes \Sym^\ast \g_1$,
  as desired.

  It remains to check that there is only one copy of the adjoint
  representation of $\mf{l}$ in $\g_1 \otimes \g_{-1}$.

  There are minuscule coweights for the Lie algebras of all classical
  types and $E_6,E_7$.  In type $A$, the minuscule coweights of
  $\mf{sl}_n(\C)$ are given by Lie algebra elements
  $\rho = \op{Diag}(1^k ,0^{n-k})$ for $1 \le k \le (n-1)/2$.  The
  simple Lie algebra is $\mf{l} = \mf{sl}_k \oplus \mf{sl}_{n-k}$.  If
  we decompose $\C^{n} = \C^k \oplus \C^{n-k}$, then the
  representation $\g_1$, $\g_{-1}$ of $\mf{l}$ are
  $\g_1 = \Hom(\C^k, \C^{n-k})$ and $\g_{-1}=\Hom(\C^{n-k}, \C^k)$.
  Thus, $\g_{\pm 1}$ are the bifundamental representations of
  $\mf{sl}_{k} \oplus \mf{sl}_{n-k}$.

  Manifestly, there is only one copy of $\mf{l}$ in the tensor product
  $\g_1 \otimes \g_{-1}$ in this case.

  For $SO(2n,\C)$, the spinorial minuscule coweight is defined as
  follows. (There are two Weyl orbits of coweights like this related
  by an outer automorphism of $SO(2n,\C)$; the argument is the same in
  each case.)  We decompose $\C^{2n} = \C^n_+ \oplus \C^n_-$ where
  $\C^n_{\pm}$ are null.  There is an element $\rho \in \mf{so}(2n)$
  where $\rho$ acts on $\C^n_+$ by $1/2$ and on $\C^n_-$ by $-1/2$.
  The subgroup $\mf{l}$ is $\mf{sl}_n$, where $\C^n_{\pm}$ are the
  fundamental and antifundamental representations of $\mf{l}$.  The
  subspaces $\mf{g}_{\pm 1}$ of $\mf{so}(2n,\C)$ are
  $\wedge^2 \C^n_{\pm}$, the exterior square of the fundamental and
  antifundamental representations of $\mf{sl}_n(\C)$.

  Again, there is only one copy of the adjoint of $\mf{sl}_n(\C)$ in
  the tensor product $\g_1 \otimes \g_{-1}$.

  For $SO(n+2,\C)$ there is another minuscule coweight obtained from
  the orthogonal decomposition $\C^{n+2} = \C^2 \oplus \C^n$.  The
  coweight $\rho$ is the generator of the $SO(2,\C)$ rotating
  $\C^2$. The simple algebra $\mf{l}$ is $\mf{so}(n,\C)$.  The spaces
  $\mf{g}_{\pm 1}$ are both the vector representation $\C^n$ of
  $\mf{so}(n,\C)$. Again, there is only one copy of $\mf{so}(n,\C)$ in
  $\g_1 \otimes \g_{-1}$.

  For $Sp(2n,\C)$ there is a minuscule coweight coming from the
  decomposition $\C^{2n} = \C^n_+ \oplus \C^n_-$ where $\C^n_{\pm}$
  are Lagrangians, and $\rho$ acts on $\C^n_{\pm}$ with weights
  $\pm 1/2$.  The Lie algebra $\mf{l}$ is $\mf{sl}_n$, and $\g_{\pm}$
  are $\Sym^2 \C^n_{\pm}$. These are the symmetric squares of the
  fundamental and antifundamental representations.  Again, it is easy
  to see that there is only one $\mf{sl}_n$ invariant element in
  $\g_1 \otimes \g_{-1}$.

  The remaining simple algebras with minuscule coweights are the
  exceptional algebras $\mf{e}_6$ and $\mf{e}_7$.  The algebra
  $\mf{e}_6$ has a minuscule coweight where $\mf{l} = \mf{so}(10)$,
  and $\g_{\pm 1}$ are the spin representations $S_{\pm}$.  If $V$
  denotes the vector representation of $\mf{so}(10)$, then
  \begin{equation} 
    S_+\otimes S_- = \wedge^0 V \oplus \wedge^2 V \oplus \wedge^4 V 
  \end{equation}
  so that the adjoint representation of $\mf{so}(10)$ appears exactly once.

  For $\mf{e}_7$, there is a minuscule coweight where
  $\mf{l} = \mf{e}_6$ and where $\g_{\pm}$ are the two irreducible
  $27$ dimensional representations of $\mf{e}_6$, which we denote
  $\mbf{27}$ and $\br{\mbf{27}}$.  One can check that there is only
  one copy of the adjoint in $\mbf{27} \otimes \br{\mbf{27}}$.
\end{proof}

\bibliographystyle{JHEP}
\bibliography{q}

\providecommand{\href}[2]{#2}\begingroup\raggedright\begin{thebibliography}{10}

\bibitem{Costello:2013zra}
K.~Costello, \emph{Supersymmetric gauge theory and the {Y}angian},
  \href{https://arxiv.org/abs/1303.2632}{{\ttfamily 1303.2632}}.

\bibitem{Costello:2017dso}
K.~Costello, E.~Witten and M.~Yamazaki, \emph{Gauge theory and integrability,
  {I}}, \href{https://doi.org/10.4310/ICCM.2018.v6.n1.a6}{\emph{ICCM Not.}
  {\bfseries 6} (2018) 46} [\href{https://arxiv.org/abs/1709.09993}{{\ttfamily
  1709.09993}}].

\bibitem{Costello:2018gyb}
K.~Costello, E.~Witten and M.~Yamazaki, \emph{Gauge theory and integrability,
  {II}}, \href{https://doi.org/10.4310/ICCM.2018.v6.n1.a7}{\emph{ICCM Not.}
  {\bfseries 6} (2018) 120} [\href{https://arxiv.org/abs/1802.01579}{{\ttfamily
  1802.01579}}].

\bibitem{Costello:2019tri}
K.~Costello and M.~Yamazaki, \emph{{Gauge Theory And Integrability, III}},
  \href{https://arxiv.org/abs/1908.02289}{{\ttfamily 1908.02289}}.

\bibitem{Vicedo:2019dej}
B.~Vicedo, \emph{Holomorphic {C}hern-{S}imons theory and affine {G}audin
  models},  \href{https://arxiv.org/abs/1908.07511}{{\ttfamily 1908.07511}}.

\bibitem{Delduc:2019whp}
F.~Delduc, S.~Lacroix, M.~Magro and B.~Vicedo, \emph{A unifying 2{D} action for
  integrable {$\sigma$}-models from 4{D} {C}hern-{S}imons theory},
  \href{https://doi.org/10.1007/s11005-020-01268-y}{\emph{Lett. Math. Phys.}
  {\bfseries 110} (2020) 1645}
  [\href{https://arxiv.org/abs/1909.13824}{{\ttfamily 1909.13824}}].

\bibitem{Fukushima:2020dcp}
O.~Fukushima, J.-i.~Sakamoto and K.~Yoshida, \emph{Yang-{B}axter deformations
  of the {$\rm AdS_5\times S^5$} supercoset sigma model from 4{D}
  {C}hern-{S}imons theory},
  \href{https://doi.org/10.1007/jhep09(2020)100}{\emph{JHEP} {\bfseries 09}
  (2020) Art. 100} [\href{https://arxiv.org/abs/2005.04950}{{\ttfamily
  2005.04950}}].

\bibitem{Bykov:2020nal}
D.~Bykov, \emph{{Quantum flag manifold $\sigma$-models and Hermitian Ricci
  flow}},  \href{https://arxiv.org/abs/2006.14124}{{\ttfamily 2006.14124}}.

\bibitem{Bittleston:2019gkq}
R.~Bittleston and D.~Skinner, \emph{Gauge theory and boundary integrability},
  \href{https://doi.org/10.1007/jhep05(2019)195}{\emph{JHEP} {\bfseries 05}
  (2019) 195, 52} [\href{https://arxiv.org/abs/1903.03601}{{\ttfamily
  1903.03601}}].

\bibitem{Bittleston:2019mol}
R.~Bittleston and D.~Skinner, \emph{Gauge theory and boundary integrability.
  {P}art {II}. {E}lliptic and trigonometric cases},
  \href{https://doi.org/10.1007/jhep06(2020)080}{\emph{JHEP} {\bfseries 06}
  (2020) 080, 34} [\href{https://arxiv.org/abs/1912.13441}{{\ttfamily
  1912.13441}}].

\bibitem{Costello:2020lpi}
K.~Costello and B.~Stefa\'{n}ski, Jr., \emph{Chern-{S}imons origin of
  superstring integrability},
  \href{https://doi.org/10.1103/physrevlett.125.121602}{\emph{Phys. Rev. Lett.}
  {\bfseries 125} (2020) 121602, 6}
  [\href{https://arxiv.org/abs/2005.03064}{{\ttfamily 2005.03064}}].

\bibitem{Baxter:1972hz}
R.J.~Baxter, \emph{Partition function of the eight-vertex lattice model},
  \href{https://doi.org/10.1016/0003-4916(72)90335-1}{\emph{Ann. Phys.}
  {\bfseries 70} (1972) 193}.

\bibitem{Bazhanov:2010ts}
V.V.~Bazhanov, T.~Lukowski, C.~Meneghelli and M.~Staudacher, \emph{A shortcut
  to the {Q}-operator},
  \href{https://doi.org/10.1088/1742-5468/2010/11/P11002}{\emph{J. Stat. Mech.}
  {\bfseries 1011} (2010) P11002}
  [\href{https://arxiv.org/abs/1005.3261}{{\ttfamily 1005.3261}}].

\bibitem{Bazhanov:2010jq}
V.V.~Bazhanov, R.~Frassek, T.~\L~ukowski, C.~Meneghelli and M.~Staudacher,
  \emph{Baxter {$\bf Q$}-operators and representations of {Y}angians},
  \href{https://doi.org/10.1016/j.nuclphysb.2011.04.006}{\emph{Nuclear Phys. B}
  {\bfseries 850} (2011) 148}
  [\href{https://arxiv.org/abs/1010.3699}{{\ttfamily 1010.3699}}].

\bibitem{Frassek:2020nki}
R.~Frassek, \emph{Oscillator realisations associated to the {D}-type {Y}angian:
  towards the operatorial {Q}-system of orthogonal spin chains},
  \href{https://doi.org/10.1016/j.nuclphysb.2020.115063}{\emph{Nuclear Phys. B}
  {\bfseries 956} (2020) 115063, 22}
  [\href{https://arxiv.org/abs/2001.06825}{{\ttfamily 2001.06825}}].

\bibitem{Witten:1979ey}
E.~Witten, \emph{Dyons of charge $e\theta/2\pi$},
  \href{https://doi.org/10.1016/0370-2693(79)90838-4}{\emph{Phys. Lett. B}
  {\bfseries 86} (1979) 283}.

\bibitem{Brundan:2004ca}
J.~Brundan and A.~Kleshchev, \emph{Shifted {Y}angians and finite
  {$W$}-algebras}, \href{https://doi.org/10.1016/j.aim.2004.11.004}{\emph{Adv.
  Math.} {\bfseries 200} (2006) 136}
  [\href{https://arxiv.org/abs/math/0407012}{{\ttfamily math/0407012}}].

\bibitem{MR3248988}
J.~Kamnitzer, B.~Webster, A.~Weekes and O.~Yacobi, \emph{Yangians and
  quantizations of slices in the affine {G}rassmannian},
  \href{https://doi.org/10.2140/ant.2014.8.857}{\emph{Algebra Number Theory}
  {\bfseries 8} (2014) 857}.

\bibitem{Frassek:2020lky}
R.~Frassek, V.~Pestun and A.~Tsymbaliuk, \emph{{Lax matrices from
  antidominantly shifted Yangians and quantum affine algebras}},
  \href{https://arxiv.org/abs/2001.04929}{{\ttfamily 2001.04929}}.

\bibitem{Braverman:2016pwk}
A.~Braverman, M.~Finkelberg and H.~Nakajima, \emph{Coulomb branches of {$3d$}
  $\mathcal{N}=4$ quiver gauge theories and slices in the affine
  {G}rassmannian},
  \href{https://doi.org/10.4310/ATMP.2019.v23.n1.a3}{\emph{Adv. Theor. Math.
  Phys.} {\bfseries 23} (2019) 75}
  [\href{https://arxiv.org/abs/1604.03625}{{\ttfamily 1604.03625}}].

\bibitem{Costello:2018txb}
K.~Costello and J.~Yagi, \emph{Unification of integrability in supersymmetric
  gauge theories},  \href{https://arxiv.org/abs/1810.01970}{{\ttfamily
  1810.01970}}.

\bibitem{Witten:2010zr}
E.~Witten, \emph{A new look at the path integral of quantum mechanics},  in
  \emph{Surveys in differential geometry. {V}olume {XV}. {P}erspectives in
  mathematics and physics}, vol.~15 of \emph{Surv. Differ. Geom.}, p.~345, Int.
  Press, Somerville, MA (2011),
  \href{https://doi.org/10.4310/SDG.2010.v15.n1.a11}{DOI}
  [\href{https://arxiv.org/abs/1009.6032}{{\ttfamily 1009.6032}}].

\bibitem{Kapustin:2005py}
A.~Kapustin, \emph{Wilson--'t {H}ooft operators in four-dimensional gauge
  theories and {$S$}-duality},
  \href{https://doi.org/10.1103/PhysRevD.74.025005}{\emph{Phys. Rev. D (3)}
  {\bfseries 74} (2006) 025005}.

\bibitem{Kapustin:2006hi}
A.~Kapustin, \emph{{Holomorphic reduction of $\mathcal{N} = 2$ gauge theories,
  Wilson--'t Hooft operators, and S-duality}},
  \href{https://arxiv.org/abs/hep-th/0612119}{{\ttfamily hep-th/0612119}}.

\bibitem{Kapustin:2006pk}
A.~Kapustin and E.~Witten, \emph{Electric-magnetic duality and the geometric
  {L}anglands program},
  \href{https://doi.org/10.4310/CNTP.2007.v1.n1.a1}{\emph{Commun. Number Theory
  Phys.} {\bfseries 1} (2007) 1}
  [\href{https://arxiv.org/abs/hep-th/0604151}{{\ttfamily hep-th/0604151}}].

\bibitem{MR3752460}
X.~Zhu, \emph{An introduction to affine {G}rassmannians and the geometric
  {S}atake equivalence},  in \emph{Geometry of moduli spaces and representation
  theory}, vol.~24 of \emph{IAS/Park City Math. Ser.}, pp.~59--154, Amer. Math.
  Soc., Providence, RI (2017).

\bibitem{Gomis:2011pf}
J.~Gomis, T.~Okuda and V.~Pestun, \emph{Exact results for 't {H}ooft loops in
  gauge theories on {$S^4$}},
  \href{https://doi.org/10.1007/JHEP05(2012)141}{\emph{JHEP} {\bfseries 05}
  (2012) 141} [\href{https://arxiv.org/abs/1105.2568}{{\ttfamily 1105.2568}}].

\bibitem{Ashwinkumar:2018tmm}
M.~Ashwinkumar, M.-C.~Tan and Q.~Zhao, \emph{Branes and categorifying
  integrable lattice models},
  \href{https://doi.org/10.4310/atmp.2020.v24.n1.a1}{\emph{Adv. Theor. Math.
  Phys.} {\bfseries 24} (2020) 1}
  [\href{https://arxiv.org/abs/1806.02821}{{\ttfamily 1806.02821}}].

\bibitem{Nekrasov:String-Math-2017}
N.~Nekrasov, ``Open-closed (little) string duality and
  {C}hern-{S}imons-{B}ethe/gauge correspondence.'' Talk at String Math 2017,
  July 24--28, 2017.

\bibitem{Kamnitzer2020}
J.~Kamnitzer, K.~Pham and A.~Weekes, \emph{Hamiltonian reduction for affine
  grassmannian slices and truncated shifted yangians},
  \href{https://arxiv.org/abs/2009.11791}{{\ttfamily 2009.11791}}.

\bibitem{MR1795753}
B.H.~Gross, \emph{On minuscule representations and the principal {${\rm
  SL}_2$}},
  \href{https://doi.org/10.1090/S1088-4165-00-00106-0}{\emph{Represent. Theory}
  {\bfseries 4} (2000) 225}.

\bibitem{Ferrando:2020vzk}
G.~Ferrando, R.~Frassek and V.~Kazakov, \emph{{QQ-system and Weyl-type transfer
  matrices in integrable SO(2r) spin chains}},
  \href{https://arxiv.org/abs/2008.04336}{{\ttfamily 2008.04336}}.

\bibitem{Zhang:2018sej}
H.~Zhang, \emph{Yangians and {B}axter's relations},
  \href{https://doi.org/10.1007/s11005-020-01285-x}{\emph{Lett. Math. Phys.}
  {\bfseries 110} (2020) 2113}
  [\href{https://arxiv.org/abs/1808.02294}{{\ttfamily 1808.02294}}].

\bibitem{MR3761996}
M.~Finkelberg, J.~Kamnitzer, K.~Pham, L.~Rybnikov and A.~Weekes,
  \emph{Comultiplication for shifted {Y}angians and quantum open {T}oda
  lattice}, \href{https://doi.org/10.1016/j.aim.2017.06.018}{\emph{Adv. Math.}
  {\bfseries 327} (2018) 349}
  [\href{https://arxiv.org/abs/1608.03331}{{\ttfamily 1608.03331}}].

\bibitem{Braverman:2016wma}
A.~Braverman, M.~Finkelberg and H.~Nakajima, \emph{Towards a mathematical
  definition of {C}oulomb branches of 3-dimensional $\mathcal{N}=4$ gauge
  theories, {II}},
  \href{https://doi.org/10.4310/ATMP.2018.v22.n5.a1}{\emph{Adv. Theor. Math.
  Phys.} {\bfseries 22} (2018) 1071}
  [\href{https://arxiv.org/abs/1601.03586}{{\ttfamily 1601.03586}}].

\bibitem{Kamnitzer}
J.~Kamnitzer, M.~McBreen and N.~Proudfoot, \emph{The quantum hikita
  conjecture},  \href{https://arxiv.org/abs/1807.09858}{{\ttfamily
  1807.09858}}.

\bibitem{Assel:2015oxa}
B.~Assel and J.~Gomis, \emph{Mirror symmetry and loop operators},
  \href{https://doi.org/10.1007/JHEP11(2015)055}{\emph{JHEP} {\bfseries 11}
  (2015) 055} [\href{https://arxiv.org/abs/1506.01718}{{\ttfamily
  1506.01718}}].

\bibitem{Dimofte:2019zzj}
T.~Dimofte, N.~Garner, M.~Geracie and J.~Hilburn, \emph{Mirror symmetry and
  line operators}, \href{https://doi.org/10.1007/jhep02(2020)075}{\emph{JHEP}
  {\bfseries 02} (2020) 075, 147}
  [\href{https://arxiv.org/abs/1908.00013}{{\ttfamily 1908.00013}}].

\bibitem{Gaiotto:2012xa}
D.~Gaiotto, L.~Rastelli and S.S.~Razamat, \emph{Bootstrapping the
  superconformal index with surface defects},
  \href{https://doi.org/10.1007/JHEP01(2013)022}{\emph{JHEP} {\bfseries 01}
  (2013) 022} [\href{https://arxiv.org/abs/1207.3577}{{\ttfamily 1207.3577}}].

\bibitem{Gukov:2006jk}
S.~Gukov and E.~Witten, \emph{Gauge theory, ramification, and the geometric
  {L}anglands program},  in \emph{Current developments in mathematics, 2006},
  p.~35, Int. Press, Somerville, MA (2008)
  [\href{https://arxiv.org/abs/hep-th/0612073}{{\ttfamily hep-th/0612073}}].

\bibitem{Cattaneo:1999fm}
A.S.~Cattaneo and G.~Felder, \emph{A path integral approach to the {K}ontsevich
  quantization formula},
  \href{https://doi.org/10.1007/s002200000229}{\emph{Comm. Math. Phys.}
  {\bfseries 212} (2000) 591}
  [\href{https://arxiv.org/abs/math/9902090}{{\ttfamily math/9902090}}].

\bibitem{MR1212625}
L.W.~Tu, \emph{Semistable bundles over an elliptic curve},
  \href{https://doi.org/10.1006/aima.1993.1011}{\emph{Adv. Math.} {\bfseries
  98} (1993) 1}.

\bibitem{Costello:2016vjw}
K.~Costello and O.~Gwilliam, \emph{Factorization algebras in quantum field
  theory. {V}ol. 1}, vol.~31 of \emph{New Mathematical Monographs}, Cambridge
  University Press, Cambridge (2017),
  \href{https://doi.org/10.1017/9781316678626}{10.1017/9781316678626}.

\bibitem{MR2415401}
S.~Fishel, I.~Grojnowski and C.~Teleman, \emph{The strong {M}acdonald
  conjecture and {H}odge theory on the loop {G}rassmannian},
  \href{https://doi.org/10.4007/annals.2008.168.175}{\emph{Ann. of Math. (2)}
  {\bfseries 168} (2008) 175}
  [\href{https://arxiv.org/abs/math/0411355}{{\ttfamily math/0411355}}].

\bibitem{Costello:2020jbh}
K.~Costello and N.M.~Paquette, \emph{Twisted supergravity and {K}oszul duality:
  a case study in {AdS}$_3$},
  \href{https://arxiv.org/abs/2001.02177}{{\ttfamily 2001.02177}}.

\bibitem{MR2119140}
R.~Bezrukavnikov and D.~Kaledin, \emph{Fedosov quantization in algebraic
  context},
  \href{https://doi.org/10.17323/1609-4514-2004-4-3-559-592}{\emph{Mosc. Math.
  J.} {\bfseries 4} (2004) 559}
  [\href{https://arxiv.org/abs/math/0309290}{{\ttfamily math/0309290}}].

\end{thebibliography}\endgroup

\end{document}